\documentclass[11pt]{article}
\usepackage{amsmath}
\usepackage{amssymb}
\usepackage{amsfonts}
\usepackage{amsthm}
\usepackage{graphicx}
\usepackage{color}
\usepackage{subfig}
\usepackage{hyperref}
\usepackage{cite}
\usepackage{comment}
\usepackage{transparent}
\usepackage{wrapfig}
\usepackage[nofillcomment,noend,linesnumbered]{algorithm2e}
\usepackage[top=1.0in, bottom=1.0in, left=1.0in, right=1.0in]{geometry}

\usepackage{colortbl}
\usepackage{tikz}

\newtheorem{theorem}{Theorem}[section]

\newtheorem{lemma}[theorem]{Lemma}

\newtheorem{observation}{Observation}
\newtheorem{definition}[theorem]{Definition}

\newtheorem{claim}[theorem]{Claim}

\newif\ifabstract
\newif\iffull
\ifabstract
	\fullfalse
\else
	\fulltrue
\fi

\def\<#1>{%
    \expandafter\ifx\csname<#1>\endcsname\relax
        \errmessage{abbreviation <#1> undefined!}
    \else
        \csname<#1>\endcsname
    \fi
}
\def\abbr#1#2{%
    \expandafter\def\csname<#1>\endcsname{#2}%
}

\pdfpageattr {/Group << /S /Transparency /I true /CS /DeviceRGB>>}

\newcommand\congsub{\mathrel{%
  \ooalign{\raise0.2ex\hbox{$\sqsubset$}\cr\hidewidth\raise-0.8ex\hbox{\scalebox{0.9}{$\sim$}}\hidewidth\cr}}}

\def\calS{{\mathcal{S}}}  %
\def\planter/{{\texttt{planter}}}
\def\topcomp/{{\texttt{top}}}
\def\leftcomp/{{\texttt{left}}}
\def\rightcomp/{{\texttt{right}}}
\def\arm/{{\texttt{arm}}}
\def\blocker/{{\texttt{blocker}}}
\def\bitalley/{{\texttt{bitAlley}}}
\def\blocker/{{\texttt{blocker}}}
\def\bumper/{{\texttt{bumper}}}

\title{Universal Simulation of Directed Systems in the abstract Tile Assembly Model Requires Undirectedness}

\author{
 Jacob Hendricks%
    \thanks{Department of Computer Science and Information Systems, University of Wisconsin - River Falls,
    \protect\url{jacob.hendricks@uwrf.edu}
    Supported in part by National Science Foundation Grants CCF-1117672 and CCF-1422152.}
\and
  Matthew J. Patitz%
    \thanks{Department of Computer Science and Computer Engineering, University of Arkansas,
    \protect\url{patitz@uark.edu}
    Supported in part by National Science Foundation Grants CCF-1117672 and CCF-1422152.}
\and
 Trent A. Rogers%
    \thanks{Department of Computer Science and Computer Engineering, University of Arkansas,
    \protect\url{tar003@uark.edu}
    Supported in part by National Science Foundation Grants CCF-1117672 and CCF-1422152, and Graduate Research Fellowship Grant DGE-1450079}
}
\date{}

\input{commands-tam.sty}

\begin{document}

\maketitle
\vspace{-6ex}
\begin{abstract}
As a mathematical model of tile-based self-assembling systems, Winfree's abstract Tile Assembly Model (aTAM) has proven to be a remarkable platform for studying and understanding the behaviors and powers of self-assembling systems.  Furthermore, as it is capable of Turing universal computation, the aTAM allows algorithmic self-assembly, in which the components can be designed so that the rules governing their behaviors force them to inherently execute prescribed algorithms as they combine.  This power has yielded a wide variety of theoretical results in the aTAM utilizing algorithmic self-assembly to design systems capable of performing complex computations and forming extremely intricate structures.  Adding to the completeness of the model, in FOCS 2012 the aTAM was shown to also be \emph{intrinsically universal}, which means that there exists one single tile set such that for any arbitrary input aTAM system, that tile set can be configured into a ``seed'' structure which will then cause self-assembly using that tile set to simulate the input system, capturing its full dynamics modulo only a scale factor.  However, the ``universal simulator'' of that result makes use of nondeterminism in terms of the tiles placed in several key locations when different assembly sequences are followed.  This nondeterminism remains even when the simulator is simulating a system which is \emph{directed}, meaning that it has exactly one unique terminal assembly and for any given location, no matter which assembly sequence is followed, the same tile type is always placed there. The question which then arose was whether or not that nondeterminism is fundamentally required, and if any universal simulator must in fact utilize more nondeterminism than directed systems when simulating them.

In this paper, we answer that question in the affirmative:  the class of directed systems in the aTAM is not intrinsically universal, meaning there is no universal simulator for directed systems which itself is always directed.  This result provides a powerful insight into the role of nondeterminism in self-assembly, which is itself a fundamentally nondeterministic process occurring via unguided local interactions.  Furthermore, to achieve this result we leverage powerful results of computational complexity hierarchies, including tight bounds on both best and worst-case complexities of decidable languages, to tailor design systems with precisely controllable space resources available to computations embedded within them. We also develop novel techniques for designing systems containing subsystems with disjoint, mutually exclusive computational powers. The main result will be important in the development of future simulation systems, and the supporting design techniques and lemmas will provide powerful tools for the development of future aTAM systems as well as proofs of their computational abilities.

\end{abstract}

\section{Introduction}\label{sec:intro}

Self-assembly is the process by which relatively simple components begin in a disorganized state and, without external guidance but only by following local rules of interaction, autonomously combine to form more complex structures.  Self-assembling systems are ubiquitous in nature, and self-assembly processes govern the formation of everything from ice crystals to cellular membranes, and despite the seemingly random nature of these systems, they serve as a ratchet for the generation of complexity on scales from the nano \cite{ke2012three,rothemund2004algorithmic} to the macro \cite{Whitesides16042002}. The random motions of components are leveraged to allow binding opportunities to growing structures, and if the dynamics of interactions fall into ranges which are restrictive enough, without being too restrictive, ordered assemblies can form.  Clearly, nondeterminism plays key roles in such systems, and our main result helps to elucidate one of them.

The abstract Tile Assembly Model (aTAM) is a mathematical abstraction of self-assembling systems based on square ``tile'' components which have ``glues'' on their sides that allow them to bind together when glues on abutting edges of tiles have matching types.  Despite being a very simplified model which uses geometrically basic building blocks, the aTAM is computationally universal \cite{Winf98} and a powerful model allowing for very efficient algorithmic self-assembly of shapes \cite{SolWin07,RotWin00}.  Another noteworthy aspect of the model is that it is intrinsically universal (IU) \cite{IUSA}, meaning that there exists a single tile set, $U$, such that given any arbitrary aTAM system $\calT$, $U$ can be given an initial configuration which will cause it to faithfully simulate the full dynamics of $\calT$ modulo a constant scale factor (dependent on $\calT$).  Since the result of \cite{IUSA}, several other results related to IU have been used to examine and classify the relative powers of a variety of models of self-assembly and classes of systems within them \cite{2HAMIU,Duples,jSignals3D,BreakableDuples,2HAMSim,IUNeedsCoop,Polygons,Polyominoes,OneTile}, thus developing a complexity hierarchy which can be used to categorize models and systems within them.

In this paper, we investigate the problem of characterizing the role of nondeterminism within the aTAM, which has previously been explored in a variety of different aspects \cite{BryChiDotKarSek10,Dot09,KaoSchweller08}.  At its core, the aTAM is an asynchronous and nondeterministic model in which tile attachments to a growing assembly, while constrained by the requirement that sufficient matching glues must bind, are random with respect to the sequence of locations and sometimes the particular types of tiles which bind.  The amount of nondeterminism of different aTAM systems can vary wildly, with some systems having uncountably infinite sets of producible, or even terminal (i.e. those which cannot grow any further), assemblies and/or sequences of assembly, to those having exactly one producible assembly and even some with just one possible assembly sequence.  This leads to questions about whether or not, and possibly how much, nondeterminism is required to give the aTAM its full power.  In this paper, we focus on this question from the perspective of the ``universal aTAM simulator'' of \cite{IUSA}, which by design has several so-called ``points of competition'', where different assembly sequences of the simulator, as it simulates a system $\calT$, race to grow paths to those points, with the first path to arrive causing a tile type specific to that path to be placed.  The fact that there are multiple assembly sequences, each growing a different path first, causes nondeterminism in the types of tiles placed in these locations.  The use of such locations is so fundamental to that universal simulator's design, allowing it to continue growth of portions of the assembly without having to rely on future paths which may or may not ever arrive, that even when it is simulating directed aTAM systems, which are those that have exactly one terminal assembly and only one possible tile type in any location regardless of the assembly sequence, the simulator itself must be undirected.  It has remained unknown whether or not such nondeterminism is fundamentally required by a universal simulator, and in Theorem~\ref{thm:directedNotIU} we prove that it is.  That is, we prove that the class containing all directed aTAM systems is not IU, meaning that there exists no tile set $U$ such that, given an arbitrary directed aTAM system, $U$ can be configured to create an aTAM  system which simulates it while itself being directed.  Stated another way, it means that any universal simulator for the aTAM must be more nondeterministic than some of the systems which it simulates.

While our main result presents key insights into the properties required of aTAM and other tile-based simulators, and shows how nondeterminism with respect to the selection of assembly sequences can force nondeterminism with respect to assemblies produced by any universal simulator, other key contributions of this paper include the development of several new system design techniques and tools useful in proving properties about the computational resources available to be harnessed by embedded algorithms, which themselves provide additional insights into the computations possible using static combinations of matter filling non-reusable space.  More specifically, we make use of computational complexity results which combine extremely tight worst-case and best-case space complexity bounds for decidable languages \cite{VeryHardLanguages}, as well as novel techniques for controlling the ``input bandwidth'' and geometries of carefully designed subassemblies which perform complex computations that are effectively hidden from each other.  These designs are likely to be useful in further tile-based self-assembly results, especially impossibility results.  Furthermore, we develop several important and potentially very useful tools which can be used to characterize properties of tile assembly systems which are simulating others, e.g. Lemma~\ref{lem:noCheating} which proves that the space complexity of computations which can be performed by a system simulating a type of system known as a zig-zag system is asymptotically no greater than that of the computations which can be performed by the original system, despite the scale factor allowed the simulator.

Section~\ref{sec:prelims} provides a set of preliminary definitions used throughout the paper, and the following section a formal statement of our main result. Next are two sections dedicated to a high-level overview of the proof,
\iffull
with sections including the full technical details following.
\else
and due to space constraints full technical details can be found in \cite{DirectedNotIUArxiv}.
\fi

\vspace{-10pt}
\section{Preliminaries}\label{sec:prelims}

In this section we provide an informal definition of the aTAM and then define what it means for one tile assembly system to simulate another, and the notion of intrinsic universality.

\subsection{Informal description of the abstract Tile Assembly Model}
\label{sec-tam-informal}

This section gives a brief informal sketch of the abstract Tile Assembly Model (aTAM). %
See Section~\ref{sec-tam-formal} for a formal definition of the aTAM.

A \emph{tile type} is a unit square with four sides, each consisting of a \emph{glue label}, often represented as a finite string, and a nonnegative integer \emph{strength}. A glue~$g$ that appears on multiple tiles (or sides) always has the same strength~$s_g$. %
There are a finite set $T$ of tile types, but an infinite number of copies of each tile type, with each copy being referred to as a \emph{tile}. An \emph{assembly}
is a positioning of tiles on the integer lattice $\Z^2$, described  formally as a partial function $\alpha:\Z^2 \dashrightarrow T$. %
Let $\mathcal{A}^T$ denote the set of all assemblies of tiles from $T$, and let $\mathcal{A}^T_{< \infty}$ denote the set of finite assemblies of tiles from $T$.
We write $\alpha \sqsubseteq \beta$ to denote that $\alpha$ is a \emph{subassembly} of $\beta$, which means that $\dom\alpha \subseteq \dom\beta$ and $\alpha(p)=\beta(p)$ for all points $p\in\dom\alpha$.
Two adjacent tiles in an assembly \emph{interact}, or are \emph{attached}, if the glue labels on their abutting sides are equal and have positive strength. %
Each assembly induces a \emph{binding graph}, a grid graph whose vertices are tiles, with an edge between two tiles if they interact.
The assembly is \emph{$\tau$-stable} if every cut of its binding graph has strength at least~$\tau$, where the strength   of a cut is the sum of all of the individual glue strengths in the cut.

A \emph{tile assembly system} (TAS) is a triple $\calT = (T,\sigma,\tau)$, where $T$ is a finite set of tile types, $\sigma:\Z^2 \dashrightarrow T$ is a finite, $\tau$-stable \emph{seed assembly},
and $\tau$ is the \emph{temperature}.
An assembly $\alpha$ is \emph{producible} if either $\alpha = \sigma$ or if $\beta$ is a producible assembly and $\alpha$ can be obtained from $\beta$ by the stable binding of a single tile.
In this case we write $\beta\to_1^\calT \alpha$ (to mean~$\alpha$ is producible from $\beta$ by the attachment of one tile), and we write $\beta\to^\calT \alpha$ if $\beta \to_1^{\calT*} \alpha$ (to mean $\alpha$ is producible from $\beta$ by the attachment of zero or more tiles).
When $\calT$ is clear from context, we may write $\to_1$ and $\to$ instead.
We let $\prodasm{\calT}$ denote the set of producible assemblies of $\calT$.
An assembly is \emph{terminal} if no tile can be $\tau$-stably attached to it.
We let   $\termasm{\calT} \subseteq \prodasm{\calT}$ denote  the set of producible, terminal assemblies of $\calT$.
A TAS $\calT$ is \emph{directed} if $|\termasm{\calT}| = 1$. Hence, although a directed system may be nondeterministic in terms of the order of tile placements,  it is deterministic in the sense that exactly one terminal assembly is producible (this is analogous to the notion of {\em confluence} in rewriting systems).

\vspace{-5pt}
\subsection{Simulation}
\label{sec:simulation_def}

To state our main results, we must formally define what it means for one TAS to ``simulate'' another.  Our definitions come from \cite{IUNeedsCoop}.
Intuitively, simulation of a system $\calT$ by a system $\calS$ requires that there is some scale factor $m \in \Z^+$ such that $m \times m$ squares of tiles in $\calS$ represent individual tiles in $\calT$, and there is a ``representation function'' capable of inspecting assemblies in $\calS$ and mapping them to assemblies in $\calT$.

From this point on, let $T$ be a tile set, and let $m\in\Z^+$.
An \emph{$m$-block supertile} over $T$ is a partial function $\alpha : \Z_m^2 \dashrightarrow T$, where $\Z_m = \{0,1,\ldots,m-1\}$.
Let $B^T_m$ be the set of all $m$-block supertiles over $T$.
The $m$-block with no domain is said to be $\emph{empty}$.
For a general assembly $\alpha:\Z^2 \dashrightarrow T$ and $(x_0, x_1)\in\Z^2$, define $\alpha^m_{x_0,x_1}$ to be the $m$-block supertile defined by $\alpha^m_{x_0, x_1}(i_0, i_1) = \alpha(mx_0+i_0, mx_1+i_1)$ for $0 \leq i_0,i_1< m$.
For some tile set $S$, a partial function $R: B^{S}_m \dashrightarrow T$ is said to be a \emph{valid $m$-block supertile representation} from $S$ to $T$ if for any $\alpha,\beta \in B^{S}_m$ such that $\alpha \sqsubseteq \beta$ and $\alpha \in \dom R$, then $R(\alpha) = R(\beta)$.

For a given valid $m$-block supertile representation function $R$ from tile set~$S$ to tile set $T$, define the \emph{assembly representation function}\footnote{Note that $R^*$ is a total function since every assembly of $S$ represents \emph{some} assembly of~$T$; the functions $R$ and $\alpha$ are partial to allow undefined points to represent empty space.}  $R^*: \mathcal{A}^{S} \rightarrow \mathcal{A}^T$ such that $R^*(\alpha') = \alpha$ if and only if $\alpha(x_0, x_1) = R\left(\alpha'^m_{x_0,x_1}\right)$ for all $(x_0,x_1) \in \Z^2$.
For an assembly $\alpha' \in \mathcal{A}^{S}$ such that $R(\alpha') = \alpha$, $\alpha'$ is said to map \emph{cleanly} to $\alpha \in \mathcal{A}^T$ under $R^*$ if for all non empty blocks $\alpha'^m_{x_0,x_1}$, $(x_0,x_1)+(u_0,u_1) \in \dom \alpha$ for some $u_0,u_1 \in U_2$ such that $u_0^2 + u_1^2 \leq 1$, or if $\alpha'$ has at most one non-empty $m$-block~$\alpha^m_{0, 0}$.

In other words, $\alpha'$ may have tiles on supertile blocks representing empty space in $\alpha$, but only if that position is adjacent to a tile in $\alpha$.  We call such growth ``around the edges'' of $\alpha'$ \emph{fuzz} and thus restrict it to be adjacent to only valid supertiles, but not diagonally adjacent (i.e.\ we do not permit \emph{diagonal fuzz}).

In the following definitions, let $\mathcal{T} = \left(T,\sigma_T,\tau_T\right)$ be a TAS, let $\mathcal{S} = \left(S,\sigma_S,\tau_S\right)$ be a TAS, and let $R$ be an $m$-block representation function $R:B^S_m \rightarrow T$.

\begin{definition}
\label{def-equiv-prod} We say that $\mathcal{S}$ and $\mathcal{T}$ have \emph{equivalent productions} (under $R$), and we write $\mathcal{S} \Leftrightarrow \mathcal{T}$ if the following conditions hold:
\begin{enumerate}
        \item $\left\{R^*(\alpha') | \alpha' \in \prodasm{\mathcal{S}}\right\} = \prodasm{\mathcal{T}}$.
        \item $\left\{R^*(\alpha') | \alpha' \in \termasm{\mathcal{S}}\right\} = \termasm{\mathcal{T}}$.
        \item For all $\alpha'\in \prodasm{\mathcal{S}}$, $\alpha'$ maps cleanly to $R^*(\alpha')$.
\end{enumerate}
\end{definition}

\begin{definition}
\label{def-t-follows-s} We say that $\mathcal{T}$ \emph{follows} $\mathcal{S}$ (under $R$), and we write $\mathcal{T} \dashv_R \mathcal{S}$ if $\alpha' \rightarrow^\mathcal{S} \beta'$, for some $\alpha',\beta' \in \prodasm{\mathcal{S}}$, implies that $R^*(\alpha') \to^\mathcal{T} R^*(\beta')$.
\end{definition}

\begin{definition}
\label{def-s-models-t} We say that $\mathcal{S}$ \emph{models} $\mathcal{T}$ (under $R$), and we write $\mathcal{S} \models_R \mathcal{T}$, if for every $\alpha \in \prodasm{\mathcal{T}}$, there exists $\Pi \subset \prodasm{\mathcal{S}}$ where $R^*(\alpha') = \alpha$ for all $\alpha' \in \Pi$, such that, for every $\beta \in \prodasm{\mathcal{T}}$ where $\alpha \rightarrow^\mathcal{T} \beta$, (1) for every $\alpha' \in \Pi$ there exists $\beta' \in \prodasm{\mathcal{S}}$ where $R^*(\beta') = \beta$ and $\alpha' \rightarrow^\mathcal{S} \beta'$, and (2) for every $\alpha'' \in \prodasm{\mathcal{S}}$ where $\alpha'' \rightarrow^\mathcal{S} \beta'$, $\beta' \in \prodasm{\mathcal{S}}$, $R^*(\alpha'') = \alpha$, and $R^*(\beta') = \beta$, there exists $\alpha' \in \Pi$ such that $\alpha' \rightarrow^\mathcal{S} \alpha''$.
\end{definition}

The previous definition essentially specifies that every time $\mathcal{S}$ simulates an assembly $\alpha \in \prodasm{\mathcal{T}}$, there must be at least one valid growth path in $\mathcal{S}$ for each of the possible next steps that $\mathcal{T}$ could make from $\alpha$ which results in an assembly in $\mathcal{S}$ that maps to that next step.

\begin{definition}
\label{def-s-simulates-t} We say that $\mathcal{S}$ \emph{simulates} $\mathcal{T}$ (under $R$) if $\mathcal{S} \Leftrightarrow_R \mathcal{T}$ (equivalent productions), $\mathcal{T} \dashv_R \mathcal{S}$ and $\mathcal{S} \models_R \mathcal{T}$ (equivalent dynamics).
\end{definition}

\newcommand{\REPL}{\mathsf{REPR}}
\newcommand{\frakC}{\mathfrak{C}}

\vspace{-5pt}
\subsection{Intrinsic Universality}
\label{sec:iu_def}
Now that we have a formal definition of what it means for one tile system to simulate another, we can proceed to formally define the concept of intrinsic universality, i.e., when there is one general-purpose tile set that can be appropriately programmed to simulate any other tile system from a specified class of tile systems.

Let $\REPL$ denote the set of all supertile representation functions (i.e., $m$-block supertile representation functions for some $m\in\Z^+$).
Define $\frakC$ to be a class of tile assembly systems, and let $U$ be a tile set.
Note that each element of $\frakC$, $\REPL$, and $\mathcal{A}^U_{< \infty}$ is a finite object, hence encoding and decoding of simulated and simulator assemblies can be represented in a suitable format for computation in some formal system such as Turing machines.

\begin{definition}\label{def:iu-specific-temp}
We say $U$ is \emph{intrinsically universal} for $\frakC$ \emph{at temperature} $\tau' \in \Z^+$ %
if there are computable functions $\mathcal{R}:\frakC \to \REPL$ and $S:\frakC \to \mathcal{A}^U_{< \infty}$ such that, for each $\mathcal{T} = (T,\sigma,\tau) \in \frakC$, there is a constant $m\in\N$ such that, letting $R = \mathcal{R}(\mathcal{T})$, $\sigma_\mathcal{T}=S(\mathcal{T})$, and $\mathcal{U}_\mathcal{T} = (U,\sigma_\mathcal{T},\tau')$, $\mathcal{U}_\mathcal{T}$ simulates $\mathcal{T}$ at scale $m$ and using supertile representation function~$R$.
\end{definition}
That is, $\mathcal{R}(\mathcal{T})$ outputs a representation function that interprets assemblies of $\mathcal{U}_\mathcal{T}$ as assemblies of $\mathcal{T}$, and $S(\mathcal{T})$ outputs the seed assembly used to program tiles from $U$ to represent the seed assembly of $\mathcal{T}$.

\begin{definition}
\label{def:iu-general}
We say that~$U$ is \emph{intrinsically universal} for $\frakC$ if it is intrinsically universal for $\frakC$ at some temperature $\tau'\in Z^+$.
\end{definition} %

\begin{definition}
We say that $\frakC$ is intrinsically universal if there exists some $U$ that is intrinsically universal for $\frakC$ and for every $\calT \in \frakC$ and $\mathcal{U}_\calT$ which simulates it, $\mathcal{U}_\calT \in \frakC$.
\end{definition}

\section{The Directed aTAM is not Intrinsically Universal}

Let $\frak{D}$ represent the class of all tile assembly systems within the aTAM which are directed.

\begin{theorem}\label{thm:directedNotIU}
$\frak{D}$ is not intrinsically universal.
\end{theorem}

Theorem~\ref{thm:directedNotIU} states that there exists no aTAM tile set $U$ such that, for any directed aTAM tile assembly system $\mathcal{D} \in \frak{D}$, where $\mathcal{D} = (T,\sigma,\tau)$, there exists a directed aTAM system $\mathcal{U}_{\mathcal{D}} \in \frak{D}$, where $\mathcal{U}_{\mathcal{D}} = (U,\sigma_{\mathcal{D}},\tau')$, scale factor $m \in \mathbb{N}$, and representation function $R: B^U_m \rightarrow T$, such that $\mathcal{U}_{\mathcal{D}}$ simulates $\mathcal{D}$ under $m$-block representation function $R$ at scale factor $m$.  Essentially, there exists no ``universal'' tile set such that for any directed aTAM system, that tile set can be configured in a simulating system which simulates the original and is itself directed too.

Our proof of Theorem~\ref{thm:directedNotIU} will be by contradiction.  Therefore, assume that such a universal tile set $U$, which can be used to simulate any directed system while using a directed system, exists.  Given that $U$, we define an aTAM system $\calT = (T,\sigma,2)$ which is directed and forms an infinite terminal assembly, explain the growth of $\calT$, and verify that it is directed. We provide a high-level overview of $\calT$ in Section~\ref{sec:the-system-overview}.  We then show why there exists no directed aTAM system $\calS = (U,\sigma_{\calT},\tau')$ which simulates $\calT$.  Section~\ref{sec:proof-overview} contains a very high-level overview of that proof.  Full details of $\calT$
\iffull
can be found in Section~\ref{sec:the-system}, and for the impossibility proof in Section~\ref{sec:proof-details}.
\else
and of the impossibility proof can be found in \cite{DirectedNotIUArxiv}.
\fi

\section{Overview of the Directed aTAM System $\calT$}\label{sec:the-system-overview}

At the highest level, $\calT$ self-assembles an infinite structure, starting from a single seed tile placed at the origin, and growing from left to right.  In well-defined intervals, as the assembly grows eastward it initiates upward growths, an infinite series of sets of three ``modules'' which are subassemblies able to grow almost entirely independently of each other once the main horizontal growing structure has placed the tiles which serve as the ``input'' for the growth of each.  The aTAM is computationally universal \cite{Winf98}, and in fact it is quite straightforward to design a tile assembly system which simulates the computation of an arbitrary Turing machine $M$ (e.g. \cite{jSADS,jCCSA}) by growing rows of tiles, one above the other, where each row represents the full configuration of $M$ at a given time step (i.e. the tape contents, read/write head location, and state) in the values of the glues encoded on their north sides, and the row immediately above it represents the full configuration of $M$ at the next time step (by designing the tile types appropriately so that the only tiles which can attach above a given row ensure that the new northern glue above a position which just had the read/write head encodes the value that would have been output given the state of $M$ and the cell's previous value, and depending on the direction the head would have moved, either the tile representing the cell to the left or write would have a glue encoding the new state of $M$ and the current value of that cell).  To provide a logically infinite tape, the tiles can be designed to grow rows ``on demand'' by extending a row by one tile each time the simulated read/write head attempts to move past the end of the currently represented row.

The three modules which grow upward are logically grouped so that there is one of each type in a set.  These three modules are designed so that they simulate three computations which require asymptotically differing space resources.  As each set is initiated with inputs of increasing values, and as the assembly grows infinitely to the right, those space requirements ensure that the smallest module cannot perform the computations of the larger two, and the mid-sized module cannot perform the computations of the largest.  The computations carried out by each set of grouped modules as well as the geometries to which they are each constrained are carefully designed such that two of the modules are necessarily completely ``ignorant'' of the eventual outputs of the others.  However, these two modules are designed so that after performing their computations, they grow assemblies representing bit strings corresponding to the outputs of their computations in locations across a one tile wide gap from each other, which we call the \bitalley/.  In locations where output bits of the two computations match, tiles attach between tiles for those bit positions.  The third module independently computes the results of the computations of both other modules and if and only if there will be no matching bits between them, it grows an assembly which is a single tile wide path down through the \bitalley/ (thus it is guaranteed not to crash into any tiles in the \bitalley/, regardless of the ordering of tile attachments).  As the overall assembly grows further right, the inputs to the modules increase and the computations simulated by the modules require more resources and the \bitalley/s become arbitrarily long.  We are able to first show that $\calT$ is directed, and then that no simulating system can be built using the tiles of a universal simulating tile set $U$ and be itself directed.  This is because any such directed simulator is forced by the dynamics of correct simulation, the mutual obfuscation of computations across modules, and geometric constraints, to effectively create bottlenecks which do not allow enough information to be transmitted to the growing assembly for correct growth and therefore simulation.  The intuition is that the simulator has to make ``guesses'' about when it may need to place tiles which cooperate across a \bitalley/ (i.e. glues from the tiles on both sides of the gap are required to allow the attachment of one between them) which, due to the fact that space cannot be reused in the aTAM, doom it to failure.  Furthermore, these guesses are required not by nondeterminism about which tiles can be placed in locations by $\calT$, since after all $\calT$ is directed, but rather due to the ordering of arrival of tiles - the particular assembly sequence which may be followed.

\vspace{-5pt}
\subsection{Overview of modules of $\calT$}
\vspace{-5pt}

Figure~\ref{fig:counter-and-wedge-TMs-overview} shows a schematic depiction of a portion of the terminal assembly of $\calT$.  We now give a very high-level description of each of the main modules, and full details can be found in
\iffull
Section~\ref{sec:the-system}.
\else
\cite{DirectedNotIUArxiv}.
\fi

\begin{figure}[htp]
\begin{center}
\includegraphics[width=6.5in]{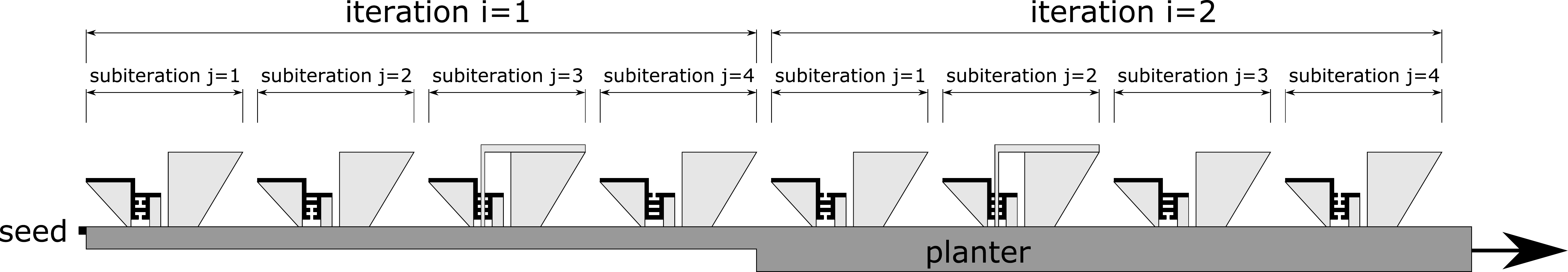}
\caption{A high-level schematic depiction of a portion of the infinite assembly produced by a directed aTAM system $\mathcal{T}$ which cannot be simulated by any directed universal simulator.}
\label{fig:counter-and-wedge-TMs-overview}
\end{center}
\vspace{-20pt}
\end{figure}

Beginning from the seed, the module which grows horizontally and initiates growth of sets of modules to its north is called the \planter/.  The \planter/ grows in a zig-zag, up and down manner, growing one column at a time.  Essentially, its job is to manage a set of nested counters, whose values are used to (1) determine the correct spacing between the modules to the \planter/'s north, and (2) serve as input to those modules.  The outermost of the nested counters counts $0 < i < \infty$, with each $i$ being what we call an \emph{iteration}.  For each value of $i$ that it counts, it holds that counter constant while it increments an inner counter from $0$ to (approximately) $2^i$.  For each value of $j$ it initiates the growth of what we call a \emph{subiteration}.  See Figure~\ref{fig:counter-and-wedge-TMs-iteration} for a high-level overview of one type of subiteration. For each subiteration, the \planter/ counts out a sequence of spacing columns (i.e. columns whose sole purpose is to put horizontal space between modules) while also computing the value $\log(i)$ and then rotating the values of the bits representing $\log(i)$ upward so that they are encoded in a row of glues on the north sides of the northern tiles of the \planter/\footnote{Note that throughout this paper, $\log$ means $\log_2$, and we use the shorthand $\log(i)$ to mean $\lceil \log(i) \rceil$.}.  From these, a \leftcomp/ module begins growth.  This module performs a stacked up series of $i$ Turing machine simulations on progressively increasing input values, with each simulation outputting a $0$ (for a rejecting computation) or a $1$ (accepting).  At the top of the stack of computations, the string of output bits is rotated to the right and then grown downward to the right of the \leftcomp/ module.  Once that growth reaches a specially marked location, the values of those bits are rotated to the right where they are presented as the eastern glues of the tiles forming the \bitalley/. (See Figure~\ref{fig:gap-overview} for a depiction of a southern portion of a \bitalley/.)

\begin{figure}[htp]
\vspace{-20pt}
\begin{center}
\includegraphics[width=5.5in]{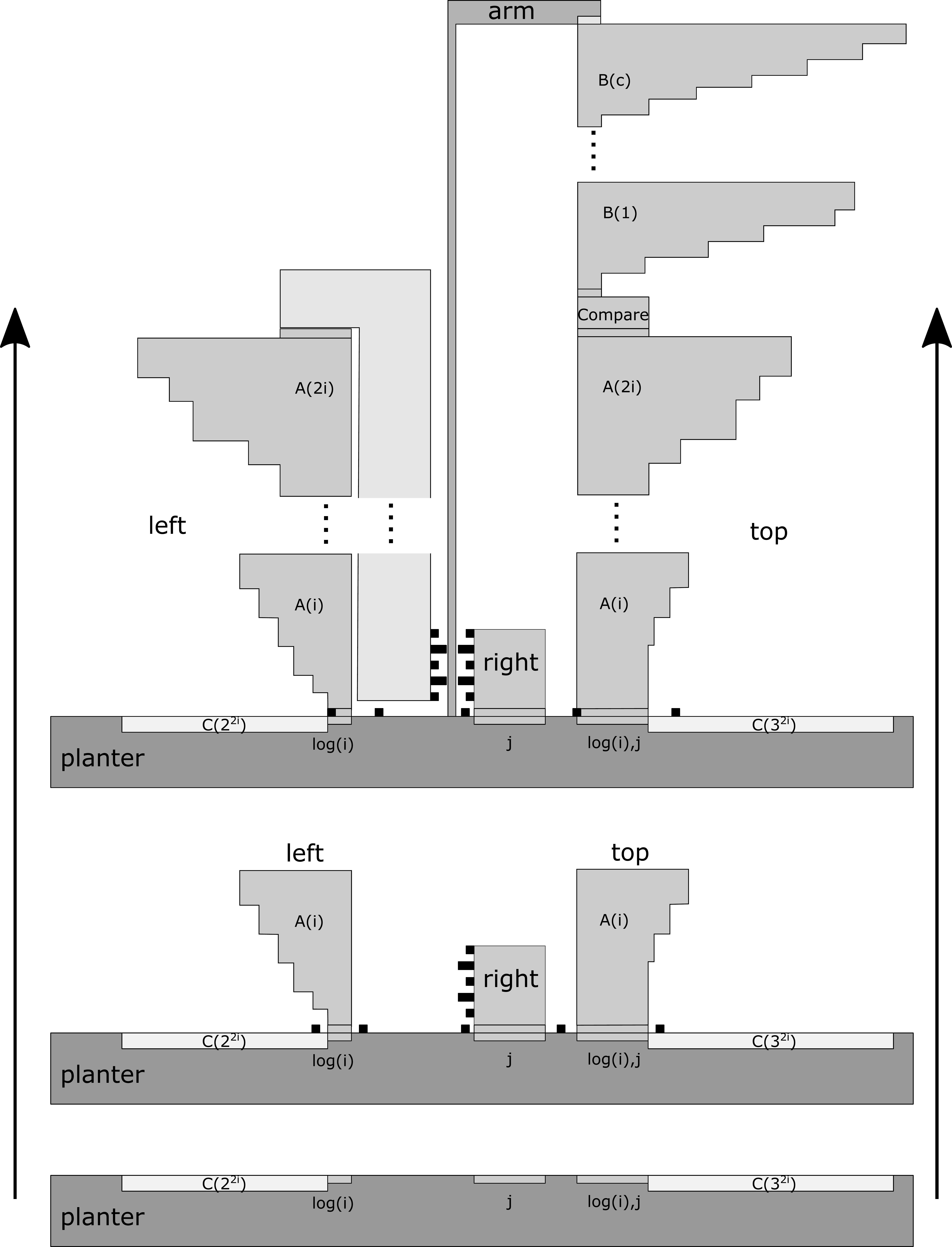}
\caption{A high-level schematic depiction of one possible ordering of growth of the modules of an empty subiteration.  (Bottom) The \planter/ lays out the inputs for the modules at the necessary spacings to prevent them from colliding, (Second) The \leftcomp/, \rightcomp/, and \topcomp/ modules begin growth, (Top) Once the \topcomp/ completes it initiates the growth of the \arm/ which grows down through the \bitalley/. Note that an \arm/ only grows in the \bitalley/ of an empty subiteration, unlike the \bitalley/ in Figure~\ref{fig:gap-overview} which shows tiles cooperatively binding across the \bitalley/ of a non-empty subiteration. Also, empty subiterations occur exponentially more rarely than non-empty ones.}
\label{fig:counter-and-wedge-TMs-iteration}
\end{center}
\vspace{-20pt}
\end{figure}

After growing a few spacing columns past the initiation point of the \leftcomp/ module, the \planter/ rotates the value of $j$ to its north side to initiate growth of a \rightcomp/ module.  This module simply rotates the values of the bits of $j$ to the left so they can be presented across the \bitalley/ from the bits output by the \leftcomp/.  Note that as the iteration number $i$ increases, so does the number of bits presented on each side of the \bitalley/, as the \leftcomp/ performs (approximately) $i$ Turing machine simulations, and \rightcomp/ actually receives the value of $j$ in binary padded with $0$'s as necessary to be the same length.

The final module to be initiated by the \planter/ in each subiteration is the \topcomp/ module.  This module receives as input both the values $\log(i)$ and $j$.  It first performs the same $i$ simulations that the \leftcomp/ performs, generating the same output bits.  It then compares those bits to the bits of $j$ to determine if there are any locations where the bits are the same.  If there are, then in the \bitalley/ there will be tiles which attach between them across the gap in those locations, and the \topcomp/ module halts its growth (in this subiteration).  It is guaranteed that in exactly one subiteration of each iteration that there will be no matching bits, since each subiteration performs the same \leftcomp/ computations on the same input and there is a unique subiteration for every possible bit string of length $i$, exactly one of which can be the complement of \leftcomp/'s output on that input.  In this special subiteration of the iteration, which we call the \emph{empty} subiteration (because the \bitalley/ will be empty of tiles cooperating across the gap), the \topcomp/ performs a new set of computations to determine which of a large number (relative to the number of tile types in the claimed universal simulator $U$) of \arm/ modules to grow.  The \arm/ module grows over to a position directly above the \bitalley/, then grows a single tile wide column of tiles down through the \bitalley/ until it crashes into the \planter/, with the specific type of tile used for the \arm/ determined by the final computations performed by the \topcomp/ module.
This completes the growth of a subiteration, and the growth of subiterations and iterations occurs for infinite numbers of each.

\begin{wrapfigure}{r}{.45\textwidth}
\vspace{-30pt}
\begin{center}
\includegraphics[width=2.3in]{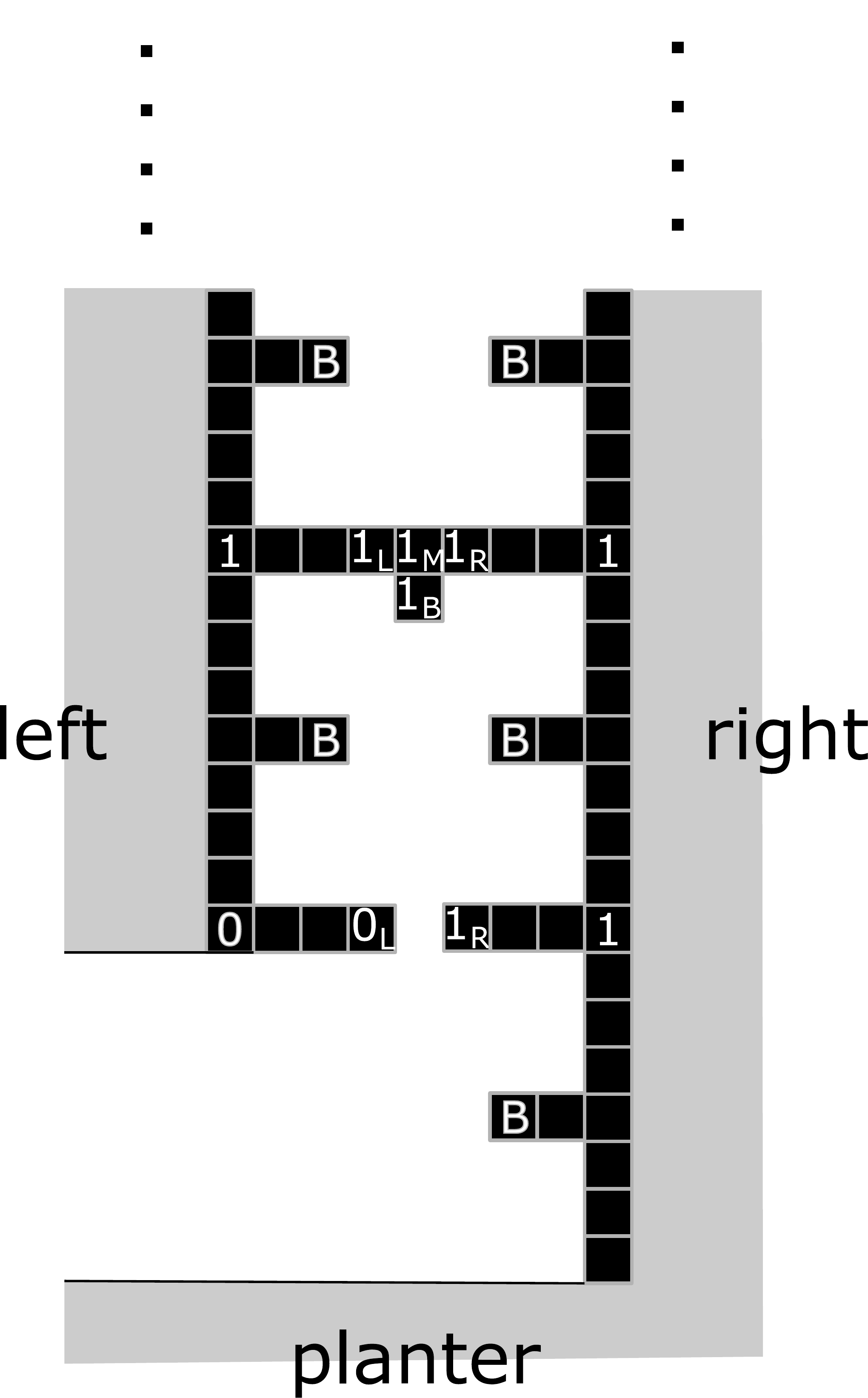}
\caption{Example \bitalley/ portion between \leftcomp/ and \rightcomp/ modules of a non-empty subiteration.}
\label{fig:gap-overview}
\end{center}
\vspace{-40pt}
\end{wrapfigure}

\vspace{-15pt}
\subsection{Directedness of $\calT$}
\vspace{-5pt}

The system $\calT$ is directed because there are no locations where tiles of multiple types might be placed during different assembly sequences, and this is ensured by carefully dictating the growth of each module (all grow in zig-zag manners), and the amount of space required for each is carefully computed and accounted for by the \planter/ so none of them can collide.  Finally, the \arm/ will only grow in empty subiterations, which can be assured by the \topcomp/ module performing the computations of \leftcomp/ and comparing the output bits to $j$, so it will never collide with tiles in the \bitalley/.  Thus, despite the fact that there are an infinite number of unique assembly sequences in $\calT$, they all result in the exact same terminal assembly in the limit.

\vspace{-10pt}
\section{Overview of Impossibility of Simulation}\label{sec:proof-overview}
\vspace{-5pt}

In this section we provide a high-level overview of the proof that $\calS$ does not simulate $\calT$.  More details can be found in 
\iffull
Section~\ref{sec:proof-details}.
\else
\cite{DirectedNotIUArxiv}.
\fi

The general idea behind the proof that $\calS$ cannot simulate $\calT$ is based around creating a situation in $\calT$ where there is a one tile wide gap between two tiles such that, depending on their types, they may or may not cooperate to place a tile in between them (i.e. a tile may bind using one glue from each of them).  However, if and only if all of these tiles in the \bitalley/ do not cooperate to place a tile between them, another assembly will grow between them without binding to either of their glues.  In $\calT$, the gap is exactly one tile wide and so is the assembly that may grow down through it. Since we are proving by contradiction, assume that such an $\calS$ exists and that it has tile set $U$ with size $|U| = t$. We design $\calT$ such that the number of unique \arm/ module tiles (which are the ones that grow between the two tiles if they do not cooperate) is exponentially larger than $t$.  This forces the simulation scale factor $m$ used by $\calS$ to be larger than $1$ because any macrotile created from tiles in $U$ must have enough tiles to uniquely identify any of the tile types in $\calT$.  Then we also note that geometrically, the only way to get two tiles to cooperate to place a tile in between them is for them to grow to positions with less than or equal to a single tile wide gap between them, which is not enough room for the macrotile of an \arm/ module, with $m > 1$, to pass through if necessary.  While the general idea seems simple, first, care must be taken in designing $\calT$ so that an \arm/ module will be grown if and only if the tiles will not cooperate across the gap, with no chance for a disagreement and collision since $\calT$ must be directed, so the portion of the assembly which initiates the growth of the \arm/ must be able to compute the tiles which will appear across the gap from each other.  Then, it must be shown that $\calS$ is forced to grow all the way to a single tile wide gap even when cooperation won't be necessary, thus blocking the \arm/. The main difficulties arise with the realization that the simulating system could attempt to compute in advance if cooperation will occur and, if so,  grow to the one tile wide gap which allows for cooperation, but if not, stop growth short of that to leave enough room for the \arm/ module to grow through.  The resulting complexity of $\calT$ arises from the need to create a system which is ``confusing'' enough for the simulator that the modules growing the macrotiles representing the tiles which may cooperate across the gap are unable to pre-compute the answer to whether or not cooperation will be necessary.  Essentially, the fact that $\calS$ cannot both cooperate and/or grow a full tile-representing assembly through a single tile wide gap dooms it to failure, but extensive machinery is required to force the situation.

A key tool in the proof is that in an arbitrary subiteration $j$ of an arbitrary iteration $i$, the output of the \leftcomp/ module is impossible to compute from within either the \planter/ or the \rightcomp/ modules, and the output of the \rightcomp/ is impossible to compute from within the \leftcomp/.  The reason for this is that (1) the Turing machines being simulated within the \leftcomp/ modules are deciding languages which cannot be recognized in infinitely often best-case space complexity \cite{VeryHardLanguages} which is greater than the space resources available to the \planter/ and \rightcomp/ modules, and thus the outputs of \leftcomp/ modules cannot be computed by them, and (2) the input $j$ passed to the \rightcomp/ module is asymptotically much greater in size than the amount of information which can be input to the \leftcomp/ module through the only $\log(\log(i))$ macrotiles allowed in the bottom row of the \leftcomp/ module to encode the value $\log(i)$, making it unable to get asymptotically more than a $\log$ size chunk of the \rightcomp/ module's input.  It is also important to note the languages being decided within the \leftcomp/ are recognized in almost everywhere worst case space complexity which is accounted for by the spacing columns of the \planter/, guaranteeing that for all but a finite number of computations, the \leftcomp/ will be able to successfully complete its computations.  It will prematurely abort any computations which attempt to run beyond those space bounds, but since there are guaranteed to be only a finite number of those, the goals of the construction and correctness of the proof aren't compromised.  It is important that these essentially arbitrarily tight bounds on the space complexities of languages is shown to be possible by Theorem 4.1 of \cite{VeryHardLanguages}, which allows for the computations embedded within the modules to be designed with great precision.  In a similar manner, the computations performed by the upper portion of the \topcomp/ module require space complexity greater that that available to either the \planter/ or \leftcomp/ of the same subiteration.  We note that
\iffull
Lemma~\ref{lem:noCheating}
\else
Lemma~$9.14$ (see \cite{DirectedNotIUArxiv})
\fi
is instrumental in proving the above facts, and is also an important tool which can be used in future simulation-based results in the aTAM, as it proves that an assembly performing a simulation of a system growing in a zig-zag manner, despite its arbitrarily large (but constant) scale factor, has asymptotically no greater space resources available than the orignal system.  The technical tools we have developed for this proof, as well as the incorporation of results from complexity theory allowing for precisely defined languages in terms of space complexity, provide a host of new construction and proof techniques which we feel will be useful for a variety of future results.

To prevent the simulator from being seeded with answers to the necessary computations, the assembly of $\calT$ must grow infinitely many iterations and subiterations.  To prevent other types of ``cheating'', rather than having potential locations of cooperation across a single gap between two tiles, the \bitalley/ becomes arbitrarily long, between an arbitrarily large set of pairs of tiles.  To prove all of the necessary properties of the simulator $\calS$ requires many more details and the use of several additional technical lemmas which may possibly be of independent interest and utility.  Please see
\iffull
Section~\ref{sec:proof-details} and Section~\ref{sec:tech-lemmas}
\else
\cite{DirectedNotIUArxiv}
\fi
for full details.

\iffull

\section{Formal description of the abstract Tile Assembly Model}
\label{sec-tam-formal}

In this section we provide a set of definitions and conventions that are used throughout this paper.

We work in the $2$-dimensional discrete space $\Z^2$. Define the set
$U_2 = \{(0,1), (1,0), (0,-1), (-1,0)\}$ to be the set of all
\emph{unit vectors} in $\mathbb{Z}^2$.
We also sometimes refer to these vectors by their
cardinal directions $N$, $E$, $S$, $W$, respectively.
All \emph{graphs} in this paper are undirected.
A \emph{grid graph} is a graph $G =
(V,E)$ in which $V \subseteq \Z^2$ and every edge
$\{\vec{a},\vec{b}\} \in E$ has the property that $\vec{a} - \vec{b} \in U_2$.

Intuitively, a tile type $t$ is a unit square that can be
translated, but not rotated, having a well-defined ``side
$\vec{u}$'' for each $\vec{u} \in U_2$. Each side $\vec{u}$ of $t$
has a ``glue'' with ``label'' $\textmd{label}_t(\vec{u})$--a string
over some fixed alphabet--and ``strength''
$\textmd{str}_t(\vec{u})$--a nonnegative integer--specified by its type
$t$. Two tiles $t$ and $t'$ that are placed at the points $\vec{a}$
and $\vec{a}+\vec{u}$ respectively, \emph{bind} with \emph{strength}
$\textmd{str}_t\left(\vec{u}\right)$ if and only if
$\left(\textmd{label}_t\left(\vec{u}\right),\textmd{str}_t\left(\vec{u}\right)\right)
=
\left(\textmd{label}_{t'}\left(-\vec{u}\right),\textmd{str}_{t'}\left(-\vec{u}\right)\right)$.

In the subsequent definitions, given two partial functions $f,g$, we write $f(x) = g(x)$ if~$f$ and~$g$ are both defined and equal on~$x$, or if~$f$ and~$g$ are both undefined on $x$.

Fix a finite set $T$ of tile types.
A $T$-\emph{assembly}, sometimes denoted simply as an \emph{assembly} when $T$ is clear from the context, is a partial
function $\pfunc{\alpha}{\Z^2}{T}$ defined on at least one input, with points $\vec{x}\in\Z^2$ at
which $\alpha(\vec{x})$ is undefined interpreted to be empty space,
so that $\dom \alpha$ is the set of points with tiles.  An assembly is {\it $\tau$-stable}
if it cannot be broken up into smaller assemblies without breaking bonds of total strength at least $\tau$, for some $\tau \in \mathbb{N}$.
We write $|\alpha|$ to denote $|\dom \alpha|$, and we say $\alpha$ is
\emph{finite} if $|\alpha|$ is finite. For assemblies $\alpha$
and $\alpha'$, we say that $\alpha$ is a \emph{subassembly} of
$\alpha'$, and write $\alpha \sqsubseteq \alpha'$, if $\dom \alpha
\subseteq \dom \alpha'$ and $\alpha(\vec{x}) = \alpha'(\vec{x})$ for
all $x \in \dom \alpha$.

We now give a brief formal definition of the aTAM. See \cite{Winf98,RotWin00,Roth01,jSSADST} for other developments of the model.  Our notation is that of \cite{jSSADST}, which also contains a more complete definition.

Self-assembly begins with a {\it seed assembly} $\sigma$ and
proceeds asynchronously and nondeterministically, with tiles
adsorbing one at a time to the existing assembly in any manner that
preserves $\tau$-stability at all times.  A {\it tile assembly system}
({\it TAS}) is an ordered triple $\mathcal{T} = (T, \sigma, \tau)$,
where $T$ is a finite set of tile types, $\sigma$ is a seed assembly
with finite domain, and $\tau \in \N$.  A {\it generalized tile
assembly system} ({\it GTAS})
is defined similarly, but without the finiteness requirements.  We
write $\prodasm{\mathcal{T}}$ for the set of all assemblies that can arise
(in finitely many steps or in the limit) from $\mathcal{T}$.  An
assembly $\alpha \in \prodasm{\mathcal{T}}$ is {\it terminal}, and we write $\alpha \in
\termasm{\mathcal{T}}$, if no tile can be $\tau$-stably added to it. It is clear that $\termasm{\mathcal{T}} \subseteq \prodasm{\mathcal{T}}$.

An assembly sequence in a TAS $\mathcal{T}$ is a (finite or infinite) sequence $\vec{\alpha} = (\alpha_0,\alpha_1,\ldots)$ of assemblies in which each $\alpha_{i+1}$ is obtained from $\alpha_i$ by the addition of a single tile. The \emph{result} $\res{\vec{\alpha}}$ of such an assembly sequence is its unique limiting assembly. (This is the last assembly in the sequence if the sequence is finite.) The set $\prodasm{T}$ is partially ordered by the relation $\longrightarrow$ defined by
\begin{eqnarray*}
\alpha \longrightarrow \alpha' & \textmd{iff} & \textmd{there is an assembly sequence } \vec{\alpha} = (\alpha_0,\alpha_1,\ldots) \\
                               &              & \textmd{such that } \alpha_0 = \alpha \textmd{ and } \alpha' = \res{\vec{\alpha}}. \\
\end{eqnarray*}
If $\vec{\alpha} = (\alpha_0,\alpha_1,\ldots)$ is an assembly sequence in $\mathcal{T}$ and $\vec{m} \in \mathbb{Z}^2$, then the $\vec{\alpha}$\emph{-index} of $\vec{m}$ is $i_{\vec{\alpha}}(\vec{m}) = $min$\{ i \in \mathbb{N} | \vec{m} \in \dom \alpha_i\}$.  That is, the $\vec{\alpha}$-index of $\vec{m}$ is the time at which any tile is first placed at location $\vec{m}$ by $\vec{\alpha}$.  For each location $\vec{m} \in \bigcup_{0 \leq i \leq l} \dom \alpha_i$, define the set of its input sides IN$^{\vec{\alpha}}(\vec{m}) = \{\vec{u} \in U_2 | \mbox{str}_{\alpha_{i_{\alpha}}(\vec{m})}(\vec{u}) > 0 \}$.

We say that $\mathcal{T}$ is \emph{directed} (a.k.a. \emph{deterministic}, \emph{confluent}, \emph{produces a unique assembly}) if the relation $\longrightarrow$ is directed, i.e., if for all $\alpha,\alpha' \in \prodasm{T}$, there exists $\alpha'' \in \prodasm{T}$ such that $\alpha \longrightarrow \alpha''$ and $\alpha' \longrightarrow \alpha''$. It is easy to show that $\mathcal{T}$ is directed if and only if there is a unique terminal assembly $\alpha \in \prodasm{T}$ such that $\sigma \longrightarrow \alpha$.

A set $X \subseteq \Z^2$ {\it weakly self-assembles} if there exists
a TAS ${\mathcal T} = (T, \sigma, \tau)$ and a set $B \subseteq T$
such that $\alpha^{-1}(B) = X$ holds for every terminal assembly
$\alpha \in \termasm{T}$.  Essentially, weak self-assembly can be thought of
as the creation (or ``painting'') of a pattern of tiles from $B$ (usually taken to be a
unique ``color'') on a possibly larger ``canvas'' of un-colored tiles.

A set $X$ \emph{strictly self-assembles} if there is a TAS $\mathcal{T}$ for
which every assembly $\alpha\in\termasm{T}$ satisfies $\dom \alpha =
X$. Essentially, strict self-assembly means that tiles are only placed
in positions defined by the shape.  Note that if $X$ strictly self-assembles, then $X$ weakly
self-assembles. (Let all tiles be in $B$.)

\section{Details of the Directed System $\calT$}\label{sec:the-system}

In this section, we provide details of the construction of $\calT$ as well as explaining its growth and verifying that $\calT \in \frak{D}$, i.e. that $\calT$ is directed.

\subsection{Languages and Turing machines used}

The decidable languages and the Turing machines which decide them and are simulated within the \leftcomp/ and \topcomp/ modules of $\calT$ are defined as follows.

\begin{enumerate}
    \item[] Let $a = 2^{2t}$ where $t = |U|$ is the size of the simulator's tile set
    \item[] Let $L_A \subset \mathbb{N}$ be a decidable language such that $L_A$ can be decided in almost everywhere worst-case space complexity $2^{n/2}$ and $L_A$ cannot be recognized in infinitely often best-case space complexity $(1/2)2^{n/2}$, i.e. for all but finitely many $n$, $L_A$ requires space greater than $(1/2)2^{n/2}$. Let $A$ be a deterministic Turing machine which decides $L_A$ within space $2^{n/2}$ almost everywhere, i.e. for all but finitely many $n$.  Note that such an $L_A$ is guaranteed to exist by Theorem 4.1 of \cite{VeryHardLanguages}.
    \item[] Let $L_{A_H} \subset \mathbb{N}$ be a decidable language such that $L_{A_H}$ can be decided in almost everywhere worst-case space complexity $2^n$ and $L_A$ cannot be recognized in infinitely often best-case space complexity $(1/2)2^n$, and $A_H$ be a deterministic Turing machine which decides $L_{A_H}$ within space $2^n$ almost everywhere.  (The existence of this language is also guaranteed by Theorem 4.1 of \cite{VeryHardLanguages}.)  Note that this means it also uses no more than $2^{2^n}$ time steps almost everywhere.
    \item[] Let $A^+$ be a Turing machine which, on input $x$, does the following. For each $2^x \le y < 2^{x+1}-1$, $A^+$ simulates $A(y)$ and records a $0$ if $A$ rejects, and a $1$ if it accepts. Then, it simulates $A_H(2^{x+1})$, recording a $0$ or $1$, accordingly.  Furthermore, while $A^+$ is performing any simulation of $A(z)$ or $A_H(z)$, it bounds the space used by the computation and if the machine attempts to use $2^{2z}+1$ unique tape cells, it halts that simulation and records a $0$ for it.  Furthermore, it also bounds the time used and if it attempts to use more than $2^{2^{z}}$ time steps, it halts that simulation and records a $0$ for it.  Once the full series of simulations of $A$ and then the one of $A_H$ have completed, $A^+$ halts with a binary sequence of length $2^x$ on its tape, representing the outputs of each of the $2^x$ computations. Furthermore, $A^+$ uses a one-way-infinite-to-the-left tape.
        \item[] Note that given input $x$, the maximum amount of space used by $A^+$ will be used during its computation of $A_H(2^{x+1})$, which is bounded by $A^+$ at $2^{2x+1}$.
        \item[] Furthermore, by the fact that $L_A$ can be decided in almost everywhere worst-case space complexity $2^{n/2}$, and that for any constant $c$ there exists some $i'$ such that for all $i > i'$, $2i > ci'$, then for all computations $A(i)$ beyond those for a constant number of values, $A(i)$ is guaranteed to halt without using space greater than $2^{i}$ or more than $2^{2^i}$ time steps, and thus all computations of $A$ on inputs greater than that will successfully complete without being halted by $A^+$.  An analogous argument can be made for $L_{A_H}$ and $A_H$.
    \item[] Let $L_B \subset \mathbb{N}$ be a decidable language such that $L_B$ can be decided in almost everywhere worst-case space complexity $3^{n/2}$ and $L_B$ cannot be recognized in infinitely often best-case space complexity $(1/2)3^{n/2}$, i.e. for all but finitely many $n$, $L_B$ requires space greater than $(1/2)3^{n/2}$. Let $B$ be a deterministic Turing machine which decides $L_B$ within space $3^{n/2}$ almost everywhere, i.e. for all but finitely many $n$.  Note that such an $L_B$ is guaranteed to exist by Theorem 4.1 of \cite{VeryHardLanguages}.
    \item[] Let $L_{B_H} \subset \mathbb{N}$ be a decidable language such that $L_{B_H}$ can be decided in almost everywhere worst-case space complexity $3^n$ and $L_B$ cannot be recognized in infinitely often best-case space complexity $(1/2)3^n$, and $B_H$ be a deterministic Turing machine which decides $L_{B_H}$ within space $3^n$ almost everywhere.  (The existence of this language is also guaranteed by Theorem 4.1 of \cite{VeryHardLanguages}.)
    \item[] Let $R$ be a Turing machine which uses a one-way-infinite-to-the-right tape and takes as input two binary strings, $i$ and $j$, where $i = |j|$, and performs the following computations:
          \begin{enumerate}
            \item[] $R$ first runs $A^+(\log(i))$ and compares the resulting $i$-bit string to $j$.
            \item[] If those bit strings match on at least one bit, $R$ halts.
            \item[] Otherwise (i.e. when they differ on every bit), for $i \leq k < i + a - 1$, $R$ simulates $B(k)$ and records a $0$ if $B$ rejects $k$, and a $1$ if it accepts. Then it simulates $B_H(i+a-1)$, recording a 0 or 1 accordingly.  Furthermore, while $R$ is performing any simulation of $B(k)$ or $B_H(k)$, it bounds the space used by the computation and if $B(k)$ attempts to use $3^{k^2}+1$ unique tape cells (i.e. space equal to $3^{k^2}+1$), it halts that simulation and records a $0$ for it.
            \item[] Once the full series of simulations of $B$ have completed, $R$ halts with a binary sequence of length $a$ on its tape, representing the outputs of each of the $a$ computations.
            \item[] Note that given inputs $i$ and $j$, the maximum amount of tape cells used by $R$ will be used during the computation of $B_H(i+a-1)$, which is bounded by $R$ at $3^{(i+a-1)^2}$.
            \item[] Furthermore, by the facts that $L_{B_H} \in \DSPACE(3^n)$ and $L_B \in \DSPACE(3^{\frac{n}{2}})$,  and that for any constant $c$ there exists some $i'$ such that for all $i > i'$, $i^2 > ci$, then for all computations $B(i)$ beyond those for a constant number of values, $B_H(i)$ is guaranteed to halt without using space greater than $3^{|i|^2}$, and thus all computations of $B$ or $B_H$ on inputs greater than that will successfully complete without being halted by $R$.
          \end{enumerate}
\end{enumerate}

Let $\mathcal{T} = (T,\sigma,2)$ be the directed aTAM system which self-assembles the infinite shape sketched in Figure~\ref{fig:counter-and-wedge-TMs-overview}. We will discuss the assembly it produces in a modular way. The seed $\sigma$ consists of a single tile placed at the origin. From the east side of the seed, the ``\planter/'' module forms.

\subsection{\planter/}

The \planter/ module does the following (and is conceptually somewhat similar to the ``\planter/'' module discussed in Section 4.5 of \cite{jCCSA}).  It is in general a log-height binary counter (i.e. a binary counter which represents consecutive numbers by bit strings in consecutive columns, and each column having height equivalent to the number of bits in the number being represented by the counter) which enumerates the positive integers, counting from $1$ to $\infty$.  For each positive integer $i \in \mathbb{Z}^+$, it initiates the growth of an \emph{iteration}, which consists of $2^i$ \emph{subiterations}, each of which initiates growth of a set of modules which grow from the north side of the \planter/.  The function of the \planter/ is to correctly space out those modules by putting them at well-defined locations, as well as to provide input to each via north-facing glues of rows of tiles which initiate the growth of each module.

To do this, the \planter/ actually contains a series of embedded counters, with the counter for the iterations, $i$, being the ``outer'' counter.  For each value of $i$, before the \planter/ increments $i$ again, it begins a nested counter $j$ which it iterates over the values $0 \le j < 2^{2^{\log(i)}} = 2^i$\footnote{Here and throughout the paper, we will use the shorthand $\log(x)$ to mean $\lceil \log(x) \rceil$ and also use the shorthand $x = 2^{\log(x)}$ despite the fact that $x \le 2^{\lceil \log(x) \rceil}$, as it will not impact the correctness of our arguments.}.  (During the columns where the value of $i$ is not incremented, its bit values are simply passed forward, i.e. to the right, unchanged.)  For each value of $j$ it will initiate growth of a subiteration by incorporating additional counters used to guarantee correct spacing of modules.  Therefore, there are also counters nested within each subiteration $j$.  The first of these is a counter which counts from $0$ to $2^{2i}$, by first computing $2^{2i}$ (simply by starting from the binary string $1$ and while counting from $0$ to $i$ adds another bit position for each count, e.g. $10$, $100$, etc.), and then starting a counter at $0$ which increments each column while passing the value $2^{2i}$ through to the right and checking each counted value until it matches $2^{2i}$ at which point it halts.  The horizontal distance grown by that counter is used to create enough space for the next module whose growth will be initiated on the north side of the \planter/ at this point.

While continuing to grow to the east, the \planter/ now computes the value $\log(i)$ (by simply counting the length of the binary representation of $i$) and rotates a copy of the value $\log(i)$, whose length is $\log(\log(i))$, northward so that after growing another $\log(\log(i))$ columns, the binary value of $\log(i)$ is represented by the north-facing glues of $\log(\log(i))$ tiles.  This binary representation of $\log(i)$ will serve as the initiation point for the growth of a \leftcomp/ module, to be discussed later.  In addition the \planter/ exposes a $\tau$ strength north-facing glue two tiles to the west of the beginning of the north-facing glues which represent the binary value of $\log(i)$ and also exposes a $\tau$ north-facing glue two tiles to the east of the end of the glues which represent the binary value of $\log(i)$.  These two glues allow for the growth of two tiles which attach via north and south glues which we call a \bumper/. Also, in locations not specifically mentioned, the northern glues of the northernmost tiles of the \planter/ are $\null$.  The \planter/ now grows an additional $2^{i} + \log(i) +11$ columns to the right, at which point it rotates a copy of the binary value of $j$ to the north.  This binary representation of $j$ will serve as the initiation location for a \rightcomp/ module, also to be discussed below.  Next, the \planter/ grows another $7$ spacing columns to the right, and then rotates copies of the binary representations of both $i$ and $j$ to the north, so that the value $i$,$j$ can serve as the initiation for a \topcomp/ module.  Similar to the initiation point for the \leftcomp/ module, we place $\tau$ strength glues around the initiation points of the \rightcomp/ and \topcomp/ initiation point which allow for the growth of \bumper/s.

The last bit of growth for the \planter/ in a subiteration is to grow the spacing rows which are necessary for the \topcomp/ module (since for its growth above the \planter/ it will grow both upward and to the right).  For this, in a manner analogous to the way it grew the $2^{2i}$ spacing rows for the \leftcomp/ module, it now grows $3^{2i}$ spacing rows.  At this point, the \planter/'s growth in relation to the $j$th subiteration of the $i$th iteration is complete.  It now increments the value of $j$ (assuming it is $< 2^i$, or specifically $< 2^{2^{\log(i)}}$), and grows the necessary columns for the next subiteration, or, if $j$ now equals $2^i$, it resets $j$ to $0$ and increments the value of $i$ to begin a new iteration.  This occurs for infinitely many iterations.

The important features of the \planter/ are that (1) it initiates the appropriate growth to allow the modules of each subiteration to grow independently of the \texttt{planter} once it has written the necessary input values for those modules, (2) that in between the inputs for the modules it creates a necessary amount of spacing columns to ensure that no modules will collide with each other due to the space complexities of the languages whose accepting Turing machines are being simulated by each (to be explained below), (3) all growth is performed in an up and down zig-zag manner, with each column completely growing before the column to its right begins growth either at the top or bottom of the column (depending on whether the column is zigging or zagging), and (4) the number of tiles used in each column, i.e. the height of the \planter/, is never more than the minimum needed for each position to represent a single bit of each counter value being propagated through.  Since the largest value of any such counter is $3^{2i}$, that means that the maximum height of the counter during any iteration $i$ is $\log(3^{2i}) = O(2i)$.  See Figure~\ref{fig:counter-and-wedge-TMs-iteration} for a sketch of the formation of an iteration.

\subsection{\leftcomp/}
For all values $i$, the \leftcomp/ module receives the input $\log(i)$, whose width is $\log(\log(i))$.  The \leftcomp/ module performs a simulation of $A^+(\log(i))$ by growing a zig-zag Turing machine which begins with a tape whose width is the width of its input plus one, with the extra cell representing a specially marked \texttt{blank} tape cell denoting the current end of the tape.  When, and only when, $A^+$ tries to access the leftmost tape cell, that cell is interpreted as a regular \texttt{blank} and the end marker is moved one position to the left to a tile which grows that column one position to the left.  In this way, the simulation of $A^+$ uses rows whose width are the same as the number of tape cells used by $A^+ + 1$.  Figure~\ref{fig:counter-and-wedge-TMs-iteration} shows how the system $\calT$ simulates the computations that compose $A^+$. The system simulates the computation $A^+$ by passing each computation $A(j)$ which composes $A^+$ three pieces of information via glues exposed on the north of the previous computation $A(j-1)$: 1) the input to the machine $A(j)$, 2) the outputs of the previous computations $A(k)$ for all $k<j$, and 3) the value $2^{\log(i)+1}$. In this way, the machine is able to simulate the computations which compose $A^+$, and halt on the last machine $A(i)$, and then write the outputs of all the machines along the north border of the machine $A(i)$ as shown in Figure~\ref{fig:counter-and-wedge-TMs-iteration}.  In addition, we embed a counter in the \leftcomp/ module which counts the number of total time steps used by all machines.  Since each computation in $A^*$ can use at most $2^{2^{i}}$ time steps, the total time steps used by all machines is at most $i*2^{2^{i}}$.  Consequently, to store the number of time steps used, space $\log(i) + 2^{i}$ is needed to store the values of the counter.  Note that we can embed this in the same cells which are used in the computation of $A^+$.

Once $A^+(\log(i))$ halts, with a bit string of length $i$ (i.e. $2^{\log(i)}$) and the total number of time steps $ts$ output to its north, those output bits are rotated up and to the right, three columns to the right beyond the right side of the computations below, and then a $\log(i) + 2^i$-width counter counts down starting from the value $ts$ until it is a distance $8\times 2^{\log(i)} + 4$ above the planter at which point the output bits of $A^+(\log(i))$ are rotated to the right, one at a time and with $7$ tiles vertically between each.  Figure~\ref{fig:gap-overview} shows an example of the first two bits output by the \leftcomp/ module.  Notice that between each bit that is output there are seven tiles and a tile which is a distance of three tiles away from each output bit grows two tiles via $E-W$ glue attachment which we call a \bumper/.

Recall that the computation $A_H$ in the series of computations performed by $A^+$ requires the most space and it was chosen such that for almost all $x$ it has worst case space complexity $2^{|x|}$.  Consequently, this means that for almost all iterations $i$ the \leftcomp/ module will need at most space $2^{2i}$. Note that the \planter/ was designed to space the modules so that in iteration $i$ the \leftcomp/ module has $2^{2i}$ space available for use.  Consequently, for almost all $i$, all of the computations performed in the \leftcomp/ module will complete without prematurely halting.

\subsection{\rightcomp/}

The \rightcomp/ module receives the input $j$ and simply rotates those bits upward and to the left, and also spaces them out with $7$ tiles in between each bit, along with $7$ before the first bit, so that they are output on the west side of the \rightcomp/ module across $7i$ rows (since $j$ consists of $i$ bits). An example of the first two bits output by the \rightcomp/ module can be seen in Figure~\ref{fig:gap-overview}.

\subsection{\topcomp/}

As soon as the \planter/ completes the formation of the portion of its northern row encoding the input for the \topcomp/ module, namely the encoding of $i$ and $j$ near the eastern side of a subiteration, the \topcomp/ module is able to begin simulating the Turing machine $R$ which is defined above.  Recall that $R$ first runs $A^+(\log(i))$.  As shown in Figure~\ref{fig:counter-and-wedge-TMs-iteration} the \topcomp/ module first simulates the $A^+$ computation on input $\log(i)$ in the same manner as the \leftcomp/ module with the exception that it simulates a TM which computes the $A^+$ computation using a one-way-infinite-to-the-right tape (rather than to the left).  As the \topcomp/ module simulates $A^+$ on input $i$, it also propagates the value of $j$ via glues.

Next, the Turing machine $R$ compares the output of $A^+(i)$ to the value $j$, so we design the \topcomp/ module so that it mimics this behavior.  Once \topcomp/ finishes it's simulation of $A^+$ on input $i$, it then compares the output of this computation to the value $j$.  If they match on any bits, the growth of \topcomp/ terminates.  However, if they differ in every single position, the \topcomp/ module simulates the series of computations $B(k)$ for $i \leq k < i+a-1$ and $B_H(i+a-1)$ in a manner similar to it's simulation of the $A^+$ computation.  After the \topcomp/ completes the simulation of the last machine $B_H$, it outputs the $a$-bit string which is the output of the series of computations to its north. This initiates the growth of the \arm/ module as shown in Figure~\ref{fig:counter-and-wedge-TMs-iteration}.

Recall that the computation $B_H$ in the series of computations requires the most space and it was chosen such that for almost all $x$, it has worst case space complexity $3^{|x|}$.  Consequently, this means that for almost all iterations $i$ the \topcomp/ module will need at most space $3^{2i}$. Note that the \planter/ was designed to space the modules so that in iteration $i$ the \topcomp/ module has $3^{2i}$ space available for use.  Consequently, for almost all $i$, all of the computations performed in the \topcomp/ module will complete.

\subsection{\arm/}

If \topcomp/ initiates the growth of the \arm/ module, there are a possible $2^a$ different binary values written as output by the \topcomp/ computations, and the function of the first row of the \arm/ is to grow to the west across those output bits so that when that completes it has selected one of $2^a$ possible types of \arm/s to grow.  The arm module is essentially first a horizontal counter which grows leftward $7+j+4$ (to pass over the spacing columns between the \rightcomp/ and \topcomp/ modules, the \rightcomp/ module, and end up directly over the center position of the \emph{bit alley} between the \leftcomp/ and \rightcomp/ modules.  Once the \arm/ reaches the end of its leftward growth, it then initiates a downward growing column from its leftmost column.  This downward growing column consists of a single repeating tile type so that the column eventually crashes into the \planter/ between the \leftcomp/ and \rightcomp/ modules of the subiteration.  The tile type of the column is determined by the value written by the \topcomp/ module, and can be any of $2^a$ types.

\subsection{\bitalley/}

Figure~\ref{fig:gap-overview} shows a portion of possible assembly growth between the \leftcomp/ and \rightcomp/ modules of a subiteration, say the $j$th of iteration $i$.  As noted in the above section, the counter which grows from the top of \leftcomp/ causes the growth of tiles as shown in Figure~\ref{fig:gap-overview} which depends on the output of $A^+(i)$. If the bit is a $0$, the rightmost of those tiles is of type $0_L$, else it's of type $1_L$.

The height of \rightcomp/ is $8i+4$, and spaced out similarly to the bits on the east side of the \leftcomp/ module, it has the bits of $j$ presented on its west side.  Also similarly to the opposite side, from each bit a pair of tiles attach with the westernmost being of type $0_R$ or $1_R$, depending on each bit value of $j$.

Since the bit strings exposed by \leftcomp/ and \rightcomp/ are aligned with each other, at each position where a bit on the left matches one on the right, a tile cooperatively binds to the two tiles, either $0_L$ and $0_R$ or $1_L$ and $1_R$.  Without loss of generality, assume such a matching bit position has value $0$.  Then, between the $0_L$ and $0_R$ which are at the same height and separated horizontally by a single space, a tile of type $0_M$ binds ``across the gap.''  Finally, a tile of type $0_B$ attaches to the south of the $0_M$ tile and growth related to this bit is complete.  Such attachments of $0_M$, $0_B$, $1_M$, and $1_B$ tiles occur at each position where bits match across \leftcomp/ and \rightcomp/, and at each position where they differ, the final tile attachments are the $0_L$, $0_R$, $1_L$, and $1_R$ tiles.

Finally, recall that exactly in subiterations where the bits in every position of \leftcomp/ and \rightcomp/ differ, an \arm/ grows, which results in a single-tile-wide column of tiles growing southward from the \arm/, between the $0_L$, $0_R$, $1_L$, and $1_R$ tiles of the \leftcomp/ and \rightcomp/ modules, and crashing into the \planter/.  Since this occurs if and only if every bit position differs, there is no possibility of a cooperatively placed $0_M$ or $1_M$ tile blocking the growth of this column.

\subsection{Summary of computations}

\begin{table}[h]
\centering
\setlength{\tabcolsep}{.5em}
\begin{tabular}{|c|c|c|c|c|}\hline\rowcolor{black!20!white}
Module & Input & Computation & Space complexity & Also able to compute \\\hline
\planter/ & none & count, $\log(i)$ & $O(i)$ & none - insufficient space\\\hline
\topcomp/ & $j$,$i$ & $A^+(\log(i))$,$B(\log(i))$ & $O(3^i)$ & $A^+(\log(i))$, $R(i)$\\\hline
\leftcomp/ & $\log(i)$ & $A^+(\log(i))$ & $O(2^i)$ & none - insufficient input \\\hline
\rightcomp/ & $j$ & $none$ & $O(i)$ & none - insufficient space\\\hline
\end{tabular}
\caption{The computations performed by each module in subiteration i,j. \label{tbl:comps}}
\end{table}

\subsection{The system $\calT$ is directed}
Since the interesting components that compose our system are based off zig-zag systems which are clearly directed, the only potential sources of nondeterminism are 1) the modules which perform computations using too much computational resources and crashing into each other, and 2) the arm growing from the \topcomp/ module causing a race condition to be created between the arm and a tile that is placed cooperatively in the bit-alley.  The first situation is prevented from arising by the counters embedded in the \leftcomp/ and \topcomp/ modules and the appropriate spacing provided by the \planter/.  The second scenario cannot arise due to the fact that the \topcomp/ grows an \arm/ if and only if the output of $A^+(\log(i))$ disagrees on all bits with $j$ which means that no tile can be cooperatively placed in the \bitalley/.  Thus, $\mathcal{T}$ is a directed aTAM system.

\section{Details of Impossibility of Simulation}\label{sec:proof-details}

The proof that $\calS$ does not simulate $\calT$ consists of two main portions, each geared toward showing that it is impossible for modules within a subiteration to receive and utilize any information about the output of the complex computations occurring in the \leftcomp/ or \topcomp/ modules prior to those computations occurring, or outside of the modules performing those computations.  Such information could potentially have allowed $\calS$ to remain directed while accurately simulating $\calT$, but instead the lack of such prior information and the chance to effectively utilize it leads to a contradiction that $\calS$, and thus $U$, exists.

We first show that the probes that the \leftcomp/ and \rightcomp/ modules grow on the sides of the \bitalley/ of each subiteration must be grown, in at least an infinite number of iterations, so that there is nothing unique about those in the empty subiteration.  Intuitively, this means that at least infinitely often the probes grown in subiterations must be ignorant of whether they will need to cooperate across the \bitalley/ and therefore ``attempt'' to grow to positions that leave no more than a single tile wide gap in the \bitalley/ to allow for correct simulation in situations where cooperation will be required across the \bitalley/.

The second main point shows that, given the infinite set of iterations just proven to exist, it must be the case that the bottlenecks created across the \bitalley/ are either not sufficient to allow the necessary \arm/ modules to grow or conversely for the necessary cooperation to occur across the \bitalley/.  Part of this relies on showing that the only way the necessary variety of possible \arm/ modules could grow is for the probes which partially block the \bitalley/ to encode information about at least which subset of possible \arm/ modules the actual \arm/ module to be grown will be selected from, in advance of the \topcomp/ module completing its computations which determine the \arm/ type.  This would allow the probe tiles to assist in the formation of the \arm/ modules, but is shown to be impossible.

The proofs of each of these portions rely upon the fact that if either of those types of information were provided in advance to the growing modules, then it would be possible to construct Turing machines which simulate the assembly of $\calS$, inspect the subassemblies thus created, and utilize that information to solve instances of computations which are known to require more space resources than such Turing machines would be using, providing the necessary contradictions.  Note the tight reliance upon the computational complexities of the corresponding decidable languages and the ability to use the tools we have developed to quantify and bound the computational resources available to the subassemblies performing the computations.  Many of these tools can be found as technical lemmas, with associated proofs, in Section~\ref{sec:tech-lemmas}.

\subsection{Empty subiterations cannot be uniquely marked in advance}

In this section, we show that it is impossible for the \leftcomp/ and \rightcomp/ modules of empty subiterations to ``cheat'' by not growing valid complementary pairs of probes, and we show how that prevents $\calS$ from successfully simulating $\calT$.  We first define some useful terms and show some properties which must be true of the assembly produced by $\calS$.

\begin{definition}
Given a subassembly $\alpha \sqsubset \alpha_U$ which represents the single tile wide vertical portion of an \arm/ in $\calT$, let $\alpha_r \sqsubset \alpha$ be the largest subassembly of $\alpha$ such that, below some initial subassembly $\alpha_0 \sqsubset \alpha$ which occurs at the top of $\alpha$, $\alpha$ consists only of repeated, nonoverlapping copies of $\alpha_r$, one immediately below the other, and the topmost immediately below $\alpha_0$.  Note that the bottommost copy of $\alpha_r$ may be truncated, and if there is no such repeating portion of $\alpha$ then $\dom(\alpha_r) = \emptyset$.
We say that the \emph{shape} of $\alpha$ is $\dom(\alpha_0) \cup \dom(\alpha_r')$ where $\alpha_r' \sqsubset \alpha$ is the topmost copy of $\alpha_r$.
\end{definition}

\begin{claim}\label{clm:bounded-arms}
Let $S$ be the set of all unique shapes of \arm/s in $\alpha_U$.  Then $n_{arms} = |S|$ depends only upon the number of tile types in $U$, $t = |U|$, and the scale factor of the simulation, $m$.
\end{claim}

\begin{proof}
To prove Claim~\ref{clm:bounded-arms}, we utilize Lemma 5.10, the Closed rectangular window movie lemma (CRWML), of \cite{jDuples}.  Note that each \arm/ is a single tile wide in $\calT$, and that in any empty subiteration it is possible for the \topcomp/ and \arm/ to grow before either the \leftcomp/ or \rightcomp/ have even begun growth.  In $\calS$ it must be possible to simulate such an assembly sequence, and that means that during the growth of the macrotiles in $\calS$ representing the \arm/ tiles of $\calT$ in such an assembly, there must be no tiles beyond a single macrotile to either their left or right sides (as those locations must map to empty space and therefore tiles in those regions are considered fuzz, which is not allowed to extend further than a single macrotile away from some macrotile which maps to a nonempty location). %
In such an assembly, the maximum width of any horizontal cut across such an \arm/ subassembly in $\calS$ is $\le 3m$, or the width of $3$ macrotiles.
Let $h = 2(3m!(4t)^{3m})$ be a constant to be explained below. To characterize all possible \arm/ shapes, we iterate over every possible configuration of tiles from $U$ which form any subset of a $3m$-tile-wide line, and for each simulate all assembly sequences which place tiles only to the south of that line, until they reach a distance $2h$ or they produce an assembly to which no additional tiles can attach below that configuration but which doesn't reach a distance of $2h$.  (We don't need to consider any which grow above the line because if that growth eventually influences the assembly below, it must do so via a path containing a tile which crosses that line, and the configuration consisting of the current configuration plus that tile will also be simulated.)  We save the shape $s$ of the \arm/ created by that configuration in $S$ provided the following hold:  (1) all assemblies produced from all possible assembly sequences starting from that configuration grew only within the $3$ macrotile wide space directly below it, and all place the same tiles at all locations, and (2) $s$ is not already in $S$.

A window movie (as defined in \cite{jDuples,IUNeedsCoop}), is a set containing the locations, types, and orderings of arrival of glues along a cut across an assembly.  If we consider windows (i.e. boxes which separate a grid graph into interior and exterior portions) whose top edges cut directly across \arm/ subassemblies and whose other edges do not pass through any portion of an assembly, we note that the number of window movies is $\leq 3m!(4t)^{3m}$ (hence our previous choice of $h$).  If we inspect all \arm/s whose shapes were saved into $S$, we first note that the number which belong to \arm/s which did not reach distance $2h$ cannot be larger than the number of ways to tile the $3m \times 2h$ region using only $t$ tile types, which is a constant dependent only upon $t$ and $m$.  We then inspect each whose \arm/ did grow to a distance of $2h$, and we note that by the pigeonhole principle, any subassembly representing an \arm/ which longer than $h$ must have two windows $w \neq w'$ such that their window movies are the same (i.e. the same glues arrive along those cuts in the same locations and orders).  If we let $\alpha$ and $\beta$ of the CRWML both be such an \arm/ assembly, and the windows be $w$ and $w'$, the CRWML tells us that the subassembly of the \arm/ between the first and second identical window movie locations could grow again after the second, and this can be applied an arbitrary number of times to show that that subassembly can be repeated indefinitely (unless blocked by some other subassembly).  (This repeating subassembly corresponds to the $\alpha_r$ of the definition of the shape of an \arm/.)  Because the number of window movies possible is determined only by $t$ and $m$, so is the number of repeating subassemblies corresponding to $\alpha_r$ in the definition of the shape of an \arm/.  Since the prefix $\alpha_0$ is also bounded by height $h$, there are a fixed number of possible prefixes, again only dependent upon $t$ and $m$, so the total number of unique shapes is also bounded by a constant dependent only upon $t$ and $m$.
\end{proof}

\begin{definition}
The set of probes which grow on the left side of the \rightcomp/ module of each subiteration can be compared with every \arm/ shape to generate a subset of \arm/ shapes which would not be geometrically blocked by those probes.  That is, assuming that the macrotiles representing all probes have fully completed growth, those \arm/s which (starting from a horizontal offset relative to those probes which would be the same as if they had grown from a \topcomp/ module) can grow downward past the probes without any collisions occurring, meaning that they would never be able to place a tile, in the absence of the probes, which disagrees with a tile placed by a probe.  We call these subsets of \arm/s the \emph{probe-dodgers} of that probe set.
\end{definition}

\begin{definition}
Given an assembly $\alpha \in \prodasm{S}$ containing the maximum number of tiles which can be placed in subassembly $i_j$ (i.e. the $j$th subiteration of the $i$th iteration) without growing above the southernmost row of macrotiles of \leftcomp/ or any tiles of \topcomp/, the \emph{signature} of a subiteration is the combination of the full specification of the macrotiles representing the fully grown \bumper/s to the left and right of the \leftcomp/ module, plus the full specification of the row of macrotiles below and between them, plus the full definition of the probe-dodgers set of that subiteration.
\end{definition}

A signature can be determined for a subiteration by growing the \planter/ module to the right beyond the \rightcomp/ module, the full \rightcomp/ module, and as many tiles as possible between the bumpers surrounding the \leftcomp/ module without growing above the first row of macrotiles in the \leftcomp/ module.  This can be always be done because $\calS$ must be able to follow the dynamics of $\calT$, in which there are assembly sequences which do exactly this.

\begin{lemma}\label{lem:noSneak}
For each iteration $i$ of an infinite set of iterations, the empty subiteration $i_j$ of that iteration must have a signature which is identical to that of another subiteration $i_k$, for $j \neq k$, of that iteration.
\end{lemma}

\begin{proof}
We will prove Lemma~\ref{lem:noSneak} by contradiction.  Therefore, assume that for no more than a constant number $c$ of iterations do the iterations of $\calS$ have empty subiterations which have signatures identical to some other in their iteration.  Thus, for each iteration $i > c$, the empty subiteration has a signature unique among all others in its iteration.

By Lemma~\ref{lem:independent-subassemblies}, we know that the full \leftcomp/ module can grow from tiles which form on paths only from its signature, since the locations of the signature would be analogous to $\gamma_M$ of the Lemma, and this includes \bumper/s and it must be possible to grow \leftcomp/ completely with or without the \bumper/s, thus any additional paths of tiles that could contribute to the growth of the \leftcomp/ would have to go through the \bumper/s and thus by growing the \bumper/s the paths which would have grown through them can be continued toward the \leftcomp/ module, allowing it to fully grow.

We now calculate the number of unique signatures in subiteration $i_j$.  To fully specify the path between the macrotiles representing the bumpers to the left and right of the \leftcomp/ module, we first note that that can consist of a maximum of $\log(\log(i))$ macrotiles to represent the first row of the \leftcomp/ module plus another constant number of macrotiles to specify the remaining macrotiles on that path, for a total of $O(\log(\log(i)))$ macrotiles.  The number of unique ways to fully specify the entire contents of $O(\log(\log(i)))$ macrotiles (which is greater than or equal to the number of ways to fully specify just the northern row), is $O(t^{m^2\log(\log(i))}) = O(\log^c(i))$.  Additionally, we note that $n_{arms}$ is a constant independent of $i$ (by Lemma~\ref{clm:bounded-arms}), and thus that is also true for the size of the power set of all \arm/ shapes, which represents the full set of possible probe-dodger sets. This means that the number of unique probe-dodger sets is constant relative to $i$, and therefore adding in specification of one of the constant number of probe-dodger set only allows for a constant multiplier for the number of unique signatures, resulting in a total of $O(\log^c(i))$ possible unique signatures for subiterations of iteration $i$.

Given the total of $O(\log^c(i))$ possibly unique signatures, we can apply Observation~\ref{obs:tech-tm} where $|E| = O(\log^c(i))$ and $|x| = \log(\log(i))$ and note that it shows that $F(n) \in \Omega(2^n)$.  However, to calculate the signatures of all subiterations of iteration $i$ requires only that we simulate the planter in such a way as to remember only the most recent two columns at any given time, requiring a maximum space $O(i)$ (i.e. bounded by its maximum height), and also to record the signatures of the unique subiterations, requiring $O(\log^c(i)^2) < 2^i$.  However, this contradicts the $F(n)$ of Observation~\ref{obs:tech-tm}, and therefore it must be the case that for an infinite number of iterations, the empty subiterations of those iterations have signatures which are identical to those of at least one other subiteration in their respective iterations.

\end{proof}

The following lemma states that for the empty subiteration, the $\arm/$ that assembles between the macro bit-alley must have a ``pinch point''. 

\begin{lemma}\label{lem:2cases}
For $i,j,k\in\N$ and some iteration $i$, suppose that the empty subiteration $i_j$ and a distinct subiteration $i,k$ have identical signatures. Additionally, let $\arm/_{i,j}$ denote the $\arm/$ that assembles in the macro bit-alley of subiteration $i,j$, and let $P_k$ be the left probes of $\leftcomp/^U_{i,k}$ and $P'_k$ be the right probes of $\leftcomp/^U_{i,k}$. Then,
there exists subconfigurations $P_j$ and $P'_j$ of $\arm/_{i,j}$ such that $P_j$ is congruent to $P_k$ and $P'_j$ is a congruent subconfiguration of $P'_k$. 
Additionally, the subconfiguration $C$ of $\arm/_{i,k}$ corresponding to $\arm/_{i,k}$ restricted to $\dom(\arm/_{i,k})\setminus (\dom(P_k) \cup \dom(P'_k))$ has the property that there exists a single tile $t_k$ in $C$ such that removing this tile from $C$ partitions $C - t_k$ into two nonempty sets of tiles such that no two tiles of these sets are adjacent. 
\end{lemma}

\begin{proof}

First, by Lemma~\ref{lem:independent-subassemblies}, the fact that the left computation of subiteration $i,j$ and subiteration $i,k$ have bumpers on the left and right side, and the assumption that these subiteration have identical signature, it follows that $P_j$ is congruent to $P_k$. Then, as both subiterations have the same set of probe-dodgers and this set must be non-empty since the $\arm/_{i,j}$ must assemble in the $i,j$ subiteration, it must be the case that $P'_j$ is a congruent subconfiguration of $P'_k$.

Finally, by Lemma~\ref{lem:gap-must-exist}, it must be the case that the gap between $P_k$ and $P'_k$ in $i_k$ must be single tile wide or less in order for the probes in $P_k$ and $P'_k$ to cooperatively place a tile. Therefore, the subconfiguration $C$ of $\arm/_{i,k}$ corresponding to $\arm/_{i,k}$ restricted to $\dom(\arm/_{i,k})\setminus (\dom(P_k) \cup \dom(P'_k))$ has the property that there exists a single tile $t_k$ in $C$ such that removing this tile from $C$ partitions $C - t_k$ into two nonempty sets of tiles such that no two tiles of these sets are adjacent. 

\end{proof}

\subsection{Turing machines simulating tile assembly systems}
In this section we prove a couple of claims on the amount of space a Turing machine requires to simulate a system which grows certain modules of the system $\calT$.

We call a tile an $L$ tile if it is of type $0L$ or $1L$.  We call a tile an $R$ tile if it is of type $0R$ or $1R$.  We define the bit-alley region of a subiteration $i,j$ to be the points which lie in between $\leftcomp/_{i,j}$ and $\rightcomp/_{i,j}$ (that is, the points which lie on the same row as a tile in $\leftcomp/_{i,j}$ and $\rightcomp/_{i,j}$) and have the same x-coordinates as points which lie between the $L$ and $R$ tiles in subiteration $i,j$.  We define the macro bit-alley of an iteration $i, j$ to be the macrotile equivalent of the bit-alley region.  Let $(x,y)$ be the bottom leftmost corner of a macrotile location which lies between an $L$ and $R$ macrotile.  We call the region $R=\{(x',y') | x-2 \leq x' \leq x+c+2, y-2 \leq y' \leq y+c+2\}$ a probing region.  Let $\alpha' \in \prodasm{\calT}$.  We say that a module $\gamma \sqsubseteq \alpha$ is not assembled in $\alpha'$ if $\dom(\gamma) \cap \dom(\alpha') = \emptyset$.  Now, let $\alpha' \in \prodasm{\mathcal{U}}$.  Similarly, we say that a macro module $\gamma \subseteq \alpha_U$ is not assembled in $\alpha'$ if the module is not assembled in the assembly $R^*(\alpha')$.

Let $\calT$ be TAS and let $\alpha \prodasm{\calT}$.  We say that a subconfiguration $\gamma' \sqsubseteq \alpha$ grows from a subconfiguration $\gamma \sqsubseteq \alpha$ provided that there exists a path in the binding graph $G_{\alpha}$ from a tile in $\gamma$ to all tiles in $\gamma'$.  Here, $\gamma$ and $\gamma'$ are assumed to be connected.  Also, we say that $\gamma$ and $\gamma'$ are a distance of at most $1$ apart if there exist tiles $t \sqsubseteq \gamma$ and $t' \sqsubseteq \gamma'$ such that the Manhattan distance between $t$ and $t'$ in $G_{\alpha}$ is at most $1$.  Otherwise we say that the distance between $\gamma$ and $\gamma'$ is greater than 1.

\begin{claim} \label{claim:TMprobes}
Let $i, j \in \mathbb{N}$.  The subconfigurations grown in the macro bit-alley from $\leftcomp/^U_{i,j}$ and $\rightcomp/^U_{i,j}$ can be output by a Turing machine $M'$ which runs in space $2^{i+1}$.
\end{claim}

In this proof we rely on a straight forward adaptation of Lemma~\ref{lem:noCheating}.  The straight forward adaptation of the lemma we discuss holds because of the key insight in the proof of Lemma~\ref{lem:noCheating} that in order for a Turing machine to simulate a system which is simulating a zig-zag system, it only needs to ``remember'' a bounded number of tiles depending on the scale factor of the simulation and the width of the system that is being simulated.  This key insight allows a Turing machine to not only output the result of a computation that takes place in the simulating system, but it also allows us to construct a Turing machine which outputs the configurations contained in certain macrotile regions of the producible assemblies of the simulating system.

\begin{proof}
Note that in $\calT$ the planter grows in a zig-zag fashion.  Consequently, it follows from a straight forward adaptation of Lemma~\ref{lem:noCheating} that the configuration of the row of macrotiles which compose the first rows of the $\leftcomp/^U_{i,j}$ and $\rightcomp/^U_{i,j}$ can be determined in space $O(i^2)$ since the width of the counter in $\calT$ is $O(i^2)$.  %

\begin{figure}[htp]
\begin{center}
\includegraphics[width=2.0in]{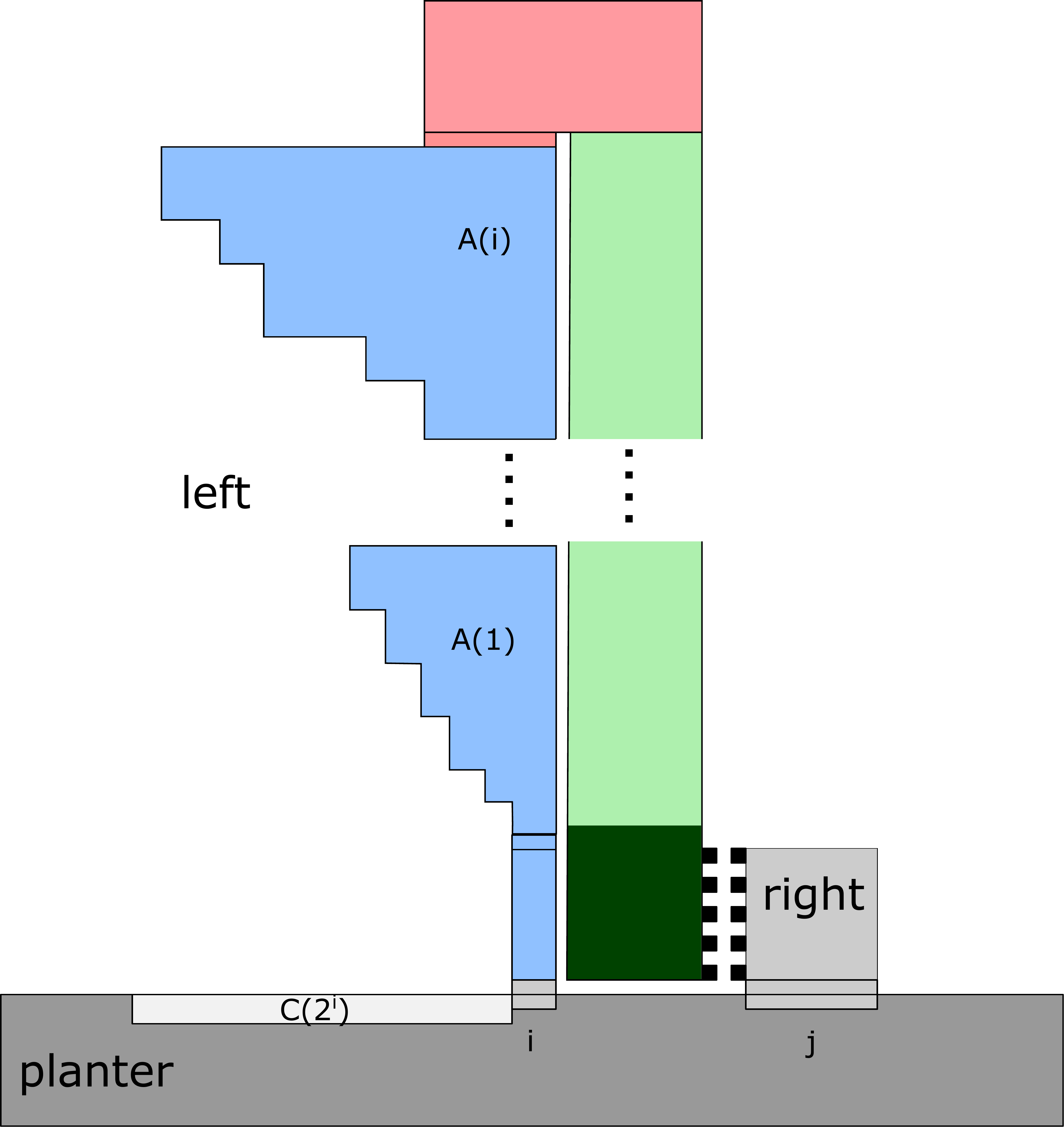}
\caption{A depiction of a portion of the $\planter/$, $\leftcomp/_{i,j}$, and $\rightcomp/_{i,j}$ of the configurations of the $(i,j)$-subiteration in $\calT$. The blue, red, light green, and dark green regions of $\leftcomp/_{i,j}$ correspond to various subassemblies of $\leftcomp/_{i,j}$. The dark green region is an $i\times i$ block of tiles which can be determined in $O(2^i)$ space.}
\label{fig:nocheating-lemma-adaptation}
\end{center}
\end{figure}

Figure~\ref{fig:nocheating-lemma-adaptation} shows a portion of the $\planter/$, $\leftcomp/_{i,j}$, and $\rightcomp/_{i,j}$. The configuration $\leftcomp/_{i,j}$ has been divided into five regions.
\begin{enumerate}
\item The blue region corresponds to the zig-zag assembly which computes $A^+(i)$,
\item the red region corresponds to the assembly which turns the $i$ bits calculated by $A^+(i)$ toward the counter,
\item the light green region corresponds to the zig-zag counter assembly which grows toward the $\planter/$, counting to $2^{2^i}$ before growing the probes of the left side of the bit-alley, and
\item the dark green region which corresponds to the subassembly containing the left portion of the bit-alley.
\end{enumerate}

Now consider the subconfiguration $L$ in $\mathcal{U}$ which represents the subassembly corresponding to the dark green region in Figure~\ref{fig:nocheating-lemma-adaptation}. A straightforward adaptation of the proof of Lemma~\ref{lem:noCheating} shows that $L$ can be determined in $O(2^i)$ space. This adaptation consists of modifying the proof so that the glue sequence tables and the assemblies produced by the procedures InitAssembly, InitGST, UpdateAssembly, and UpdateGST follow the zig-zag assembly sequence of the subassembly in the blue region before turning and following the zig-zag assembly sequence of the light green region. Both of these zig-zag assemblies have width in $O(2^i)$. Consequently, there exists a Turing machine which runs in space $O(2^i)$ and outputs the configuration in the dark green region of $\leftcomp/^U_{i,j}$.  Then, as $\rightcomp/^U_{i,j}$ consists of $O(i^2)$ tiles, the claim holds.
\end{proof}

\begin{claim}
Let $i,j \in \mathbb{N}$ be such that $i, j$ is an empty subiteration in $\alpha_U$.  Let $\alpha_U' \in \prodasm{mathcal{U}}$ be such that the $\leftcomp/^U_{i,j}$ and $\rightcomp/^U_{i,j}$ of subiteration $i, j$ is assembled, but the $\topcomp/^U_{i,j}$ module has not been assembled.  Then if $\alpha_U'$ contains configurations $P$ and $P'$ which grow from $\leftcomp/^U_{i,j}$ and $\rightcomp/^U_{i,j}$ respectively in the macro bit-alley of subiteration $i,j$ which are at most a distance of 1 apart, then there exists a Turing machine $M'$ which takes $\alpha_U'$ as input and outputs a set of $t$ arm types, denoted $A'$, such that the arm which grows in subiteration $i,j$ in $\alpha_U$ is in the set $A'$.  Furthermore, $M'$ runs in space $O(|U|\times c^2)$.
\end{claim}

\begin{proof}
Let the hypotheses hold, and assume that $\gamma$ is the subconfiguration grown in the probing region $R$ which contains $P$ and $P'$.  By assumption there exists subconfigurations $P$ and $P'$ which are a distance of 1 apart.  We can use these two subconfigurations to construct $|U|$ different systems $\calT_t = (U, \sigma_t, \tau)$.  For each $t \in U$ the system $\calT_t$ is constructed by constructing $\sigma_t$ so that it consists of $\gamma$ with the tile type $t$ placed in the single tile wide gap between $P$ and $P'$.  Note that there could be more than one single tile wide gap between $P$ and $P'$.  It doesn't matter which one we choose as long as we choose the same one when constructing different systems. We can then simulate the growth of all the systems with a Turing machine $M'$ in the following manner.  For each $\calT_t$, $M'$ simulates the assembly of the system until two full macrotile regions form.  At that point, the Turing machine is then able to determine what tile types the macrotiles map to in $T$.  Observe that since the configuration $\gamma \cup t$ ``cuts'' the bit-alley region, and an arm must be able to grow in the subiteration since it is assumed to be empty, there exists at least one $t$ such that $\gamma \cup t$ which grows into two full macrotile regions and does not place any tiles outside of the bit-alley.  From this, the Turing machine can determine what types of arms are able to grow into the subiteration.  This Turing machine can clearly be designed to run space $O(|U| \times c^2)$.
\end{proof}

\subsection{A Contradiction} \label{sec:contradiction}
Let the $B^+$ machine be defined analogously to the $A^+$ machine described in Section~\ref{sec:the-system}, but for the series of computations $B_k$ for $0 \leq k \leq a-1$ and $B_H$.  Note that there does not exist a machine which outputs $B^+(x)$ and runs in time $G(n)$ for $G(n) \in o(3^n)$ for infinitely many inputs $x\in \mathbb{N}$.  This follows from the description of the languages that the machines $B^+$ decides which are described in Section~\ref{sec:the-system}.

In this section, we show that under the hypothesis there exists a simulator $\mathcal{U}$ for $\calT$ we can construct a Turing machine $M^*$ which outputs $B^+(x)$ and runs in time $G(n)$ for $G(n) \in o(3^n)$ for infinitely many inputs $x\in \mathbb{N}$.  This will contradict the assumption that the language $L_{B_H}$ cannot be recognized in i.o. space complexity $G(n)$ for $G(n) \in o(3^n)$.

For each input $x \in \mathbb{N}$, $M^*$ does the following.  First, $M^*$ simulates the growth of an assembly $\alpha_e \in \prodasm{\mathcal{U}}$ such that $\alpha_e$ only grows the \planter/ and signatures (which means that $\alpha_e$ also contains \rightcomp/ modules) in order to determine whether the empty subiteration in iteration $x$ is unique.  Note that this requires space $O(|x|)$ by Lemma~\ref{lem:noCheating1} since this is the space used by the \planter/ in $\calT$.  If the empty subiteration in iteration $x$ is unique, the machine $M^*$ simply runs the machine $B^+$ on input $x$, outputs $B^+(x)$ and halts.  That is, compared to $B^+(x)$, $M^*$ with input $x$ is no more space efficient in this case.

\subsubsection{Creating a set $E$ of $t$ arm types}
We refer to the tile types of $T$ that assemble the various $\arm/$s in $\calT$ as arm types. By an abuse of notation, we refer to the macrotile that maps to an arm type as an arm type of $\mathcal{U}$. Then, if the empty subiteration is non-unique $M^*$ creates a set $E$ of $t$ arm types such that the arm type which grows into the empty subiteration $i_j$ is guaranteed to be in $E$. It does so in the following manner.  Denote the empty subiteration by $x_j$ and denote the non-empty subiteration with an identical signature by $x_k$.  Next, the TM $M^*$ determines the left and right probes grown by $\leftcomp/^U_{x_k}$ and $\rightcomp/^U_{x_k}$.  It follows from Claim~\ref{claim:TMprobes} that this can be done in space $2^{|x|}$.   It follows from Lemma~\ref{lem:2cases} that any left probe $P$ of $\leftcomp/^U_{x, k}$ and any right probe $P'$ of $\rightcomp/^U_{x, k}$ are such that the $\arm/^U_{i,j}$ is \emph{consistent} with the translation of $P$ and $P'$ by some $\vec{v}$.  That is, they do not have different tiles in the same location after the translation.

\begin{figure}[htp]
\begin{center}
\includegraphics[width=5.0in]{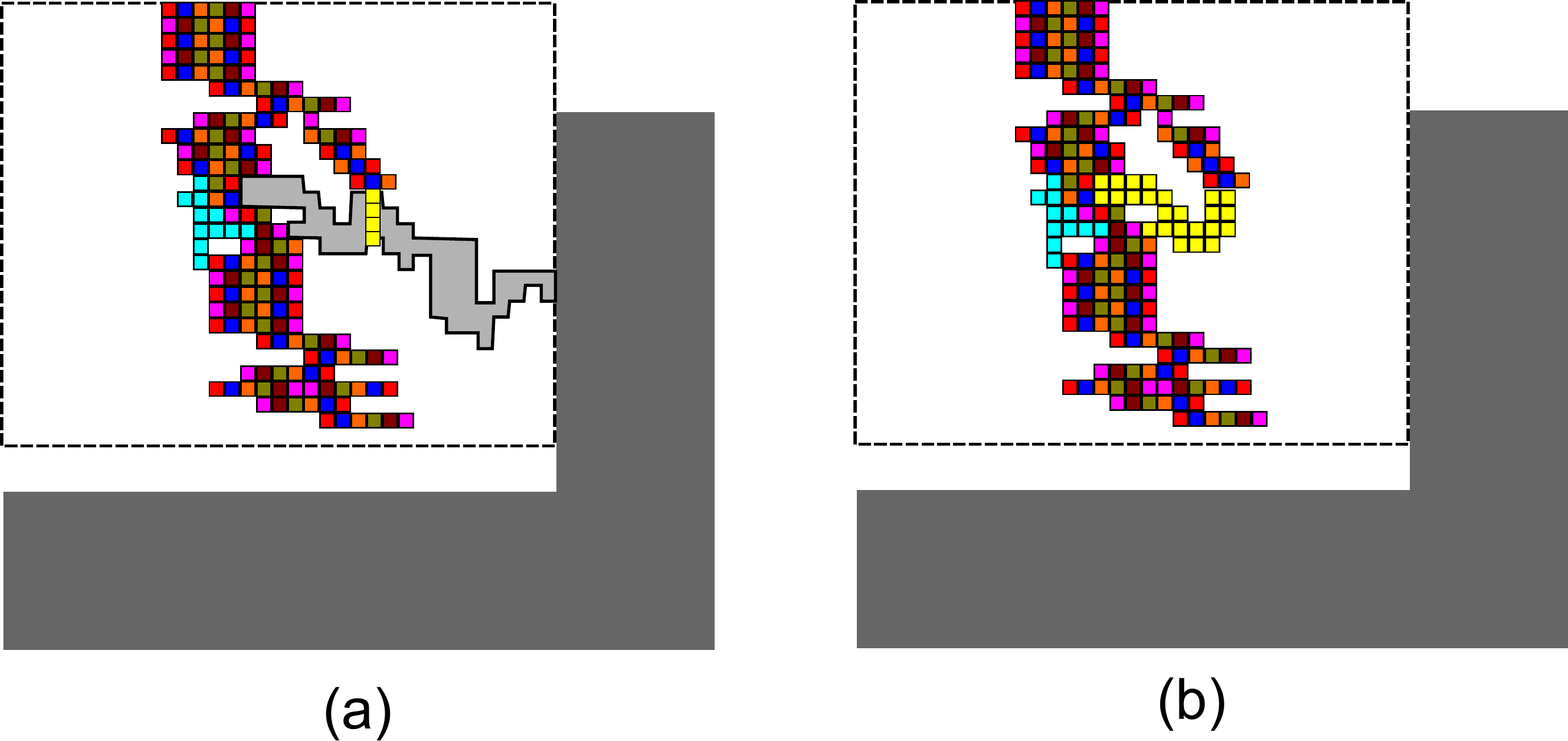}
\caption{If a subconfiguration $\gamma$ (shown as the yellow strip in part (a)) grows a strip of tiles which is consistent with and completely spans another subconfiguration $\gamma'$, then whatever grows in $\gamma'$ after the strip can grow as shown in part (b).}
\label{fig:subcut}
\end{center}
\end{figure}

Note that at this point we know that there is some strip of $\arm/^U_{x,j}$ that contains tiles in $P$ and is completely consistent with the subconfiguration $P'$.  Note that it follows from Lemma~\ref{lem:independent-subassemblies} that if $\arm/^U_{x,j}$ grows a subconfiguration of tiles which completely cuts the region where it must be consistent with $P'$ as shown in Figure~\ref{fig:subcut}  part (a), then it must be the case that it can grow all of $P'$ which grows after the cut as shown in part (b) of Figure~\ref{fig:subcut}.  Furthermore, note that this growth can occur without any tile placements outside of $\dom(P')$.  Consequently this means if we want to create the set $E$, all we need to do is for each tile $u \in U$, create a system which contains only the probe $P$ and the tile $u$ placed at position $p_t$.  We then grow the system until the assembly which we obtain is terminal or the diameter of the assembly we obtain is greater than $2c$.  If the diameter of the assembly is greater than $2c$ and some portion of it maps to an arm tile $t_a$ under the representation function, then we add $t_a$ to our set $E$.  Let $A$ be a subconfiguration and let $R$ be the infinite which consists of infinite columns such that $\dom(A) \subset R$ and if any column is removed from $R$, it is no longer true that $\dom(A) \subset R$.  Then we call the southern boundary of $A$ the set of points $x \in \dom(A)$ such that $x$ contains a path $p$ completely contained in $R$ to a point which lies to the south of any point in $\dom(A)$.  Intuitively, this is the set of tiles on the ``bottom path'' of $A$.  Note this does not just include tiles in $A$ which have an empty location to the south.  If the assembly is terminal and has diameter less than $2c$, we then create a new seed which consists of the previous seed and the tile $t_s$ in $P'$ such that 1) it is not contained in the previous seed, 2) it is one of tiles in the southernmost boundary of $P'$) and 3) the path of points contained in the southern boundary tiles from $t_s$ to $p_t$, denoted $\partial_s$, is such that all the points in $\partial_s$ were in the previous seed.  We then repeat this process for all $u \in U$.

The Turing machines $M^*$ runs the algorithm shown in Algorithm~\ref{alg:arm-types}.  This algorithm is just a formalization of the intuitive idea discussed above.  Here are the variables we use in this algorithm:
\begin{enumerate}
\item $P$ is the left probe grown from $\leftcomp/^U_{x,k}$.
\item $P'$ is the right probe grown from $\rightcomp/^U_{x,k}$.
\item $p_t$ is a point which lies adjacent to both $P$ and $P'$ (i.e. the single tile wide gap between $P$ and $P'$).
\item $P^* = \dom(P) \cup \dom(P') \cup p_t$.
\item $\partial_{P^*}$ is the \emph{southern boundary} of $P^*$.
\item Intuitively $Q$ is the min queue which contains tiles on the southern boundary of $P'$ and they are added to the queue based off of how far they are away from $p_t$.
\end{enumerate}

\begin{algorithm}
\caption{An algorithm for constructing the arm types for a non-unique empty subiteration.~\label{alg:arm-types}}
\SetKwRepeat{Do}{do}{while}%
	\KwIn{$P, P', U, p_t$ as described above}
	\KwOut{A set $E$ of arm types of size at most $|U|$.}
	set $Q$ to be a min queue of tiles in $\partial_{P'}$ where the key for a tile $t'$ is the length of the path from $\dom(t')$ to $p_{r'}$ in the grid graph restricted to only points in $\partial_{P^*}$\;
	\For{$t^* \in U$}{
        set $\calT^*$ to the system $(U, \sigma_t, \tau)$ where $\sigma_t$ is the assembly $P$ with tile $t^*$ placed at point $p_t$\;
	\Do{$Q$ is not empty}{
      grow $\calT$ until the assembly $\alpha$ is terminal or $\diam(\alpha)> 2c$\;
      \If{$\diam(\alpha)> 2c$}{
      	break\;
      }\Else{
      	$\sigma_t = \sigma_t \cup \operatorname{pop}(Q)$ (that is, redefine $\sigma_t$ to be the assembly $P$ with tile $t^*$ placed at point $p_t$ and the tile popped from $Q$ at its tile location\;
      }
    }

        \If{$\alpha$ contains a macrotile which maps to an arm tile, $t_a$ say}{
			add $t_a$ to $E$\;
        }
	
	}
\end{algorithm}

\subsubsection{Using $E$ to compute $B^+(x)$}
Note that the arm tiles in $E$ correspond to the output of the machines $B_i$, so we can think of each arm tile in $E$ as corresponding to a string $x_0x_1x_2...x_{2t-1}$ where $x_i$ represents the output of $M_{B_i}$ on the input received by \topcomp/.  Once $M*$ determines the $t$ potential arms that can grow into the empty subiteration $x, j$, it can create a set of $t$ strings of length $2t$ which correspond to the arms.  Furthermore, it must be the case that one of these $t$ strings corresponds to the output of $B^+$ on input $x$.  Thus, we have a set of $t$ strings of length $2t$ which contains the solution to $B_0(x) B_1(x) \hdots B_H(x)$.  Then by Lemma~\ref{lem:tech-tm-constant}, for almost all inputs, $M^*$ requires at least space $\Omega(3^{\frac{n}{2}})$, but as we observed $M^*$ only uses space $O(2^{n+1})$.  This is a contradiction.

\subsubsection{Proof of correctness for the algorithm which generates $E$}
We begin by noting that the set $E$ produced by Algorithm~\ref{alg:arm-types} is guaranteed to contain a translation of a macrotile grown in $\arm/^U_{x_j}$. The algorithm implicitly described above is guaranteed to grow a macrotile that the arm grows for the following reasons.  Let $t_p$ be the tile such that the arm $\arm/^U_{x_j}$ is consistent with the configuration $P \cup t_p \cup P'$ and let $\partial_{P^*}$ be the southern boundary of this configuration.  We show that there is an assembly sequence in $\mathcal{U}$ such that the tiles placed in $\partial_{P^*}$ before growth continues to the south of $\partial_{P^*}$ is the same as some seed in our algorithm.  Let $P_j$ and $P_j'$ be the subconfigurations grown from $\leftcomp/^U_{x,j}$ and $\rightcomp/^U_{x,j}$ respectively such that $P_j \cup P_j' \congsub P \cup P'$.  First, there is guaranteed to be an assembly sequence where $P_j$ is present since the assumption that subiteration $i_j$ and subiteration $i_k$ have the same signatures implies that there exists an assembly where their $\leftcomp/^U$ modules are exactly the same up to translation.  This means that there exists an assembly where $P_j$ is present before $\arm/^U_{x_j}$ grows into the $\bitalley/^U_{x_j}$.  In addition, there is an assembly sequence where first $P_j$ appears (because of the previous point), and then a translation of $t_p$ appears next in $\partial_P$ so that it prevents the cooperative growth of the macrotile which grew from the probes in the subiteration $x_k$ (otherwise the macrotile grown in the simulated bit-alley of iteration $x_k$ could assemble).  Finally, there exists an assembly sequence where after $\arm/^U_{x_j}$ places a macrotile on the southern boundary of $\dom(P_j')$ all of the tiles in $P_j'$ which lie to the west of the strip can grow with only tile placements in $\dom(P_j')$.

Now, notice that when $U$ and $c$ are fixed, Algorithm~\ref{alg:arm-types} runs in constant time.

Finally, we claim that the algorithm adds at most $t$ arm types to the set $E$.  Indeed, this is true since it is clear from the algorithm that at most one arm type can be added to $E$ with each iteration of the outer loop of which there are $|U|$.

\section{Technical Lemmas}\label{sec:tech-lemmas}
In this section we prove a number of technical lemmas which will be of assistance in later sections.  The first technical lemma we prove shows that in a directed system if a subconfiguration $\gamma$ ``grows through'' another subconfiguration $\gamma'$, then the subconfiguration $\gamma'$ can grow the portion of $\gamma$ that assembles after $\gamma$ ``grows through'' $\gamma'$.  The second technical lemma roughly states that the growth of macrotile that represents a bit-alley tile in $\calT$ must stem from the cooperative placement of tiles by subconfigurations grown from the left and right machines.  In the third lemma, we show that in a system $\mathcal{U}=(U, \sigma, \tau)$ the number of assemblies that can grow from subconfigurations which are exactly the same for all but one tile is no more than $|U|$.  This lemma will be used in a later section to show that the number of arms that can grow into an empty subiteration which has at most a one tile wide gap is at most $|U|$.  The next technical lemma we prove in this section shows that if a Turing machine $M$ is able to narrow down the solution space for the outputs of a series of computations on some input $x$, then $M$ must use the same amount of space as some of the computations in the series.  This lemma will allows us to put constraints on the types of tricks the adversary is able to use in order to simulate the system $\calT$. Finally, we will prove a technical lemma which proves that the space complexity of computations possible within an assembly simulating a zig-zag assembly is asymptotically no greater than the space complexity of the system being simulated.

\subsection{Miscellaneous Definitions}
\begin{definition}
Let $R \subset \mathbb{Z}^2$ be an $m \times n$ rectangular region with the bottom leftmost corner at location $(0,0)$.  Then \\
1) $perim_S = \{x | x = (p, 0), p \in [0, n]\}$,\\
2) $perim_E = \{x | x = (n, p), p \in [0, m]\}$,\\
3) $perim_N = \{x | x = (p, m), p \in [0, n]\}$, and\\
4) $perim_W = \{x | x = (0, p), p \in [0, m]\}$.\\
\end{definition}

\begin{definition}
Let $\calT = (T, \sigma, 2)$ be a directed TAS and let $\alpha \in \termasm{\calT}$.  Let $\gamma \sqsubseteq \alpha$ and $\beta \sqsubseteq \alpha$.  Then we say $\gamma$ and $\beta$ are congruent and write $\gamma \cong \beta$ if $|\dom(\beta)| = |\dom(\gamma)|$ and there exists $\vec{v} \in \mathbb{Z}^2$ such that for all $x \in \dom(\gamma), \gamma(x) = \beta(x+\vec{v})$.
\end{definition}

\begin{definition}
Let $\calT = (T, \sigma, 2)$ be a directed TAS and let $\alpha \in \termasm{\calT}$.  Let $\gamma \sqsubseteq \alpha$ and $\beta \sqsubseteq \alpha$.  Then we say $\gamma$ is a congruent subconfiguration of $\beta$ and write $\gamma \congsub \beta$ if there exists a subconfiguration $\beta' \sqsubseteq \beta$ such that $\gamma \cong \beta'$.
\end{definition}

Let $i, j \in \mathbb{N}$ We call $P$ a left (right) probe if $P$ is grown from $\leftcomp/^S_{i_j}$ ($\rightcomp/^S_{i_j}$) and there exists $P' \sqsubseteq \alpha^S$ such that $P'$ grows from $\rightcomp/^S_{i_j}$ ($\leftcomp/^S_{i_j}$) and $P$ and $P'$ are a distance of $1$ apart.  Throughout this paper, we assume $\alpha^S$ is the single terminal assembly of $\calS$.  

\subsection{Path-crossing subconfigurations}
We now show that in a directed system if a subconfiguration $\gamma$ ``grows through'' another subconfiguration $\gamma'$, then the subconfiguration $\gamma'$ can grow the portion of $\gamma$ that assembles after $\gamma$ ``grows through'' $\gamma'$.

We now begin with some definitions to allow us to more concretely define what it means for a subconfiguration to grow another subconfiguration and what it means for a subconfiguration to grow through another subconfiguration.  As we will see, the intuitive notion of a subconfiguration $\gamma$ growing a subconfiguration $\gamma'$ means that there exists a path in the directed binding graph from $\gamma$ to $\gamma'$.  Also, we will see that the idea of a subconfiguration $\gamma$ growing through a subconfiguration $\gamma'$ means that all paths in the binding graph from $\gamma$ ``cut'' through the subconfiguration $\gamma'$.

\begin{definition}
Let $\calT$ be a directed TAS such that $\vec{\alpha}$ is an assembly sequence of $\calT$.  We now define the directed binding graph of an assembly sequence $\vec{\alpha}$ which we denote by $G_{\vec{\alpha}}$.  The vertices of the directed binding graph $G_{\vec{\alpha}}$ are tiles in $res(\vec{\alpha})$ and there is an edge from tile $t$ to tile $t'$ in $G_{\vec{\alpha}}$ provided that a glue on $t$ serves as an input glue to $t'$.
\end{definition}

\begin{definition}
Let $\calT$ be a directed TAS such that $\vec{\alpha}$ is an assembly sequence of $\calT$.  Also let $\alpha \in \termasm{\calT}$ and let $\gamma \sqsubseteq \alpha$.  Let $p$ be a path from a tile $t$ to a tile $t'$ in $G_{\vec{\alpha}}$ such that $p$ contains tiles that belong to a connected subconfiguration $\gamma$.  Then we say that the path $p$ cuts $\gamma$ provided that the subgraph of $G_{\vec{\alpha}}$ which contains only tiles in $\gamma$ is disconnected when the vertices in $p$ are removed.
\end{definition}

\begin{figure}[htp]
\begin{center}
\includegraphics[width=4.0in]{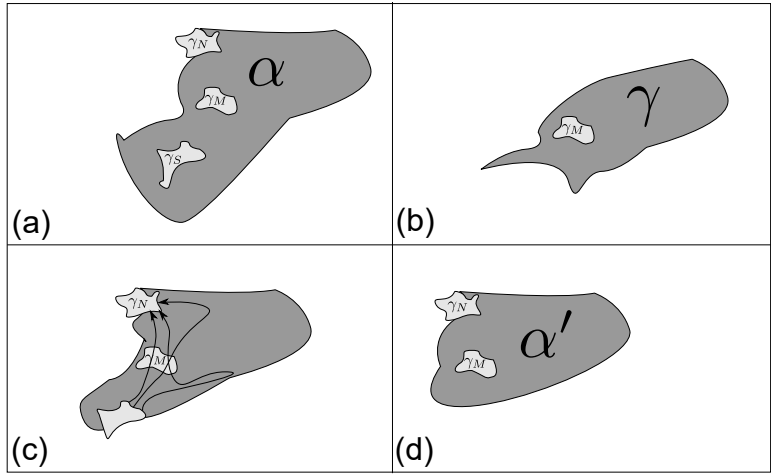}
\caption{A schematic diagram of the assemblies and subconfigurations discussed in Lemma~\ref{lem:independent-subassemblies}.  Part (a) shows the subconfigurations $\gamma_S$, $\gamma_M$, $\gamma_N$ and the terminal assembly $\alpha$.  Part (b) shows a schematic representation of the assembly $\gamma$ from condition (1) of the lemma statement.  Part (c) shows a schematic representation of all the paths in the binding graph of $G_{\vec{\alpha}}$ from tiles in $\gamma_S$ to $\gamma_N$. Part (d) shows the producible assembly $\alpha'$ in the conclusion of the statement.}
\label{fig:digraphCrossing}
\end{center}
\end{figure}

Figure~\ref{fig:digraphCrossing} shows the schematic representation of the conditions listed in the statement of Lemma~\ref{lem:independent-subassemblies} and it's conclusion.  Intuitively, the first condition of Lemma~\ref{lem:independent-subassemblies} states that the growth of $\gamma_M$ is not dependent on the growth of $\gamma_S$.  Consequently, in our schematic representation, this means that there exists an assembly which looks like the one shown in part (b) of Figure~\ref{fig:digraphCrossing}.  The second condition of the lemma statement says that there exists some assembly sequence such that $\gamma_N$ grows independently of $\gamma_S$ or $\gamma_S$ always grows through $\gamma_M$ to grow $\gamma_N$.  This is represented schematically in part (c) of Figure~\ref{fig:digraphCrossing}.  The result of Lemma~\ref{lem:independent-subassemblies} is that $\gamma_N$ can be grown without growing $\gamma_S$ which is represented schematically in part (d) of Figure~\ref{fig:digraphCrossing}.

\begin{lemma}\label{lem:independent-subassemblies}
Let $\calT$ be a directed TAS and let $\alpha \in \termasm{\calT}$.  If $\gamma_N, \gamma_M, \gamma_S \sqsubseteq \alpha$ are subconfigurations such that
\begin{enumerate}
\item there exists $\gamma \in \prodasm{\calT}$ such that $\gamma_M \sqsubseteq \gamma$ and for all $x \in \dom(\gamma_S)$, $x \not\in \dom(\gamma)$, and
\item there exists $\vec{\alpha}$ of $\calT$ such that $\res{\vec{\alpha}} = \alpha^*$ where $\gamma_N, \gamma_M \sqsubseteq \alpha^*$ and for all paths $p$ in $G_{\vec{\alpha}}$ from a tile in $\gamma_S$ to a tile in $\gamma_N$, $p$ cuts $\gamma_M$,
\end{enumerate}
then there exists $\alpha' \in \prodasm{\calT}$ such that $\gamma_M, \gamma_N \sqsubseteq \alpha'$ and for all $x \in \dom(\gamma_S)$, $x \not\in \dom(\alpha')$.
\end{lemma}

\begin{proof}
Let the hypotheses hold.  We now construct an assembly sequence $\vec{\beta}$ which contains an assembly $\alpha'$ such that $\gamma_M, \gamma_N \sqsubseteq \alpha'$ and for all $x \in \dom(\gamma_S)$, $x \not\in \dom(\alpha')$.  Let the assembly sequence $\vec{\gamma}^{i=k}_{i=0}$ be an assembly sequence in $\calT$ such that $res(\vec{\gamma}) = \gamma$.  We construct $\vec{\beta}$ by passing $\vec{\alpha}$ and $\vec{\gamma}$ as arguments to Algorithm~\ref{alg:cross} and store the output of the algorithm in $\vec{\beta}$.
\begin{algorithm}
\caption{An algorithm for constructing $\vec{\beta}$.~\label{alg:cross}}
	\KwIn{$\vec{\gamma} = (\gamma_0, \gamma_1, ... , \gamma_k)$, $\vec{\alpha} = (\alpha_0, \alpha_1, ...)$}
	\KwOut{$\vec{\beta}=(\beta_0, \beta_1, ... )$}
	\For{$i \in [0, k]$}{
        $\beta_i \defeq \gamma_i$\;
	}
    \For{$i \in [1, |\vec{\alpha}|]$}{
        $t \defeq \dom(\alpha_i) \setminus \dom(\alpha_{i-1})$\;
        \If{$t \not\in \dom(res(\beta))$ and $t \not\in \dom(\gamma_S)$}{
		  $\beta_{k+i} \defeq \beta_{k+i-1}+t$\;
	    }
        unless there exists a path from $\gamma_S$ to $t$ in $G_{\vec{\alpha}}$ which does not cut $\gamma_M$\;
	}
\end{algorithm}
Note that for all $\beta_i \in \vec{\beta}$, $\beta_i \longrightarrow^t \beta_{i+1}$ is valid since 1) the algorithm ensures that $t \not\in \beta_i$ and 2) $t$ attaches with strength $\tau$ since the algorithm ensures all of $t's$ input glues are present.  Also note that since it is assumed that for any tile $t$ in $\gamma_N$ all paths from $\gamma_S$ to $t$ cut $\gamma_M$, there exists an assembly $\alpha'' \in \vec{\beta}$ such that $\gamma_M \sqsubseteq \alpha''$, $\gamma_N \sqsubseteq \alpha''$ and for all $x \in \dom{\gamma}_S$, $x \not\in \dom(\alpha'')$.
\end{proof}

\subsection{Necessity of Probes}

We say a macrotile is an $L$-macrotile or an $R$-macrotile if, under the representation function $R$, the macrotile maps to either a $0_L$ or $1_L$ tile type or either a $0_R$ or $1_R$ tile type respectively. Additionally, we say a macrotile is a bit-alley macrotile if the macrotile maps to either a $1_M$, $1_B$, $0_M$, or $0_B$ tile under the representation function. Finally, we call the $m\times (m+4)$ region which consists of the macrotile region between $L$-macrotiles and $R$-macrotiles extended by two tile widths to the west and east the \emph{probing region}.

\begin{figure}[htp]
\begin{center}
\includegraphics[width=4.0in]{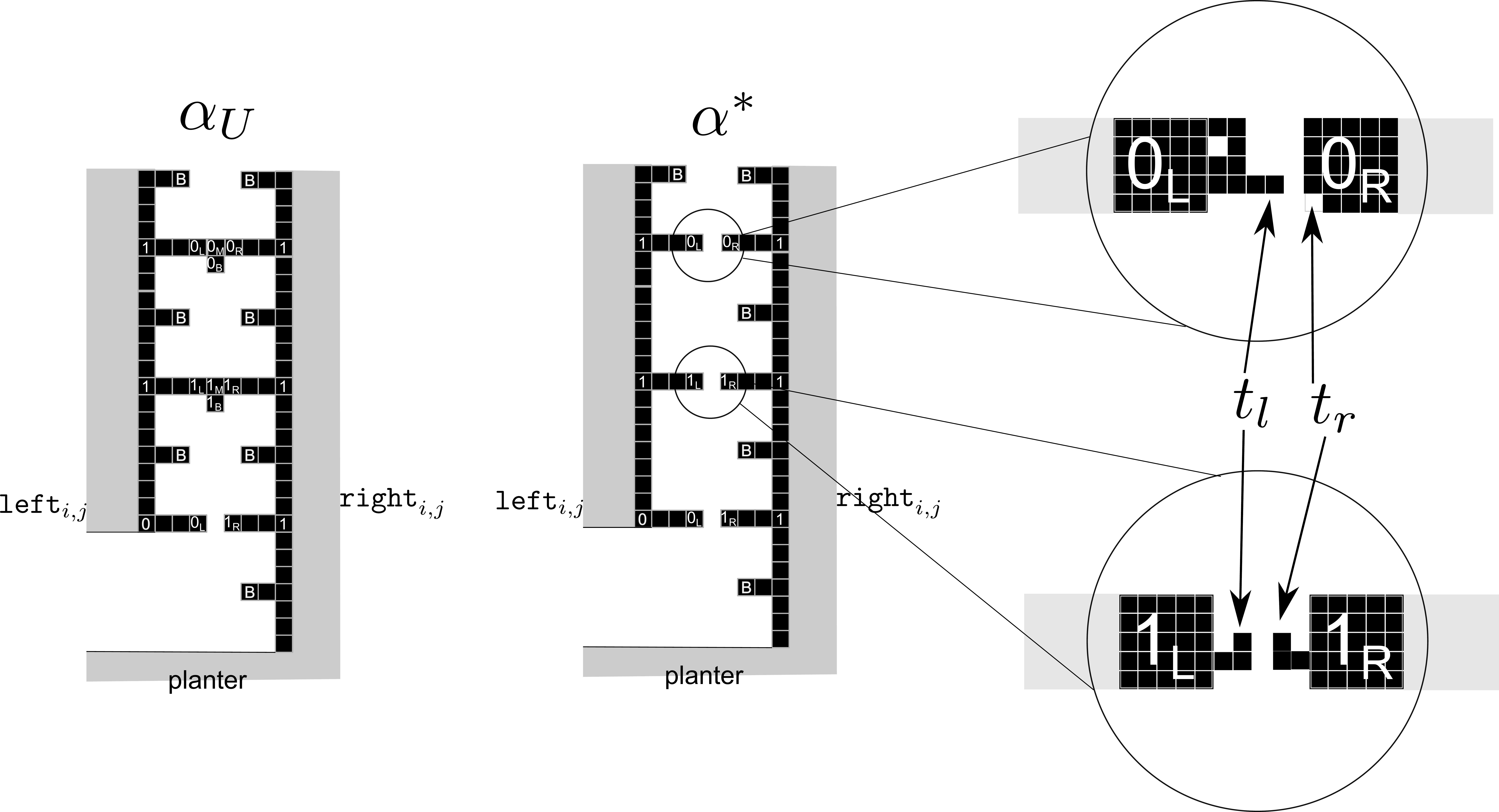}
\caption{An example of subconfigurations of assemblies discussed in Lemma~\ref{lem:gap-must-exist} along with the tiles $t_l$ and $t_r$ also discussed in the lemma.  This left part of this figure shows the \bitalley/ subconfiguration of subiteration $i, j$ in $\alpha^S \in \termasm{\mathcal{S}}$.  The right part of this image shows the \bitalley/ subconfiguration of subiteration $i, j$ in $\alpha^* \in \prodasm{\mathcal{S}}$ and also shows a zoomed in view which shows examples of the tiles $t_l$ and $t_r$ as described in the lemma statement.}
\label{fig:probe-gap}
\end{center}
\end{figure}

Next we show a lemma which intuitively says that for a non-empty subiteration, subiteration $i,j$ say, a valid simulation of $\calT$, it must be the case that simulated probes on the right side of $\leftcomp/^S_{i,j}$ and the simulated probes on the left side of $\rightcomp/^S_{i,j}$ must in fact cooperate. Referring to Figure~\ref{fig:probe-gap}, in order for $\mathcal{S}$ to simulated $\calT$, the macrotiles $0_L$ and $0_R$ depicted in the figure must assemble probes that come within one tile of each other to allow for cooperation across the simulated bit-alley. The high-level idea is that the macrotile $0_M$ cannot assemble to represent a non-empty tile of $T$ until the macrotiles $0_L$ and $0_R$ have assembled. Therefore, the macrotiles $0_L$ and $0_R$ must ``coordinate''. Moreover, this coordination cannot be the result of growing a path through the macrotile regions corresponding to simulations of bumper tiles.

\begin{lemma}\label{lem:gap-must-exist}
Let $i, j \in \mathbb{N}$.  There exists $\alpha^* \in \prodasm{\mathcal{S}}$ such that 1) $\leftcomp/_{i,j}$ and $\rightcomp/_{i,j}$ modules are grown without assembling any bit-alley macrotiles, and 2) if $\alpha^S$, the terminal assembly of $\mathcal{S}$,  contains a bit-alley macrotile in the probing region $R$, then $\alpha^*$ contains the tiles $t_l$ and $t_r$ in the probing region $R$ which grow from an $L$ macrotile and $R$ macrotile respectively such that $t_l$ and $t_r$ are at most one tile width apart.
\end{lemma}
\begin{proof}

Let $C$ denote the $(i,j)$ subconfiguration. Let $\beta_{i,j}$ be an assembly in $\prodasm{T}$ such that $C$ contained in $\beta_{i,j}$ has assembled to where every tile of type $1_M$, $1_B$, $0_M$, or $0_B$ that can bind has. Now, suppose that no tile of type $B$ in this subconfiguration has attached. Refer to Figure~\ref{fig:gap-overview} for a depiction of these tiles types and the tile locations where they bind in a subconfiguration. Under the assumption that $\mathcal{S}$ simulates $\mathcal{T}$, it must that there exists $\beta'$ in $\prodasm{S}$ such that $R^*(\beta') = \beta$. We now use $\beta'$ to construct $\alpha^*$. First, we have the following observation.

\begin{claim}\label{obs:nogap}
For each probe region $R$ of $\beta'_{i,j}$ that contains a bit-alley macrotile, there must be a path of adjacent tiles with matching glues along their adjacent edge which starts with a tile at a westernmost location in $R$ and ends with a tile at an easternmost location of $R$.
\end{claim}

We prove Claim~\ref{obs:nogap} by contradiction. Therefore, suppose that there exists a probe region $R$ in $\beta'_{i,j}$ that contains a bit-alley macrotile such that there is no path of adjacent tiles with matching glues along their adjacent edge which starts with a tile on the west edge of $R$ and ends with a tile on the east edge of $R$. Let $\mu$ and $\eta$ be the macrotiles subassemblies of $\beta'_{i,j}$ that map to $1_M$ or $0_M$ and $1_B$ or $0_B$ respectively. Note that no tile of $\eta$ can bind until the macrotile region $M$ above $\eta$ maps to $1_M$ or $0_M$ under $R$, for otherwise $\mathcal{S}$ is not a valid simulation.  Moreover, $M$ must map to the empty tile until the macrotile to west of $M$ maps to either a $0_L$ or $1_L$ tile type and the macrotile to the east of $M$ maps to either a $0_R$ or $1_R$ tile type.

Under the assumption that there is no path of adjacent tiles with matching glues along their adjacent edge which starts with a tile at a westernmost location of $R$ and ends with a tile at an easternmost location of $R$, it must be the case that there is a cut $v$ of $R$ such that the sum of all glue strengths corresponding to glues of adjacent tiles on each side of this cut is zero.

Therefore, in any assembly sequence of $\beta'_{i,j}$ there must be a path $p$ of adjacent tiles starting from a tile in the $L$-macrotile and ending with a tile in the $R$-macrotile. $p$ cannot cross the cut $v$ since the sum of all glues along this cut is zero.
Hence, $p$ must contain tiles in macrotile regions that map to $B$ tile types or map to tile types of the \leftcomp/ or \rightcomp/ module that adjacent to $B$ tile types.
In any case, now consider the assembly sequence in $\mathcal{T}$ where tiles of type $B$ always bind before tiles of type $0_R$, $1_R$, $0_L$, or $1_L$.
In the simulation of this sequence, it must be the case that the macrotile regions mapping to $B$ tile types or to a tile adjacent to a $B$ tile type contain enough tiles of $p$ to assemble the portion of $p$ starting from the macrotile regions mapping to $B$ tile types or to a tile adjacent to a $B$ tile type and ending with a tile of $\eta$. This follows from Lemma~\ref{lem:independent-subassemblies}. This portion of the path $p$ either contains a tile in an $L$-macrotile or an $R$-macrotile. Suppose this portion of the path $p$ contains a tile in an $L$-macrotile (the $R$-macrotile case is similar).
Then, note that this portion of the path can assemble even in the absence of any tiles of the $R$-macrotile. Therefore, there is an assembly sequence of $\mathcal{S}$ such that a tile of $\eta$ binds before any tiles of the $R$-macrotile bind, which is before $M$ can map to $1_M$ or $0_M$. This violates valid simulation of $\mathcal{T}$ by $\mathcal{S}$ since the tile of $\eta$ that binds is outside of any fuzz region.  Therefore, Claim~\ref{obs:nogap} holds.

To finish the proof of Lemma~\ref{lem:gap-must-exist}, start with the assembly $\beta'$. By Claim~\ref{obs:nogap}, it must be the case that for each probing region $R$, there is no strength zero cut of $R$ separating tiles of $\beta'|_R$. Hence, we can obtain $\alpha^*$ as follows. For each probing region $R$, if a single tile wide gap does not exist in a probing region $R$, remove a tile from the subassembly of $\beta'$ contained in $R$ in such a way that the resulting assembly is still valid until there is a single tile wide gap. This tile removal corresponds to ``rewinding'' the assembly sequence of $\beta'$ in the region $R$ just to the point were there is a single tile wide gap between tiles stably attached to tiles of an $L$-macrotile and tiles stably attached to tiles of an $R$-macrotile.
\end{proof}

\subsection{Narrowing down the outputs of a set of Turing machines}
Let $(M_i)^{i=n}_{i=0}$ be a sequence of Turing machines which are almost everywhere hard for any TM which uses less space than $T(n)$.  The next lemma essentially says that if there is a Turing machine $M'$ which outputs a set of strings such that there is a string $w$ in the set such that the $i^{th}$ bit of $w$ corresponds to the output of TM $M_i$ on input $x$, then the TM $M'$ must use either at least as much space as the least space complex TM $M_i$ or it uses at least as much space as the most complex TM $M_j$ depending on the strings in the set.

Lemma~\ref{lem:tech-tm-constant} is used in Section~\ref{sec:contradiction}.  In Section~\ref{sec:contradiction}, we must show that in some empty subiteration $i,j$, the modules $\leftcomp/^S_{i,j}$ and $\rightcomp/^S_{i,j}$ cannot somehow narrow down the arms which are going to be grown from the $\topcomp/_{i,j}$ and grow the probes in the \bitalley/ region so that the probes which they grow are ``consistent'' with the arm which will grow from $\topcomp/_{i,j}$.  In Section~\ref{sec:contradiction} we show that it is impossible for the adversary to narrow down the arms which will grow from $\topcomp/_{i,j}$ to a set of $t$ arm types.  We use Lemma~\ref{lem:tech-tm-constant} to show this.  Intuitively, this lemma says it is impossible for the adversary to narrow down the arms which can grow for the following reasons.  First note, that there are $2^t$ possible arm types that can grow from $\topcomp/_{i,j}$ each of which represents the output of the series of $B$ computations on input $i$.  So, we can think of this set of arm types as a set of strings which correspond to the output of the series of $B$ computations on input $i$, and we denote this set by $E$.  Also, recall that the series of $B$ computations require space $O(3^i)$ and the $\leftcomp/^S_{i,j}$ module has space $O(2^i)$.  Suppose the adversary can narrow down the arm type grown from the $\topcomp/^S_{i,j}$ module to $t$ choices (implying $|E| = t$). If there exists a bit position $k$ such that all the strings in set $E$ agree on bit position $k$, then the adversary knows the output of machine $B_k$ on input $i$, which contradicts the space complexity of the language decided by $B_k$.  So it must be the case that for bit position $k$ in every $X^0 \in E$ there exists a string $X^1 \in E$ so that $X^0$ and $X^1$ disagree on bit position $k$.  But if this is the case, note that the adversary could then run at most the first $|E|$ $B$ computations to prune his set $E$ down to a single string which must contain the answer to the last computation $B_H$ which requires asymptotically more space.  This implies the adversary is able to recognize the language decided by $B_H$ in less space than required which is impossible.  So, it must be the case he can't narrow down the set of arms that can grow from $\topcomp/^S_{i,j}$ to $t$ choices.

The preceding discussion highlighted the main idea of Lemma~\ref{lem:tech-tm-constant} which we sate and prove now.

\begin{lemma} \label{lem:tech-tm-constant}
Let $T(n) \in o(T'(n))$.  Also, let $(M_i)^{i=h}_{i=0}$ be a sequence of Turing machines which decide the languages $L_i$ such that for all $0 \leq i < h$, $L_i \in DSPACE(T(n))$ and cannot be recognized in i.o best case complexity $H(n)$ where $H(n) \in o(T(n))$, and $L_h \in DSPACE(T'(n))$ and cannot be recognized in i.o best case complexity $H'(n)$ where $H'(n) \in o(T'(n))$.  Suppose $M'$ is a Turing machine which on input $x$ outputs a set of $m < h$ strings $E\subset \{0,1\}^h$ such that there exists $X \in E$, so that the $i^{th}$ bit of $X$ is $M_i(x)$.  Then for almost all $x \in \{0, 1\}^*$, 1) if there exists $0 \leq k < h$ such that all strings output by $M'$ on input $x$ have the same value at bit position $k$, $M'$ requires at least space $T(n)$ or 2) if there does not exist such a $k$, $M'$ on input $x$ requires space $T'(n)$.
\end{lemma}

\begin{proof}
Assume the hypotheses, and suppose for the sake of contradiction that on infinitely many $x \in \{0, 1\}^*$, there exists $0 \leq k < h$ such that all strings output by $M'$ on input $x$ have the same value at bit position $k$, and suppose $M'$ uses space $G(n)$ where $G(n) \in o(T(n))$.  Let $0 \leq k' < h$ be the position such that for infinitely many $x \in \{0, 1\}^*$, the set of strings output by $M'(x)$ agrees on the bit position $k'$.  We know such a $k'$ exists by the pigeonhole principle.  Then on infinitely many $x \in \{0,1\}^*$, we can construct a Turing machine $M^*$ which decides the language $L_{k'}$ by first running the TM $M'$ on $x$.  Then, $M^*$ determines if all the strings output by $M'$ on $x$ agree on bit position $k'$.  If the strings do all agree on bit position $k'$, then $M^*$ outputs that bit.  Otherwise, $M^*$ runs the TM $M_{k'}$ on $x$ and outputs $M_{k'}(x)$.  Now, observe that on the infinitely many $x \in \{0,1\}^*$ where the output of $M'$ agrees on the bit position $k'$, the machine $M^*$ only uses the amount of space required by $M'$ on input $x$ which is $G(n)$.  This contradicts the assumption that $L_{k'}$ cannot be recognized in i.o best-case space complexity $H(n)$ for $H(n) \in o(T(n))$.

Now, suppose that on infinitely many $x \in \{0, 1\}^*$, the set $E$ of strings $M'$ outputs when run on $x$ is such that for every bit position $i$ in $X^0 \in E$, there exists $X^1 \in E$ such that $X^0$ and $X^1$ disagree on bit position $i$.  And once again, assume for the sake of contradiction that on such $x$, $M'$ runs in space $G(n) \in o(T'(n))$.  Then on infinitely many $x \in \{0, 1\}^*$, we construct a Turing machine $M^*$ which decides the language $L_h$ in the following manner.  First, $M^*$ runs $M'$ on input $x$ and saves the set $E$ of strings output by $M'$ on input $x$.  Now, $M^*$ runs the TM $M_i$ on input $x$ for $i < h$ and tosses out any string $Y \in E$ where the $i^{th}$ bit position of $Y$ disagrees with the output of $M_i$ on input $x$.  Note that by the assumption that for every bit position $i$ in $X^0 \in E$, there exists $X^1 \in E$ such that $X^0$ and $X^1$ disagree on bit position $i$, we need to run at most $|E|$ of the $(M_i)^{i=|E|-1}_{i=0}$ Turing machines before we are left with a single string which contains the answer to the computation $M_h(x)$.  Now, $M^*$ outputs the value of the bit position $h$ in the last string left.  By assumption, this value is equal to $M_h(x)$.  Observe that our machine can be designed to use space $F(n)$ where
\begin{eqnarray} \label{eqn:space}
F(n) &=& |E|*h + T(n +|E|)
 \end{eqnarray}
 since $|E|*h$ space is required to save the output of $M'$ and the machine $M_{|E|}$ requires space $T(|x| +|E|)$. Note that the space $M^*$ uses to compute $M_i$ can be reused to compute the TM $M_{i+1}$.  Since $|E|$ and $|H|$ are fixed, we have that the TM $M^*$ uses space $O(T(n))$ on infinitely many $x$.  This contradicts the assumption that $L_h$ cannot be recognized in i.o. best-case space complexity $H(n)$ for $H(n) \in o(T(n))$.
\end{proof}

The idea behind Observation~\ref{obs:tech-tm} is the same as that behind Lemma~\ref{lem:tech-tm-constant}, but is more generalized to allow the series of Turing machines and set of ``guesses'' to change size with the input.  This will be needed in the proof of Lemma~\ref{lem:noSneak}.

\begin{observation} \label{obs:tech-tm}
Let $A^+$ be the machine described in Section~\ref{sec:the-system} and let $c \in \mathbb{N}$.  Let $M'$ be any Turing machine such that on input $x$ $M'$ outputs a set of $\log^c(|x|)$ strings, denoted $E_x \subset \{0, 1\}^{|x|}$, such that $A^+(x) \in E_x$.  Then for almost all $x \in \{0, 1\}^*$, 1) if there exists $0 \leq k < n$ such that all strings output by $M'$ on input $x$ have the same value at bit position $k$, $M'$ requires at least space $2^{\frac{|x|}{2}}$ or 2)if there does not exist such a $k$, $M'$ on input $x$ requires space $F(n)$ where $F(n)$ is such that $|x| \log^c(|x|) + F(|x| + \log^c(|x|)) \in \Omega(2^n)$.
\end{observation}
Note that if $|E|$ and $h$ are not constant in Lemma~\ref{lem:tech-tm-constant}, as is the case in Lemma~\ref{lem:noSneak}, then it must be the case that $|E|*h + T(n +|E|) \in \Omega(T'(n))$ where $T(n)$ and $T'(n)$ are as defined in the statement of Lemma~\ref{lem:tech-tm-constant}.  This follows from the fact that $M^*$ decides the language $L_h$ in space $|E|*h + T(n +|E|)$ as shown in Equation~\ref{eqn:space}.

\subsection{Zig-zag assembly systems}
In \cite{SingleNegative}, a system $\mathcal{T} = (T, \sigma, \tau)$ is called a zig-zag tile assembly system provided that (1) $\mathcal{T}$ is directed, (2) there is a single sequence $\vec{\alpha} \in \mathcal{T}$ with $\termasm{\calT} = \{\vec{\alpha}\}$, and (3) for every $\vec{x} \in \dom \alpha, (0,1) \not\in$ IN$^{\vec{\alpha}}(\vec{x})$. We say that an assembly sequences satisfying (2) and (3) is a \emph{zig-zag assembly sequence}.
Intuitively, a zig-zag tile assembly system is a system which grows to the left or right, grows up some amount, and then continues growth again to the left or right.  Again, as defined in \cite{SingleNegative}, we call a tile assembly system $\mathcal{T} = (T, \sigma, \tau)$ a \emph{compact zig-zag tile assembly system} if and only if $\termasm{\calT} = \{\vec{\alpha}\}$ and for every $\vec{x} \in \dom \alpha$ and every $\vec{u} \in U_2$, $\textmd{str}_{\alpha(\vec{x})}(\vec{u}) + \textmd{str}_{\alpha(\vec{x})}(-\vec{u}) < 2\tau$. Informally, this can be thought of as a zig-zag tile assembly system which is only able to travel upwards one tile at a time before being required to zig-zag again. The assembly sequence of a compact zig-zag system is called a \emph{compact zig-zag assembly sequence}. Figure~\ref{fig:zig-zag} depicts a compact zig-zag assembly sequence. As in the definition of a zig-zag system and throughout this section, we assume that each row of a zig-zag systems binds to the north of the previous row.

\subsection{Space complexity of zig-zag systems is invariant under simulation}

In this section, we give a formal definition of a language defined by a zig-zag system. We next show that such a language can be computed in space on the order of the maximal width of the zig-zag assembly grown to a finite height. While this result is fairly straightforward, we include it for the sake of completeness and because it serves as a basic example of how we will prove the main result of this section (Lemma~\ref{lem:noCheating}). We give a formal definition of a language defined by a simulation of a zig-zag system, and Lemma~\ref{lem:noCheating} states that such a language can be computed in space on the order of the maximal width of the zig-zag assembly grown to a finite height.

Here is some of the notation used in this section. Let $\mathcal{T} = (T, \sigma, \tau)$ be a temperature $\tau$ compact zig-zag system with a seed $\sigma$ consisting of a single tile, and let $\alpha$ be an assembly in $\prodasm{T}$. Since all of the results in this section hold regardless of the location of $\sigma$, without loss of generality, throughout this section, we assume that the location of $\sigma$ is $(0,0)$. Finally, we will use the term \emph{configuration} to denote a partial function from a finite domain in $\Z^2$ to $T$, and \emph{finite configuration} when the domain of the partial function from $\Z^2$ to $T$ is finite.

\subsubsection{Computational complexity and zig-zag systems}

Let $T_1\subseteq T$ be a subset of $T$, and let $r:\N \to \{0,1\}$ be the function defined as

\[
   r(n) = \left\{
     \begin{array}{ll}
       1 & (0,n)\in \dom\alpha \text{ and } \alpha((0,n)) \in T_1\\
       0 & \text{otherwise.}
     \end{array}
   \right.
\]

\noindent Now, let $f:\N \to \N$ be the function
$$f(n) = \max\{w_j\mid w_j \text{ is the width of the } j^{th} \text{ row of } \alpha \text{ for } 0\leq j\leq n \}.$$

\noindent Finally, let $L_r = \{ n\in \N\mid r(n)=1 \}$. We call $r$ the \emph{characteristic function for} $\mathcal{T}$ given $T_1$, and $L_r$ the \emph{language defined by} $\mathcal{T}$ \emph{given} $r$. Notice that $r$ is a computable function, $f$ is a proper function, and $L_r$ is a computable set. See Figure~\ref{fig:zig-zag} for a description of how $r(n)$ is computed.

\begin{figure}[htp]
\begin{center}
\includegraphics[width=2.0in]{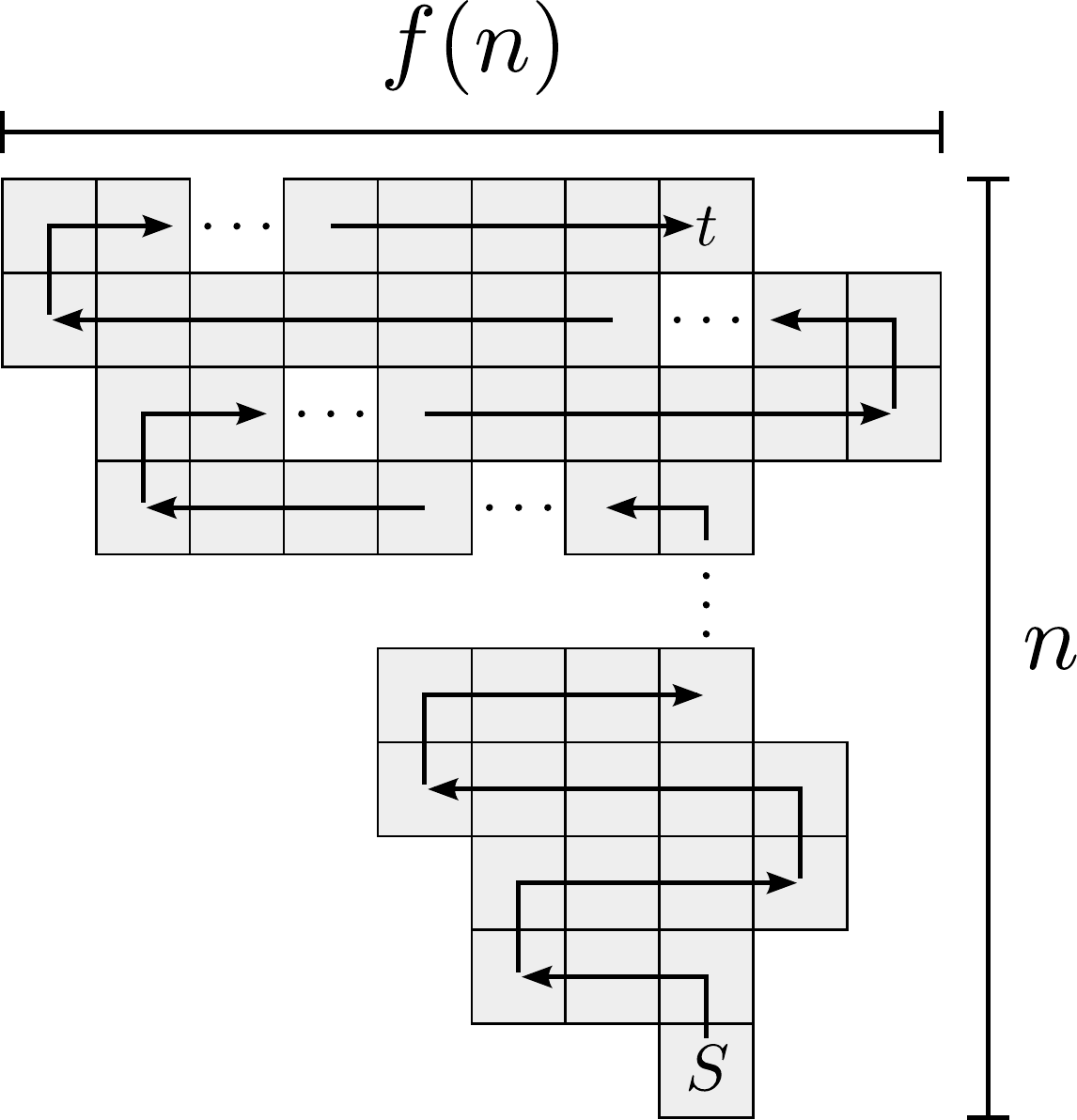}
\caption{An assembly with a zig-zag	 assembly sequence. The assembly sequence is indicated with arrows. The tile labeled $S$ makes up the seed $\sigma$, and $r(n)=1$ if and only if the tile type of the tile labeled $t$ is in $T_1$.}
\label{fig:zig-zag}
\end{center}
\end{figure}

The following lemma gives an upper bound on the space complexity of a language defined by a zig-zag system.

\begin{lemma}\label{lem:ZigZagSpaceComplexity}
Let $\mathcal{T} = (T, \sigma, \tau)$ be a zig-zag system with tile set $T$, seed assembly $\sigma$, and temperature $\tau$. Let $T_1$ be some subset of $T$, let $r$ be the characteristic function for $\mathcal{T}$ given $T_1$, and let $L_r$ be the language defined by $\mathcal{T}$ given $r$.  Finally, let $f(n)$ denote the width of the longest row of the assembly of $\mathcal{T}$ consisting of $n$ completed rows. Then, $L_r \in \DSPACE(f(n))$.
\end{lemma}

\begin{proof}

Algorithm~\ref{alg:zig-zag-algorithm} defines steps for computing $r(n)$.
\begin{algorithm}
\caption{An algorithm for computing $r(n)$.~\label{alg:zig-zag-algorithm}}
	\KwIn{$n\in \N$}
	\KwOut{$r(n)$}
	$\alpha \defeq \sigma$\;
	\While{$\partial^\tau \alpha \neq \emptyset$}{
		choose $(i, j)\in \Z^2$ such that $(i,j) \in \partial^\tau \alpha$\;
		\If{$j = n+1$}{
			break\;
		}
		$l\defeq (i,j)$\;
		choose $t\in T$ such that $l\in \partial^\tau_t \alpha$\;
		$\alpha \defeq \alpha + (l \mapsto t)$\;
	}
	\eIf{$(0,n)\in \dom(\alpha)$ and $\alpha((0,n))\in T_1$}{
		return $1$\;
	}{
		return $0$\;
	}
\end{algorithm}

For an input $n\in \N$, notice that Algorithm~\ref{alg:zig-zag-algorithm} decides if $n$ belongs to $L_r$.  Notice that to perform this algorithm, as the top row of $\alpha$ assembles, the $j^{th}$ row say,
only the top two non-empty rows of the zig-zag assembly are needed. In other words, there is a computation that only requires space on the order of the tiles with locations $(i,j-1)$ or $(i,j)$ for some $i\in\Z$. This is due to the fact that for zig-zag systems, once a tile is added to the $j^{th}$ row, only tiles in these locations have glues that allow for the binding of an additional tile.
There are at most $2f(n)$ tiles with locations $(i,j-1)$ or $(i,j)$ for some $i\in\Z$ by the definition of $f$, and hence at most $2|T|f(n)$ tiles are required to compute $r(n)$. Therefore, $L_r\in \DSPACE(f(n))$.
\end{proof}

There are many generalities of Lemma~\ref{lem:ZigZagSpaceComplexity} that can be made at this point.
In the lemma, $r$ is defined in terms of $\alpha$, a compact zig-zag system, such that each row of $\alpha$ assembles to the north of the previous row. We could just as easily have defined a compact zig-zag system so that each row assembles to the south, east, or west of a previous row.
Also, note that it is not necessary that $\alpha$ is an assembly in a zig-zag system. In fact, $\mathcal{T}$ could be any TAS. It is only necessary that $\alpha$ in $\prodasm{T}$ has a compact zig-zag assembly sequence. Finally, the function $r$ can also be generalized. $r$ is defined using a single tile location, namely $(0,n)$, and the output of $r$ is determined by $\alpha((0,n))$ and a subset of tile types $T_1$. Lemma~\ref{lem:ZigZagSpaceComplexity} also holds if we define $r$ to be defined using any finite number of tile locations and a set of configurations.

\subsubsection{Paths in the binding graph}

Before we proceed with the proof of Lemma~\ref{lem:noCheating1}, we give a way to find special paths in the binding graph of assemblies. Let $\mathcal{S} = (S,\sigma_S, \tau')$ be a TAS and let $\alpha$ be in $\prodasm{S}$. Let $S_1$ and $S_2$ be some nonempty finite sets of nonempty finite subassemblies in $\alpha$. We say that $S_1$ and $S_2$ are \emph{pairwise nonoverlapping} if for any pair of subassemblies $\beta_1 \in S_1$ and $\beta_2 \in S_2$, $\dom(\beta_1) \cap \dom(\beta_2) = \emptyset$.
Moreover, we say that $S_2$ \emph{requires} $S_1$ if for any assembly sequence $\vec{\alpha}$ with result $\alpha$, there exists an assembly $\alpha_1$ in $\vec{\alpha}$ such that some subassembly of $S_1$ is a subassembly of $\alpha_1$ and no subassembly of $S_2$ is a subassembly of $\alpha_1$. Less formally, $S_2$ requires $S_1$ if at least one subassembly in $S_1$ must completely assemble prior to the assembly of any subassembly in $S_2$. Now we can state the following lemma about finding paths in the binding graph of an assembly.

\begin{lemma}\label{lem:paths}
Let $\alpha$ be a stable finite assembly with an arbitrary valid assembly sequence $\vec{\alpha}$, and let $S_1$ and $S_2$ be nonempty finite sets of nonempty finite subassemblies of $\alpha$ such that 1) $S_1$ and $S_2$ are pairwise nonoverlapping and 2) $S_2$ requires $S_1$. Also, let $\alpha_1$ be the first assembly in the assembly sequence $\vec{\alpha}$ containing any assembly of $S_1$, and let $\alpha_0$ be the assembly in $\vec{\alpha}$ that immediately proceeds $\alpha_1$ in $\vec{\alpha}$. Then there exists a path $p$ in the binding graph of $\alpha$ with vertex set $V \subseteq \Z^2$ such that
	\begin{enumerate}\vspace{-5pt}
		\itemsep1pt \parskip0pt \parsep0pt
		\item for some $\gamma_1\in S_1$ and $\gamma_2\in S_2$, $V\cap \dom(\gamma_1) \neq \emptyset$ and $V\cap \dom(\gamma_2) \neq \emptyset$, and \label{itm:prop1}
		\item $V\cap \dom(\alpha_0) = \emptyset$\label{itm:prop2}
	\end{enumerate}
\end{lemma}

\begin{proof}
Since $\alpha$ is a stable assembly, it is clear that at least one path $p$ in the binding graph of $\alpha$ satisfying Property~\ref{itm:prop1} always exists. That we can find such a path satisfying Property~\ref{itm:prop2} follows from the assumption that $S_2$ requires $S_1$.
\end{proof}

The path $p$ in Lemma~\ref{lem:paths} corresponds to a path of tiles in $\alpha$, and we call such a path the \emph{a new path of tiles connecting} $S_1$ and $S_2$ to emphasize Property~\label{itm:prop2}. Property~\label{itm:prop2} essentially says that such a path $p$ forms only after a configuration in $S_1$ has assembled.
To help motivate Lemma~\ref{lem:paths}, consider that for a simulation $\mathcal{S}$ of a zig-zag system $\mathcal{T}$, in order for $\mathcal{S}$ to capture the dynamics of $\mathcal{T}$ correctly, for a row of macrotiles simulating a row of zig-zag growth from left to right from a tile $t_1$ to a tile $t_2$ say (respectively right to left), the set $S_1$ of macrotile subassemblies that represent $t_1$ requires the set $S_2$ of macrotile subassemblies that represent $t_2$. By Lemma~\ref{lem:paths}, this implies that there is a new path of tiles connecting $S_1$ and $S_2$. In the proof of Lemma~\ref{lem:noCheating1}, we will use Lemma~\ref{lem:paths} to limit the amount of space that can be used as assembly in the simulating system proceeds.

\subsubsection{Simulations of zig-zag systems}

Let $\mathcal{S} = (S, \sigma_S, \tau')$ be a TAS that simulates $\mathcal{T}$ with representation function $R$ and scale factor $c$. Since all of the results here hold upto translation of assemblies in $\Z^2$, without loss of generality, throughout this section we assume that the bottom-right tile of $\sigma_{\mathcal{T}}$ has location $(0,0)$. Also let $\alpha'$ in $\prodasm{S}$ and $\alpha$ in $\prodasm{T}$ be assemblies such that $R^*(\alpha') = \alpha$. Furthermore, let $\vec{\alpha'}$ be an assembly sequence with result $\alpha'$.
 Here we give a similar result to Lemma~\ref{lem:ZigZagSpaceComplexity} for systems such as $\mathcal{S}$ that simulate zig-zag systems.

First, we introduce some notation similar to the notation used for stating Lemma~\ref{lem:ZigZagSpaceComplexity}. Let $\alpha'_n$ denote the subassembly of $\alpha'$ such that
for all $(i,j) \in \dom(\alpha'_n)$, $j\leq n$ and for all $(i,j) \in \partial^{\tau'}\alpha'_n$, $j\geq n+1$. For some $L\subset \Z^2$ such that $|L|< \infty$ and for $\vec{v}\in\Z^2$, let $L_{\vec{v}}$ denote $\{\vec{l} + \vec{v} \mid \vec{l}\in L\}$.
Also, for $\vec{n} = (0,n)$, let $C_n \subseteq \{w \mid w:L_{\vec{n}} \to S \text{ is a partial function} \}$ be a subset of configurations over $S$ with domain in $L_{\vec{n}}$. Then we let $r':\N \to \{0,1\}$ be the function
\[r'(n) = \left\{
     \begin{array}{ll}
       1 & \alpha'_n|_{L_{\vec{n}}} \in C_n \\
       0 & \text{otherwise.}
     \end{array}
   \right. \]

Essentially, $r'$ is the function obtained by growing $\alpha'$ to the point where the next tile added must be at a location above the line $y=n$, and then considering some finite configuration of this assembly. $r'(n)=1$ if and only if this configuration is in $C_n$. Furthermore, let $f':\N \to \N$ be the function defined as the maximum width $w$ such that $w = |x_1 - x_2|$ where $(x_1,y_1)$ and $(x_2,y_2)$ are locations of some tiles in $\alpha'$ such that $0\leq y_1, y_2\leq n$.
Finally, let $L_r' = \{ n\in \N\mid r'(n)=1 \}$. As with zig-zag systems, we call $r'$ the \emph{characteristic function for} $\mathcal{S}$ \emph{given} $C_n$ and $L_{r'}$ the \emph{language defined by} $\mathcal{S}$ \emph{given} $r'$ \emph{and} $C_n$. Notice that $r'$ is a computable function
and $L_{r'}$ is a computable set. Lemma~\ref{lem:noCheating1} gives an upper bound on the space complexity of a language defined by a system that simulates a zig-zag system.

\begin{lemma}\label{lem:noCheating1}
Let $\mathcal{T} = (T, \sigma, \tau)$ be a zig-zag system and let $\mathcal{S} = (S, \sigma_\mathcal{T}, \tau')$ be a directed system that simulates $\mathcal{T}$ with terminal assembly $\alpha'$. Moreover, let $L$ be a finite subset of $\Z^2$ and let $C_n$ be a set of finite configurations over $S$ with domain in $L_{\vec{n}}$, $r'$ be the characteristic function for $\mathcal{S}$ given $C_n$. Finally, $L_{r'}$ the language defined by $\mathcal{S}$ given $r'$ and $C_n$, and
let $f'(n)$ be the maximum width $w$ such that $w = |x_1 - x_2|$ where $(x_1,y_1)$ and $(x_2,y_2)$ are locations of some tiles in $\alpha'$ such that $0\leq y_1, y_2\leq n$.
Then, $L_{r'}\in \DSPACE(f'(n))$.
\end{lemma}

\subsubsection*{Overview of the proof of Lemma~\ref{lem:noCheating1}}

In order to prove Lemma~\ref{lem:noCheating1}, we show that $r'(n)$ can be computed using space in $O(f'(n))$. Let $c$ denote the scale factor of the simulation of $\mathcal{T}$ by $\mathcal{S}$, and let $h$ denote the height of the smallest rectangle bounding $L$. We let $k = \max\{ c^2 + 2c + 2, h \}$. As assembly proceeds from $\sigma_\mathcal{T}$, we take note of two regions where tiles may bind. First, tiles may bind in macrotile regions representing leftmost or rightmost tiles of the simulated zig-zag system or fuzz regions to the left or right of these macrotile regions. We call these macrotile regions the \emph{left or right sides} of an assembly of $\mathcal{S}$. The second region where tiles may bind is the complement of the left or right sides of an assembly of $\mathcal{S}$.

The proof relies on a data structure called a \emph{glue sequence table}\footnote{A glue sequence table is closely related to a \emph{window-movie} (GST) as described in~\cite{IUNeedsCoop}.}. This table, which we define later, is essentially a constant size (depending on $U$ and $c$) lookup table that maps sequences of glues along a portion of a cut corresponding to the left or right sides of an assembly to sets of glues. The cut divides the $\Z^2$ lattice into two regions. One ``above'' the cut and one ``below''. The sets of glues mapped to by the table correspond to the glues that appear on the cut under the assumption that the sequence of glues are exposed along the cut in the order given by the sequence. Then, by using these exposed glues, tiles bind below the cut until there are no tile locations below the cut where a tile can stably bind. The glue sequence table has an entry for every possible glue sequence that can appear along the portion of the cut. Basically, the glue sequence table captures all ``what-if'' scenarios in the sense that it tells us what glues may be exposed along a cut by tiles binding below the cut if a glue (or a sequence of glues) is exposed along the cut by tiles located above the cut. See Figure~\ref{fig:glue-map} for a depiction of this process. We will describe the construction of a GST in the formal proof. It should be noted that it may not be possible for a sequence in the domain of a glue sequence table to correspond to an actual valid assembly.

\begin{figure}[htp]
\centering
  \subfloat[][InitAssembly and InitGST]{%
        \label{fig:noCheating-overview1}%
	        \includegraphics[width=2.75in]{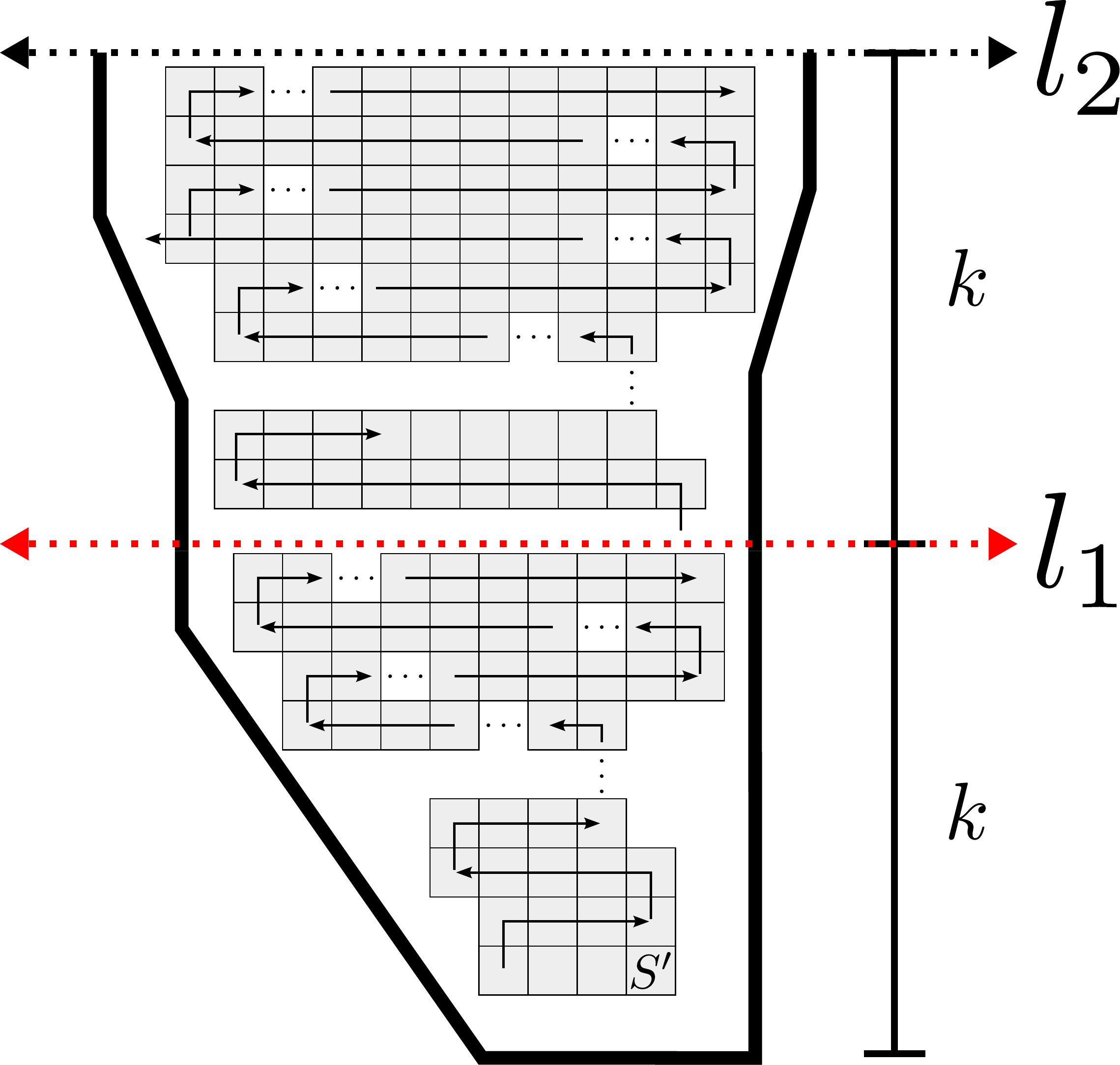}
        }%
        \quad\quad
  \subfloat[][UpdateAssembly and UpdateGST]{%
        \label{fig:noCheating-overview2}%
        		\includegraphics[width=2.75in]{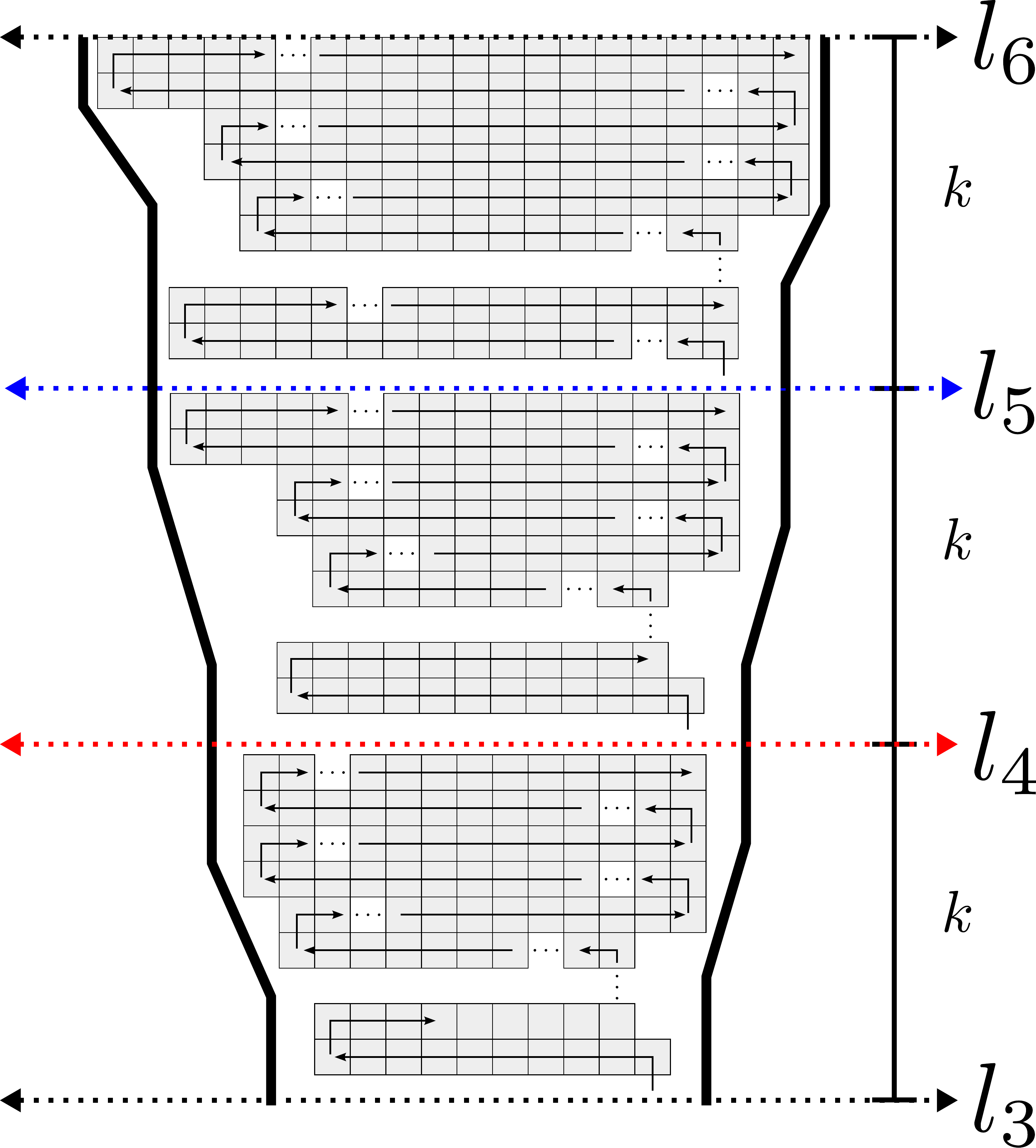}
        }%
  \caption{Determining configurations over $L_{\vec{n}}$.\vspace{-15pt}}
  \label{fig:noCheating-overview}
\end{figure}

In the proof, we describe four procedures for computing configurations over $L_{\vec{n}}$ starting from $\sigma_\mathcal{T}$. These procedures are 1) InitAssembly, 2) InitGST, 3) UpdateAssembly, and 4) UpdateGST. Figure~\ref{fig:noCheating-overview} depicts a sketch of the assemblies formed by these procedures. Referring to Figure~\ref{fig:noCheating-overview1}, starting from the seed $\sigma_\mathcal{T}$, InitAssembly is the process of attaching all tiles that can stably bind at tile locations with $y$ values $\leq 2k$. InitGST constructs the glue sequence table corresponding to cut shown in red by considering assemblies below the line that may form assuming that some glue (or sequence of glues) is exposed along the cut.
Referring to Figure~\ref{fig:noCheating-overview2}, with the assembly below the line $l_5$ present, UpdateAssembly is the process of attaching all tiles that can stably bind a tile locations below $l_6$. This process uses the existing assembly between $l_3$ and $l_5$ and the glue sequence table corresponding to a cut shown in red. The result of UpdateAssembly is that the assembly below the the line $l_6$ and above the the line $l_3$ is terminal. Once UpdateAssembly is finished, UpdateGST constructs a new glue sequence table corresponding to a cut shown in blue. When the assembly below $l_6$ and above $l_3$ has formed and the glue sequence table has been updated, the tiles below the line $l_4$ are no longer necessary for continuing the assembly of the simulation of the zig-zag system to the north. Using this fact, to finish the proof, we show that the four methods InitAssembly, InitGST, UpdateAssembly, and UpdateGST require space on the order of $f'(n)$ and can be used to compute $\alpha'_n|_{L_{\vec{n}}}$ in order to determine $r'(n)$.

\subsubsection*{Proof of Lemma~\ref{lem:noCheating1}}
\begin{proof}

For $n\in \N$, we let $H_n$ denote the set $\{(x,y)\in \Z^2\mid y\leq n \}$. Based on Algorithm~\ref{alg:simAlgo}, we will give a means of computing $r'(n)$ that requires space in $O\left(f'(n)\right)$.

\begin{algorithm}
\caption{An algorithm for computing $r'(n)$.\label{alg:simAlgo}}
	\KwIn{$n\in \N$}
	\KwOut{$r'(n)$}
	\SetKwFunction{initAssembly}{InitAssembly}
	\SetKwFunction{initGST}{InitGST}
	\SetKwFunction{updateAssemble}{UpdateAssembly}
	\SetKwFunction{updateGST}{UpdateGST}
	\initAssembly\;
	
	\initGST\;
	
	\While{$\partial^{\tau'} \alpha' \cap H_n \neq \emptyset$}{
		\updateAssemble\;
		
		\updateGST\;
		
	}
	\eIf{$\alpha'(L_{\vec{n}}) \in C_n$}{
		return $1$\;
	}{
		return $0$\;
	}
\end{algorithm}

The algorithm consists of four procedures that we describe next.
Before we describe each of the procedures used in Algorithm~\ref{alg:simAlgo}, we introduce some notation. Let $h$ denote the height of $L_{\vec{n}}$ and let $k = \max\{ c^2 + 2c + 2, h \}$.

\subsubsection*{InitAssembly}

We begin assembly in $\mathcal{S}$ by attaching tiles to $\sigma_S$ until any tile that can stably bind to below the line $y=2k$ has attached. We denote this step as the ``InitAssembly'' procedure. The InitAssembly procedure consists of the steps given in Algorithm~\ref{alg:primeAssembly}.

\begin{algorithm}
\caption{An algorithm describing the InitAssembly procedure.\label{alg:primeAssembly}}
	\KwIn{$\sigma_\mathcal{T}$}
	\KwOut{$\alpha'_{2k}$}
	$\alpha' \defeq \sigma_S$\;
	\While{$\partial^{\tau'} \alpha' \cap H_{2k} \neq \emptyset$}{
		choose $(i, j)\in \Z^2$ such that $(i,j) \in \partial^{\tau'} \alpha'$ and $j \leq 2k$\;
		$l\defeq (i,j)$\;
		choose $t\in T$ such that $l\in \partial^{\tau'}_t \alpha'$\;
		$\alpha' \defeq \alpha' + (l \mapsto t)$\;
	}
\end{algorithm}

\subsubsection*{InitGST}

Before we describe the InitGST procedure, we define a data structure that is used in the algorithm called a \emph{glue sequence table}. A glue sequence table is defined relative to a cut in the binding graph of $\alpha'_{k}\in \prodasm{S}$. Figure~\ref{fig:zig-zag-simulation} depicts this cut, which we now describe.

To define the cut, we find two tiles such that the location of each of these tiles is above the line $y=k$.
Let $i = \ceil{k/c}$ and let $M^r_i$ be the macrotile region that maps to the tile, $t^r_i$ say, in $\alpha$ that is farthest to the right on the $i^{th}$ row $\alpha$.  Similarly, let $M^l_i$ be the macrotile region that maps to the tile, $t^l_i$ say, in $\alpha$ that is farthest to the left on the $i^{th}$ row $\alpha$. Moreover, let $S^r_i$ (respectively $S^l_i$) be sets of configurations over $S$ with domain in $M^r_i$ (respectively $M^i_i$) such that each configuration of $S^r_i$ (respectively $S^l_i$) maps to $t^r_i$ (respectively $t^l_i$). Note that depending on the direction of growth (left or right) of the $i^th$ row of the zig-zag system $\mathcal{T}$, either $S^r_i$ requires $S^l_i$ or $S^l_i$ requires $S^r_i$. In either case, there must be a new path $p_i$ of tiles from a tile of a configuration in $S^r_i$ to a tile of a configuration in $S^l_i$. Note that for $i$ between $k$ and $2k$, each $p_i$ is disjoint. This follows from Lemma~\ref{lem:paths}. Now, $p_i$ either contains a tile at a tile location that is below any tile of $\sigma_{\mathcal{T}}$ or it does not. The former case can only occur fewer than $2c$ times. This follows from the fact that if there were $2c$ or more such paths $p_i$, then these paths must contain tiles outside of any valid macrotile region representing a tile as well as outside of any fuzz regions of such macrotiles. In the latter case,
there must be two tiles contained in $M^r_i$ and $M^l_i$, which we denote by $tl$ and $tr$ respectively, such that these tiles are connected by a path of tiles in $\alpha'_{2k}$ and this path does not contain a tile at a location below any tile of $\sigma_{\mathcal{T}}$. We will use the locations of $tl'_i$ and $tr'_i$ to define a cut of the grid graph.

Let $(x_r, y_r)$ be the location of $tr$ and let $(x_l, y_l)$ be the location of $tl$. For the assembly $\alpha'_{2k}$ produced by the InitAssembly procedure, the cut, $f_k$ say, is defined by the edges in the grid graph that intersects lines
\begin{enumerate}
\item $y=y_l+\dfrac{1}{2}$ for $x$ in $(-\infty, x_l)$\label{item:cut-portion0}
\item $x=x_l + \dfrac{1}{2}$ for $y$ between $y_l$ and $y_r$\label{item:cut-portion1}
\item $y=y_r + \dfrac{1}{2}$ for $x$ in $(x_l, x_r)$\label{item:cut-portion2}
\item $y=y_r+\dfrac{1}{2}$ for $x$ in $(x_r, \infty)$\label{item:cut-portion4}
\end{enumerate}

At a high-level, we have chosen this cut so that as assembly proceeds, tiles must either cross the portion of the cut corresponding to Lines~\ref{item:cut-portion0} and~\ref{item:cut-portion4} or are prevented from growing lower than all of the tiles of the path of tiles from $tl$ to $tr$ divide the plane into two disjoint sets. This idea is depicted in Figure~\ref{fig:zig-zag-simulation}.

We then denote the glue sequence table associated to the cut $f_k$ by $GST_{f_k}$. Note that this cut extends infinitely to the left and right dividing the $\Z^2$ lattice into points ``above'' $f_k$ and points ``below'' $f_k$. Analogously, for $j\in\N$, $f_{jk}$ is defined as in $f_{k}$ so that $f_{jk}$ is a cut corresponding to lines that lie between the lines $y=ik$ and $y=(i+1)k$.

\begin{figure}[htp]
\begin{center}
\includegraphics[width=3.0in]{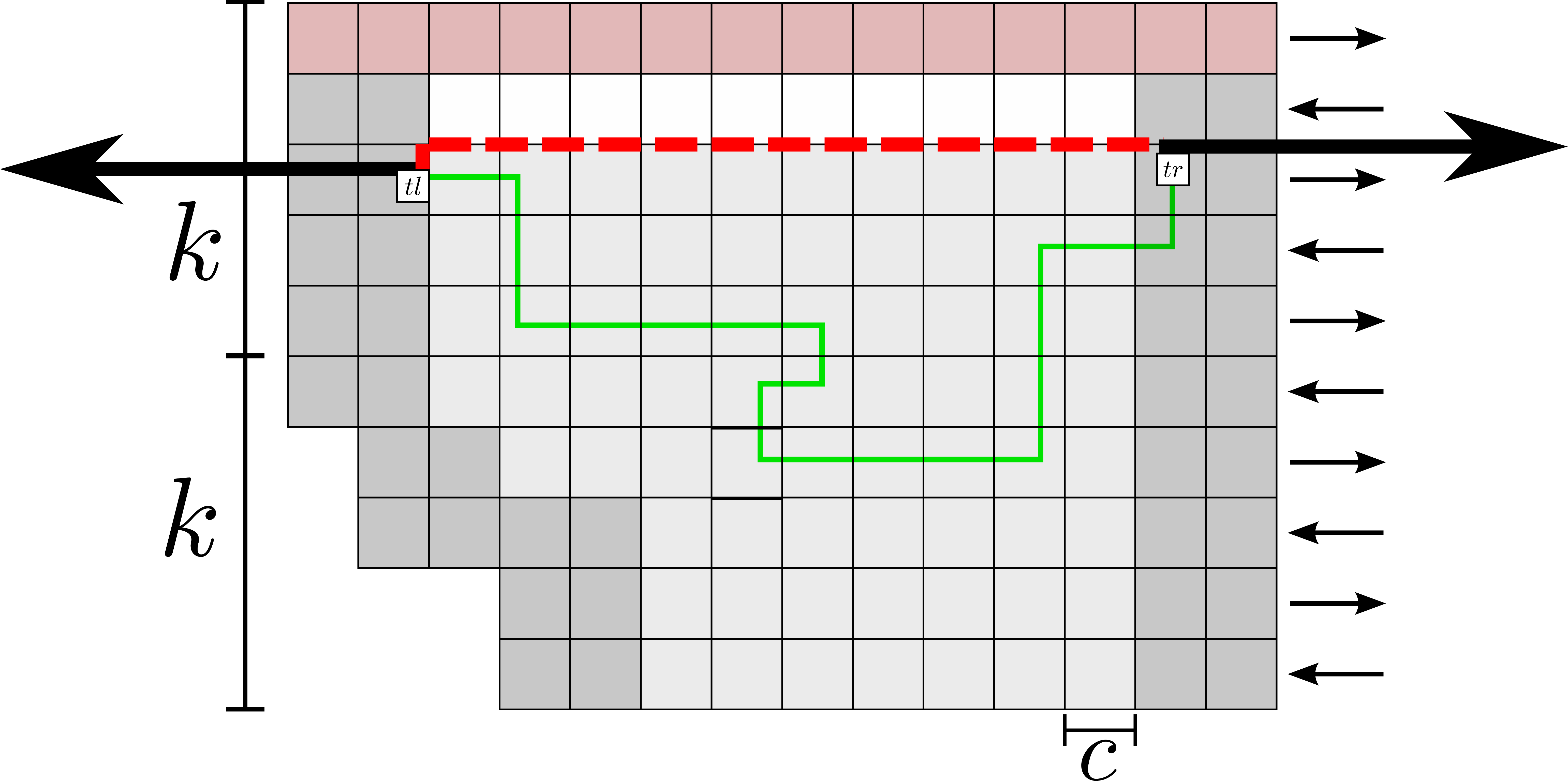}
\caption{The cut used to define a GST. Glues which cross the red dotted portion of the cut are not in the domain of the GST. The green path is a representation of a path of tiles from $tl$ to $tr$.}
\label{fig:zig-zag-simulation}
\end{center}
\end{figure}

Let $G$ be the set of glues associated to tiles in $S$ and let $D=\{(0,1), (0,-1),$ $(1,0), (-1,0)\}$ be a set corresponding to the north, south, east, and west edges of a tile respectively. In addition, let $\tilde{G}$ denote the product $G\times D\times \Z^2$ such that $(t, d, l)$ is in $\tilde{G}$ iff $l=(x,y)$ where $y-\dfrac{1}{2}$ is not a point on Line~\ref{item:cut-portion1} or~\ref{item:cut-portion2}. We say that $\tilde{g}\in \tilde{G}$ \emph{crosses the cut $f_k$} iff $\tilde{g} = ( g, d, l )$ for $g\in G$, $d\in D$, and $l\in \Z^2$ such that $l$ lies on one side of the cut $f_k$ and $l+d$ lies on the other side of the cut $f_k$. Moreover, if $l$ is above (respectively below) the cut, we say that $\tilde{g}$ \emph{crosses the cut from above to below} (respectively \emph{crosses the cut from below to above}). Also, we take the phrase \emph{$\tilde{g}$ is present} to mean that for $\tilde{g} = ( g, d, l )$ for $g\in G$, $d\in D$, and $l\in \Z^2$, there is a tile at location $l$ with glue $g$ exposed on its edge corresponding to $d$.

Then, $GST_{f_k}$ is a relation between the set of all sequences of length at most $4c$ over $\tilde{G}$ to the power set $\mathcal{P}\left(\tilde{G}\right)$. For a sequence of glues $\Sigma = \langle \tilde{g}_i \rangle_{0\leq i<4c, \tilde{g}_i\in \tilde{G}}$ and a set $F$ in $\mathcal{P}\left(\tilde{G}\right)$, a pair $\left(\Sigma, F\right)$ is in $GST_{f_k}$ if and only if each $\tilde{g}_i$ in $\Sigma$ crosses the cut $f_k$ from above to below and $F$ is the set of elements of $\tilde{G}$ obtained as follows. First, assume that the subassembly below the cut $f_k$ is terminal and that the set $F$ is empty. Now, consider assembly below the cut while assuming that $\tilde{g}_0$ is present. Note that it may not even be possible for a tile to be part of a valid assembly so that $\tilde{g_0}$ is present. With $\tilde{g_0}$ present, attach tiles at tile locations below the cut $f_k$. $\tilde{g}_0$ must be used to start tile attachment below the cut $f_k$. Continue attaching tiles until the newly assembled subassembly below the cut $f_k$ is terminal. Call the new configuration of tiles $\beta'_0$. Then add any glues crossing the cut $f_k$ from below to above to the set $F$. See Figure~\ref{fig:glue-map} for more detail. Now, assume that $\tilde{g}_1$ is present and attach tiles to $\beta'_0$ below the cut $f_k$ until the produced assembly below the cut $f_k$ is terminal. Add any new glues of this assembly that are crossing the cut $f_k$ from below to above to the set $F$. Continue this process for each $\tilde{g}_i$ in $\Sigma$ to construct the set $F$. This process is described in Algorithm~\ref{alg:initGST}. In this algorithm, we use the notation $\partial^{\tau'}(\alpha'\cup \tilde{g})$ to denote empty tile locations where, under the assumption that $\tilde{g}$ is present, exposed glues of $\alpha'$ and $\tilde{g}$ allow for a tile to be placed so that the sum of the glue strengths of glues of this tile that match exposed glues of $\alpha'$ and/or $\tilde{g}$ is greater than or equal to $\tau'$. It should be noted that in the following algorithm, $\alpha'$ does not necessarily denote a \emph{stable} assembly; it only denotes a configuration of tiles.

\begin{algorithm}
\caption{An algorithm describing the procedure InitGST.\label{alg:initGST}}
	\KwIn{$\alpha'_{2k}$}
	\KwOut{$GST_{f_k}$}
	$\alpha' \defeq \alpha'_{2k}$\;
	set $\Sigma^*$ to the set of all sequence over $\tilde{G}$\;
	set $B$ to the set of all $(i,j)\in \Z^2$ below $f_k$\;	
	\For{$\Sigma\in \Sigma^*$}{
	set $F$ to the empty set\;
	set $E$ to the empty set\;
	\For{$\tilde{g} \in \Sigma$}{
	assume $\tilde{g}$ is present\;
	\While{$B \cap \partial^{\tau'}(\alpha'\cup \tilde{g}) \neq \emptyset$}{
		choose $(i, j)\in \Z^2$ such that $(i,j) \in B \cap \partial^{\tau'}(\alpha'\cup \tilde{g})$\;
		$l\defeq (i,j)$\;
		
		choose $t\in T$ such that $l\in \partial^{\tau'}_t (\alpha'\cup \tilde{g})$\;
		$\alpha' \defeq \alpha' + (l \mapsto t)$\;
		$\tilde{b} \defeq (t, d, l)$\;
		\If{$\tilde{b}$ crosses the cut $f_k$ from below to above}{
			add $\tilde{b}$ to $E$\;
		}

		}
	}
	set $F$ to $F \cup E$\;
	add $(\Sigma, F)$ to $GST_{f_k}$\;
	}
\end{algorithm}

\begin{figure}[htp]
\begin{center}
\includegraphics[width=3.0in]{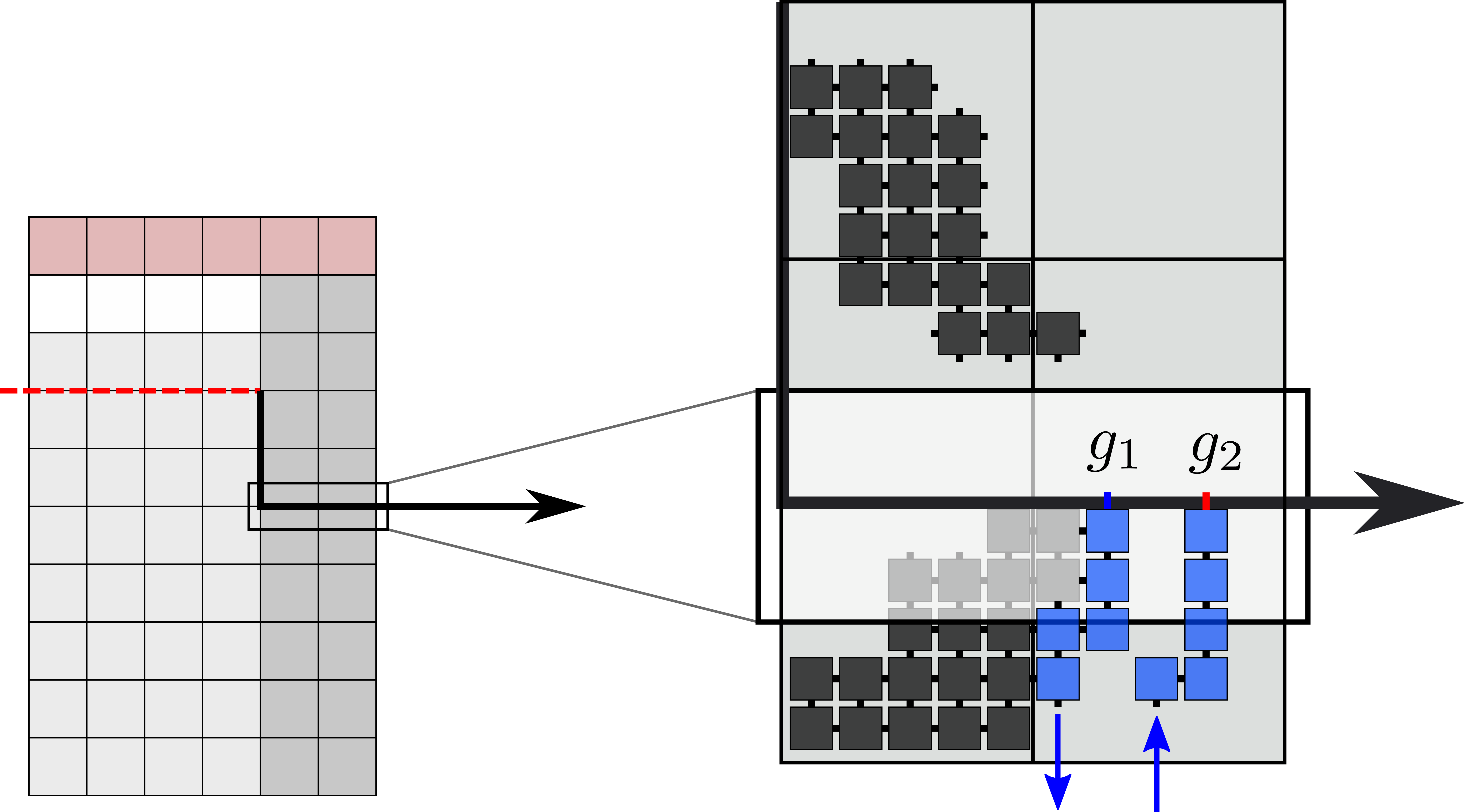}
\caption{Construction of a glue sequence table. The bold line along with the red dotted line depict the cut $f_k$. Dark grey tiles are part of the existing assembly. The subassembly below the cut $f_k$ is assumed to be terminal. If the blue glue $g_1$ is present, then assembly of the blue tiles may be possible resulting in glue $g_2$ crossing the cut $f_k$ from below to above. This glue is added to the set $F$.}
\label{fig:glue-map}
\end{center}
\end{figure}

\subsubsection*{UpdateAssembly}

Let $B_i$ be the set of all $(x,y)\in\Z^2$ such that $(i-2)k \leq y \leq ik$. Moreover, let $\beta'_{ik}$ be the subconfiguration of $\alpha'_{ik}$ contained in ${B_i}$. In other words, $\beta'_{ik}$ is the map $\alpha'_{ik}|_{B_i}$. The UpdateAssembly procedure consists of the steps given in Algorithm~\ref{alg:updateAssembly}. This algorithm computes $\beta'_{ik}$ using $\beta'_{(i-1)k}$ and $GST_{f_{(i-2)k}}$. The idea is to assemble the portion of $\beta'_{ik}$ by allowing tiles to bind to tiles of $\beta'_{(i-1)k}$ when appropriate glues are present for strength two binding. In addition, when a tile is placed so that a glue, $g$ say, on the tile crosses the cut $f$, rather than continue to attach tiles below the cut using this exposed glue, the table $GST_{f_{(i-2)k}}$ is used to lookup which glues will cross the $f$ from below to above as a result of $g$.

\begin{algorithm}
\caption{An algorithm describing the procedure UpdateAssembly.\label{alg:updateAssembly}}
	\KwIn{$\beta'_{(i-1)k}$ and $GST_{f_{(i-2)k}}$}
	\KwOut{$\beta'_{ik}$}
	$\alpha' \defeq \beta'_{(i-1)k}$\;
	set $\Sigma$ to the empty sequence over $\tilde{G}$\;
	set $B$ to the set of $(x,y)\in\Z^2$ such that $(i-2)k \leq y \leq ik$\;
	\While{$\partial^{\tau'} \alpha' \cap B \neq \emptyset$}{
		choose $(i, j)\in \Z^2$ such that $(i,j) \in \partial^{\tau'} \alpha' \cap B$\;
		$l\defeq (i,j)$\;
		
		choose $t\in T$ such that $l\in \partial^{\tau'}_t \alpha'$\;
		$\alpha' \defeq \alpha' + (l \mapsto t)$\;

		\For{$d\in\{(0,-1), (1,0), (-1,0)\}$}{
			$\tilde{g} \defeq (t, d, l)$\;
			
			\If{$\tilde{g}$ crosses the cut $f_{(i-2)k}$ from above to below}{
				add $\tilde{g}$ to $\Sigma$ as last element\;
				assume that any unexposed glues in $GST_{f_{(i-2)k}}(\Sigma)$ along the cut $f_{(i-2)k}$ are present\;	
			}		
		}
	}

	$\beta'_{ik} = \alpha'|_{B}$\;
\end{algorithm}

It still remains to be shown that this algorithm correctly yields $\beta'_{ik}$. To see this, we prove the following claim.

\begin{claim}\label{clm:shortDip}
For tiles $a$ and $b$ in $\alpha'_{ik}$ but not in $\alpha'_{(i-1)k}$ such that the location $(x_1,y_1)$ of $a$ is above the line $y=(i-1)k$ and $b$ is a tile belonging to a macrotile such that the location $(x_2,y_2)$ of $b$ is not contained in a macrotile belonging to the left or right side regions of $\alpha'_{ik}$, it must be the case that $|y_2 - y_1| < c^2 + 2c + 2$.
\end{claim}

Note that $a$ requires $b$. That is, $b$ must be placed prior to $a$ in any assembly sequence. By Lemma~\ref{lem:paths} there is a path of tiles from $a$ to $b$. Let $p_{ab}$ denote this path.
Tiles $a$ and $b$ as described in Claim~\ref{clm:shortDip} are depicted in Figure~\ref{fig:memory-use}. To prove the claim, consider the two tiles on the left and right ends of each row in the simulated zig-zag system. Call the tile farthest to the left $t_1$ and the other $t_2$. For the macrotile regions $M_1$ and $M_2$, where $M_1$ is to the left of $M_2$, of the simulating system that map to the tiles $t_1$ and $t_2$ respectively. We define $S_1$ to be the set of configurations over $M_1$ that map to $t_1$ under $R^*$. Similarly, we define $S_2$ to be the set of configurations over $M_2$ that map to $t_2$ under $R^*$. Note that for $\mathcal{S}$ to be a valid simulation of $\mathcal{T}$, either $S_2$ requires $S_1$ or $S_1$ requires $S_2$. Assume that $S_2$ requires $S_1$ (the other case is similar). Lemma~\ref{lem:paths} implies that there is a new path $p$ of tiles from $S_1$ to $S_2$. In Figure~\ref{fig:memory-use1} and~\ref{fig:memory-use2}, $p$ is depicted as the green line. Note that since $b$ is not a tile in $\alpha'_{(i-1)k}$, we can assume that the path $p$ from $S_1$ to $S_2$ does not intersect $p_{ab}$. Now we consider two cases. First, as the path $p$ assembles, a tile is placed below any tile of the seed $\sigma_\mathcal{T}$. Second, as the path $p$ assembles, a tile does not bind at a tile location that is below any tile location of the seed $\sigma_\mathcal{T}$. In the first case, the path must grow down the left or right side regions of $\alpha'_{ik}$, and by our choice of cut $f_{ik}$, this path must cross the cut. For a valid simulation, the maximum number of such paths that can assemble is $2c$, for otherwise the paths would exceed the fuzz region allowable in simulation. In the second case, it must be the case that the tiles belonging to the path $p$ place a tile in the middle region below $b$. In this case, the most such paths that can assemble is $c^2$. This follows from the fact that each such path places tiles either in the macrotile region containing the location of $b$ or below the macrotile region containing $b$. $c^2$ such paths would prevent any tiles from being placed in the macrotile region containing $b$, however, this would contradict the fact that $\mathcal{S}$ is a valid simulation. Therefore, Claim~\ref{clm:shortDip} holds.

\begin{figure}[htp]
\centering
  \subfloat[][In this case, the tile $b$ is contained in a macrotile region that is not a subregion of the left or right side regions shown as gray blocks.]{%
        \label{fig:memory-use1}%
	        \includegraphics[width=2.75in]{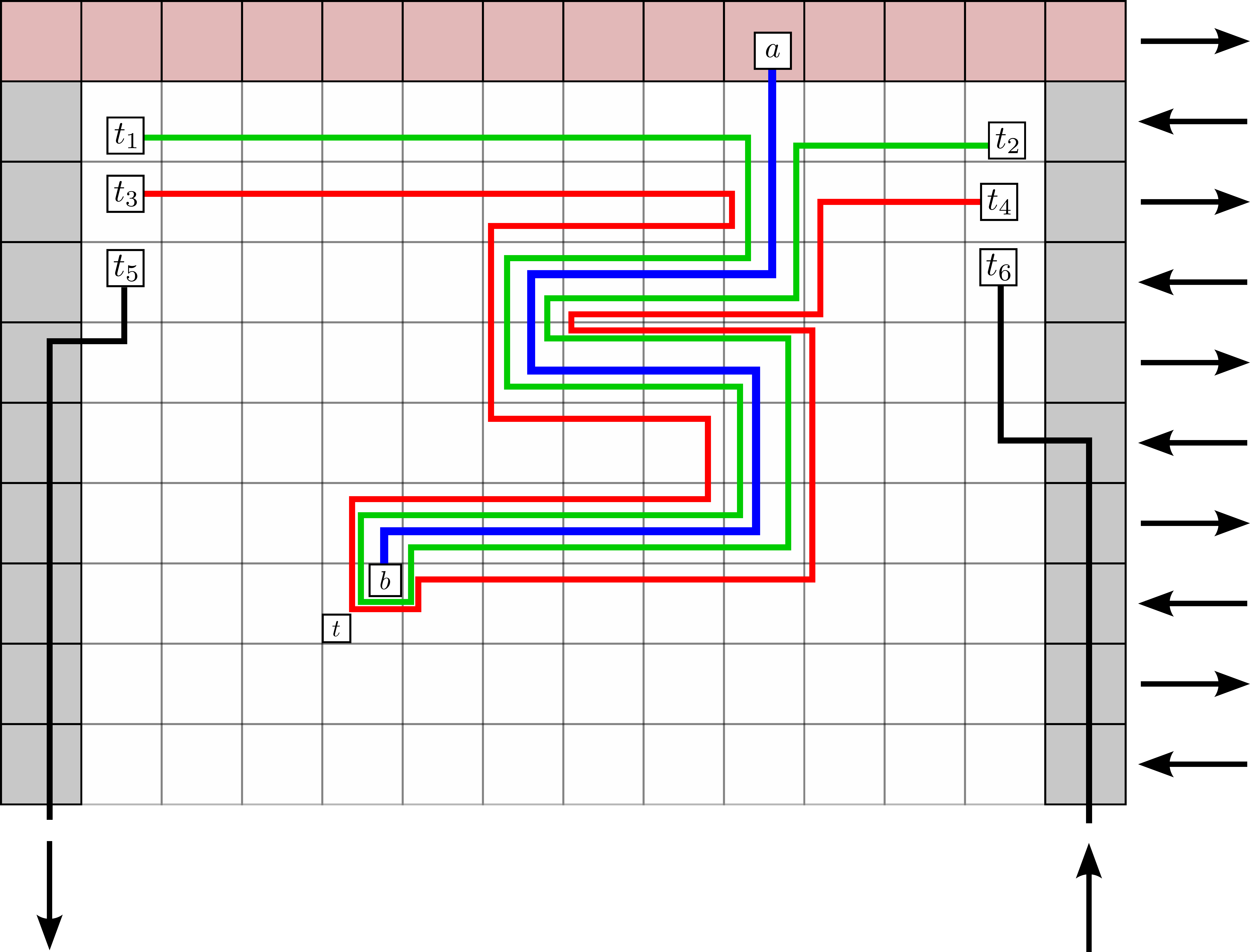}
        }%
        \quad\quad
  \subfloat[][In this case, the tile $b$ is contained in a macrotile region that is a subregion of the left or right side regions.]{%
        \label{fig:memory-use2}%
        		\includegraphics[width=2.75in]{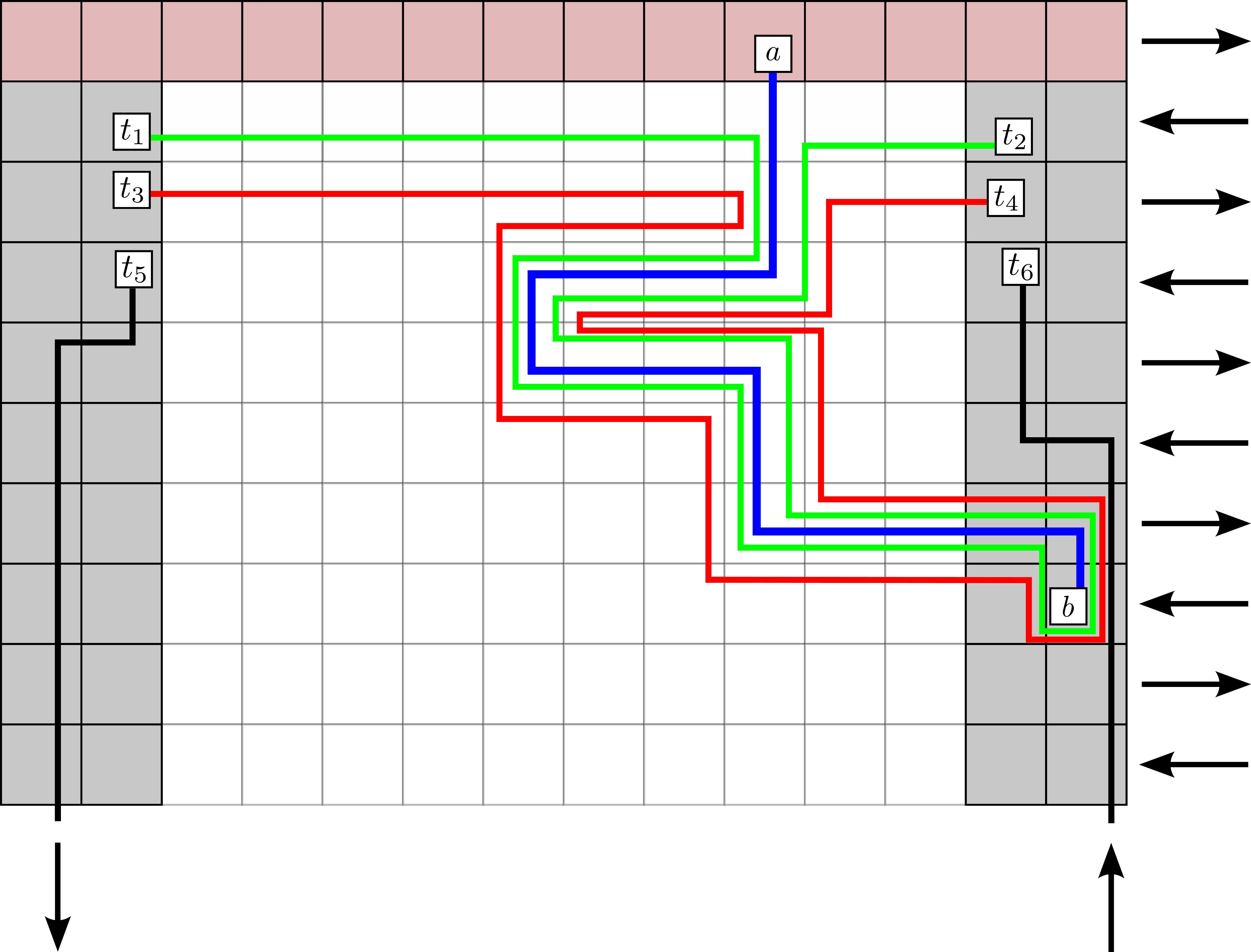}
        }%
  \caption{The blue, green, and red paths represent non-overlapping paths tiles of tiles that must assemble by Lemma~\ref{lem:paths}.\vspace{-5pt}}
  \label{fig:memory-use}
\end{figure}

By Claim~\ref{clm:shortDip}, we see that since we have chosen $k$ to be larger than $c^2+2c+2$, UpdateAssembly can be used to determine the tiles of $\beta'_{ik}$ from $\beta'_{(i-1)k}$ and $GST_{f_{(i-2)k}}$.

\subsubsection*{UpdateGST}

For each $i < n/k$, the procedure UpdateGST is used to create a new glue sequence table relative to the cut $f_{ik}$. The UpdateGST procedure is similar to the InitGST procedure, except that $GST_{f_{ik}}$ is constructed with the help of $GST_{f_{(i-1)k}}$. This procedure is given as Algorithm~\ref{alg:updateGST}.

\begin{algorithm}
\caption{An algorithm describing the procedure UpdateGST.\label{alg:updateGST}}
	\KwIn{$\beta'_{ik}$ and $GST_{f_{(i-2)k}}$}
	\KwOut{$GST_{f_{(i-1)k}}$}
	$\alpha' \defeq \beta'_{ik}$\;
	set $\Sigma^*$ to the set of all sequence over $\tilde{G}$\;
	set $B$ to the set of all $(x,y)\in \Z^2$ below $f_{ik}$\;	
	\For{$\Sigma\in \Sigma^*$}{
	set $F$ to the empty set\;
	set $M$ to the empty set\;
	\For{$\tilde{g} \in \Sigma$}{
	set $\Pi$ to the empty set;
	assume that $\tilde{g}$ is present\;
	\While{$B \cap \partial^{\tau'}(\alpha'\cup \tilde{g}) \neq \emptyset$}{
		choose $(x, y)\in \Z^2$ such that $(x,y) \in B \cap \partial^{\tau'}(\alpha'\cup \tilde{g})$\;
		$l\defeq (x,y)$\;
		
		choose $t\in T$ such that $l\in \partial^{\tau'}_t (\alpha'\cup \tilde{g})$\;
		$\alpha' \defeq \alpha' + (l \mapsto t)$\;		
		$\tilde{b} \defeq (t, d, l)$\;
		\If{$\tilde{b}$ crosses the cut $f_{(i-1)k}$ from below to above}{
			add $\tilde{b}$ to $M$\;
		}

		\For{$d\in\{(0,-1), (1,0), (-1,0)\}$}{
			$\tilde{d} \defeq (t, d, l)$\;
			
			\If{$\tilde{d}$ crosses the cut $f_{(i-2)k}$ from above to below}{
				add $\tilde{d}$ to $\Pi$ as last element\;
				expose any unexposed glues in $GST_{f_{(i-2)k}}(\Pi)$ along the cut $f_{(i-2)k}$\;	
			}		
		}

		}
	}
	set $F$ to $F \cup M$\;
	add $(\Sigma, F)$ to $GST_{f_{(i-1)k}}$\;
	}
\end{algorithm}

\subsubsection{The space complexity of the computation of $r'(n)$ is $O\left(f'(n)\right)$}

In this section, we prove two claims. First, that Algorithm~\ref{alg:simAlgo} correctly computes $r'(n)$, and second, that Algorithm~\ref{alg:simAlgo} can be computed in space $O(f'(n))$. We first argue that Algorithm~\ref{alg:simAlgo} can be computed in space $O(f'(n))$. Algorithm~\ref{alg:simAlgo} consists of four procedures: InitAssembly, InitGST, UpdateAssembly, and UpdateGST. First, we make the following observation.

\begin{observation}
$|\{(x,y)\mid GST_{f_k}(x) = y \}|$ is bounded by a constant that only depends on $c$ and $|U|$. We denote this constant by $K_{GST}$. In particular, if $g$ is the number of glues of tiles of $U$, $(g+1)^{4c}(4c)!$ is such a constant as this is the total number of sequences of $g$ glues (plus the null glue) of length $4c$.
\end{observation}

From Algorithm~\ref{alg:primeAssembly} and~\ref{alg:initGST}, it is clear that InitAssembly and InitGST each require $O(f'(n) + K_{GST})$ space. Moreover, since for all $i$, $\beta'_{ik}$ is bounded in width by $f'(n)$ and in height by $2k$, UpdateAssembly and UpdateGST each require $O(f'(n))$ space, with each procedure requiring at most $f'(n)*3k$ tile locations.

Now let $B_i$ denote the set of $(x,y)\in\Z^2$ such that $(i-2)k \leq y \leq ik$. It remains to be shown that Algorithm~\ref{alg:simAlgo} correctly computes $r'(n)$.
To see this, note that UpdateAssembly computes $\beta'_{ik} = \alpha'_{ik}|_B$ and $L_{n}\subseteq B$. It follows that, $\beta'_{ik}|_{L_{n}} = \alpha'_{ik}|_{L_{n}}$. Therefore, Algorithm~\ref{alg:simAlgo} correctly computes $r'(n)$.

\end{proof}

Note that in Lemma~\ref{lem:noCheating1}, the assumption that $\mathcal{S}$ is directed can be removed by defining the glue sequence table to map into the set of sets of glues corresponding to each possible set of glues that may cross the cut corresponding to this glue sequence table. This set of sets is still bounded by a constant depending on $c$ and $S$. Fix an enumeration of this set of sets. Then, we modify the procedures UpdateAssembly and UpdateGST so that if a glue crosses the cut from above to below, we expose glues corresponding to the first set of glues in the enumeration of the set of sets of glues. Now we can state Lemma~\ref{lem:noCheating} which we refer to as the ``no cheating lemma''.

\begin{lemma}[No Cheating Lemma]\label{lem:noCheating}
Let $\mathcal{T} = (T, \sigma, \tau)$ be a zig-zag system and let $\mathcal{S} = (S, \sigma_S, {\tau'})$ be a system that simulates $\mathcal{T}$ at temperature $\tau'$ with scale factor $c$. Let $n$ be in $\N$, and let $f(n)$ be the width of the longest row of the assembly of $\mathcal{T}$ consisting of $n$ completed rows. Moreover, let $C_{cn}$ be a set of finite configurations, let $r'$ be the characteristic function for $\mathcal{S}$ given $C_{cn}$, and let $L_{r'}$ be the language defined by $\mathcal{S}$ given $r'$. Then, $L_{r'}\in \DSPACE(f(n))$.

\end{lemma}

\begin{proof}
For the scale factor $c$, this follows from the fact that $f'(n) \leq cf(n) + 2c$. The addition of $2c$ accounts for fuzz regions.
\end{proof}

\fi

\section*{Acknowledgements}
The authors would like to thank Jack Lutz for helpful guidance while searching for much needed computational complexity results.

\bibliographystyle{plain}
\bibliography{tam,complexity,experimental_refs}

\end{document}